\documentclass{llncs}

\newcommand{\norm}[1]{}
\usepackage{amsmath} 

\usepackage[utf8]{inputenc}
\usepackage[mathscr]{eucal}
\usepackage{graphicx}
\usepackage{listings}
\usepackage{amssymb}
\usepackage{dsfont}
\usepackage{mathtools}
\usepackage{algorithm}
\usepackage{algpseudocode}
\usepackage{xcolor}
\usepackage{multirow}
\usepackage[english]{babel}
\usepackage{wasysym}
\usepackage{qip}
\usepackage{qm}
\usepackage{csquotes}
\usepackage{environ}

\usepackage[draft=false]{hyperref}
\usepackage[normalem]{ulem}
\usepackage[capitalise]{cleveref}
\crefname{Claim}{Claim}{Claims}
\crefname{construct}{Construction}{Constructions}

\usepackage{comment}
\usepackage{paralist}
\usepackage{wrapfig}
\usepackage[normalem]{ulem}

\usepackage[clock]{ifsym}

\usepackage[
	n,
	operators,
	advantage, 
	sets,
	adversary,
	landau,
	probability,
	notions,
	logic,
	ff,
	mm,
        primitives,
        events,
        complexity,
        asymptotics, 
        keys]{cryptocode}

\usepackage[maxbibnames=10]{biblatex}

\spnewtheorem{construct}{Construction}{\bfseries}{\itshape}
\spnewtheorem*{thm}{Theorem}{\bfseries}{\itshape}
\spnewtheorem{Claim}{Claim}{\bfseries}{\itshape}
\spnewtheorem*{lem}{Lemma}{\bfseries}{\itshape}
\usepackage{xcolor}

\definecolor{darkblue}{rgb}{0,0,0.6}
\definecolor{darkgreen}{rgb}{0,0.5,0}

\newcommand{\hil}[1]{\ensuremath{\mathcal{#1}}}

\newcommand{\isbot}{\mathsf{Is}\text{-}\bot}

\usepackage{xspace}
\newcommand{\PRP}{\mathsf{PRP}}

\newcommand{\PRS}{\mathsf{PRS}}
\newcommand{\PRG}{\mathsf{PRG}}
\newcommand{\PRF}{\mathsf{PRF}}

\newcommand{\OWSG}{\mathsf{OWSG}}
\newcommand{\PRU}{\mathsf{PRU}}
\newcommand{\OWF}{\mathsf{OWF}}

\newcommand{\SPRS}{\mathsf{SPRS}}
\newcommand{\LPRS}{\mathsf{LPRS}}

\newcommand{\botPRG}{\bot\text{-}\mathsf{PRG}}

\newcommand{\QKPRP}{\mathsf{PRP}^{\text{qs}}}
\newcommand{\QKPRU}{\mathsf{PRU}^{\text{qs}}}

\newcommand{\QKPRS}{\mathsf{PRS}^{\text{qs}}}
\newcommand{\QKPRG}{\mathsf{PRG}^{\text{qs}}}
\newcommand{\QKPRF}{\mathsf{PRF}^{\text{qs}}}

\newcommand{\BQSPRS}{\mathsf{BC}\text{-}\mathsf{SPRS}^{\text{qs}}}
\newcommand{\BQLPRS}{\mathsf{BC}\text{-}\mathsf{LPRS}^{\text{qs}}}
\newcommand{\BQPRS}{\mathsf{BC}\text{-}\mathsf{PRS}^{\text{qs}}}
\newcommand{\BQPRU}{\mathsf{BQ}\text{-}\mathsf{PRU}^{\text{qs}}}
\newcommand{\BQPRP}{\mathsf{BQ}\text{-}\mathsf{PRP}^{\text{qs}}}

\newcommand{\QKSPRS}{\mathsf{SPRS}^{\text{qs}}}
\newcommand{\QKLPRS}{\mathsf{LPRS}^{\text{qs}}}

\newcommand{\botPRF}{\bot\text{-}\mathsf{PRF}}

\newcommand{\Hy}{\ensuremath{\textsf{H}}}

\let\oldmaketitle\maketitle
\renewcommand{\maketitle}{\oldmaketitle\setcounter{footnote}{0}}
\begin{document}
\title{MicroCrypt Assumptions with Quantum Input Sampling and Pseudodeterminism: Constructions and Separations}

\author{Mohammed Barhoush\inst{1}\thanks{Part of this work was done while visiting NTT Social Informatics Laboratories as an internship.} \and Ryo Nishimaki\inst{2} \and Takashi Yamakawa\inst{2}}
\institute{Universit\'e de Montr\'eal (DIRO), Montr\'eal, Canada\\  \email{mohammed.barhoush@umontreal.ca} \and
NTT Social Informatics Laboratories\\
\email{ryo.nishimaki@ntt.com}, \email{takashi.yamakawa@ntt.com}}%

\maketitle
\begin{abstract}


We investigate two natural relaxations of quantum cryptographic primitives. The first involves quantum input sampling, where inputs are generated by a quantum algorithm rather than sampled uniformly at random. Applying this to pseudorandom generators ($\PRG$s) and pseudorandom states ($\PRS$s), leads to the notions denoted as $\QKPRG$ and $\QKPRS$, respectively. The second relaxation, $\bot$-pseudodeterminism, relaxes the determinism requirement by allowing the output to be a special symbol $\bot$ on an inverse-polynomial fraction of inputs.

We demonstrate an equivalence between bounded-query logarithmic-size $\QKPRS$,  logarithmic-size $\QKPRS$, and $\QKPRG$. Moreover, we establish that $\QKPRG$ can be constructed from $\botPRG$s, which in turn were built from logarithmic-size $\PRS$. Interestingly, these relations remain unknown in the uniform key setting.

To further justify these relaxed models, we present black-box separations. Our results suggest that $\bot$-pseudodeterministic primitives may be weaker than their deterministic counterparts, and that primitives based on quantum input sampling may be inherently weaker than those using uniform sampling.

Together, these results provide numerous new insights into the structure and hierarchy of primitives within MicroCrypt.



  \end{abstract}
   
\setcounter{secnumdepth}{3}
 \setcounter{tocdepth}{3}
\newpage
    \tableofcontents 
    \newpage
   
\section{Introduction}

The search for the minimal assumptions required for quantum cryptography was triggered with the astonishing discovery that pseudorandom states ($\PRS$s) \cite{JLS18} may exist even when quantumly-evaluable \footnote{In this work, all primitives including $\OWF$s and $\PRG$, refer to the quantum-evaluable versions, unless stated otherwise.} one-way functions $(\textsf{OWF})$ do not, relative to an oracle \footnote{Note that one-way functions and pseudorandom generators are equivalent.} \cite{K21}. $\PRS$s serve as the quantum analog to $\PRG$s, outputting a state instead of a classical string. Critically, this difference does not prevent $\PRS$s supporting some applications similar to those enabled by $\PRG$s, such as commitments, one-time signatures, and one-way state generators $(\OWSG)$s \cite{MY22a,AQY22}. 

This separation naturally raised questions on the minimal assumptions required to build quantum cryptographic primitives. Addressing this question has fueled significant research, leading to a variety of quantum assumptions. Different assumptions provide a different balance between how well they replicate \textsf{OWF}s in cryptography and how strong of an assumption they constitute. The resulting assumptions are intricately related in what has now came to be known as \emph{MicroCrypt}. This field comprises various assumptions derived from \textsf{OWF}s, but where the other direction is not known. Despite substantial progress, numerous questions remain unanswered. What is clear, however, is that MicroCrypt is significantly more intricate than its classical counterpart.

Several of the MicroCrypt assumptions introduced parallel their classical counterparts but incorporate quantum elements. For instance, quantum unpredictable state generators \cite{MYY24} and one-way state generators \cite{MY22a} yield quantum outputs, similar to $\PRS$s. Additionally, assumptions such as $\botPRG$ \cite{BBO+24} and one-way puzzles \cite{KT24} involve only classical communication but rely on quantum computation. These quantum elements are believed to make the assumptions weaker. 

Despite advances, using general $\PRS$s as a complete replacement for $\PRG$s in cryptographic applications has been challenging. Some progress has been made in the specific case of logarithmic-size pseudorandom states, which we denote by \emph{short $\PRS$ ($\SPRS$)} as in \cite{ALY23}, where tomography can transform the state into a classical pseudorandom string \cite{ALY23}. However, tomography is not deterministic, resulting in what has been termed \emph{pseudodeterministic $\PRG$}. This roughly means that on $1-1/\poly$ fraction of inputs, the output is the same except with negligible probability  \footnote{``Pseudodeterminism'' is sometimes defined differently in other works \cite{BBV24,B21,BKN+23}. We follow the definition given in \cite{ALY23}.}. While these generators proved useful in various applications, the pseudodeterminism is sometimes problematic when using them in place of traditional $\PRG$s. This obstacle motivated a follow-up work \cite{BBO+24}, that introduced an intermediate notion called $\botPRG$, which is built from pseudodeterministic $\PRG$s. With $\botPRG$s, the non-deterministic outcomes can be detected and replaced with $\bot$, allowing many $\PRG$ applications to proceed by handling $\bot$ cases separately. This approach enabled significant applications, such as many-time digital signatures and quantum public-key encryption with tamper-resilient keys, which had eluded MicroCrypt.

While these applications are powerful, $\botPRG$s and $\SPRS$s have not been black-box separated from $\OWF$s \footnote{Notably, the separation between $\PRS$ and $\OWF$s \cite{K21} only applies to linear-sized $\PRS$s and not to $\SPRS$.}, which somewhat limits the significance of these results. In fact, most MicroCrypt assumptions, such as pseudorandom function-like states with proofs of destruction \cite{BBS23} and efficiently verifiable one-way puzzles (\textsf{Ev-OWPuzz}s) \cite{KT24,CGG24}, have been conjectured to be weaker than (quantum-evaluable) $\OWF$s, but their separability has not been established. Understanding which assumptions are separated from $\OWF$s and, more generally, the relations among different MicroCrypt primitives is an important goal in the field. 

\subsection{Our Work} 
Traditionally, many cryptographic primitives such as $\PRG$s and $\PRS$s rely on inputs sampled uniformly at random. The main idea of our work is that sampling inputs with a quantum procedure, instead of at random, yields fundamentally different primitives. We denote the resulting primitives with a superscript such as $\QKPRG$s and $\QKPRS$s. 

To clarify, a pair of QPT algorithms $(\textsf{QSamp},G)$ is a $\QKPRG$ if $\textsf{QSamp}(1^\lambda)$ samples a $\lambda$-bit input, which is then mapped by $G$ to an output of length $\ell > \lambda$, and the following security condition is satisfied: For any QPT adversary $\adv$, 
\begin{align*}\left| \Pr_{k\leftarrow \textsf{QSamp}(1^\lambda)}\left[\adv(G(k))=1\right]-\Pr_{y\leftarrow \{0,1\}^\ell}\left[\adv(y)=1\right]\right| \leq \negl[\lambda].
\end{align*}
Naturally, we also require that $G$ is \emph{almost-deterministic}, meaning that it returns the same output on a fixed input except with negligible probability. Notice that this notion differs from a traditional $\PRG$, where the above security condition is guaranteed if the input $k$ is sampled uniformly at random. Similarly, a $\QKPRS$ consists of a pair of algorithms $(\textsf{QSamp}',G')$ such that for an input $k\leftarrow \textsf{QSamp}'(1^\lambda)$, polynomial copies of $\ket{\psi_k}\leftarrow G'(k)$ are indistinguishable from polynomial copies of a Haar random state.

While variants of MicroCrypt primitives based on quantum input sampling have been considered before for notions such as $\OWSG$s and $\PRS$s \cite{MY22b,KT24,BS20}, the fundamental distinction between primitives based on quantum versus uniform input sampling has not been previously recognized.

Note that any classical input sampling algorithm can be derandomized and replaced with uniform key sampling. Therefore, using a classical input sampling algorithm is unnecessary. However, a quantum procedure cannot be derandomized and might include non-deterministic quantum computations. Postponing such computations to the evaluation or state-generation phase can result in different outcomes across executions, which is problematic for deterministic primitives such as $\textsf{PRS}$s and $\textsf{PRG}$s.

Before moving to a detailed explanation, we first list a brief high-level description of the main contributions of this paper. For the sake of simplicity, this summary is somewhat weaker than what we actually accomplish. We show the following:

\begin{itemize}
    \item The pseudodeterminism error in $\botPRG$s can be eliminated if quantum input sampling is allowed. In particular, $\botPRG$ imply $\QKPRG$. This stands in contrast to the fact that $\botPRG$ are not known to imply $\PRG$.
    \item Primitives with quantum input sampling behave differently from their uniform input sampling counterparts. In particular, we realize that $\QKPRG$, bounded-copy $\QKSPRS$, and $\QKSPRS$ are all equivalent {under a certain parameter regime}. Such an equivalence is not known to exist in the uniform input sampling setting.
    \item While $\PRG$s trivially imply $\QKPRG$s, we show that the reverse implication is unlikely by showing that there is no black-box construction of $\PRG$s from a $\QKPRF$s, even with unitary and inverse access to the $\QKPRF$. In addition, we provide more fine-grained black-box separations: separating $\PRG$ from $\botPRG$, and $\botPRG$ from $\QKPRG$, albeit within a weaker oracle access model.
    Thus, we establish a hierarchy among uniform input sampling, pseudodeterministic, and quantum input sampling primitives within MicroCrypt.  
\end{itemize}


\subsubsection{Quantum Input Sampling.}

In the first part of this work, we introduce natural variants of MicroCrypt primitives that incorporate quantum input sampling, and demonstrate how this framework helps address the issue of pseudodeterminism. 

Recall that \cite{BBO+24} extended the applicability of $\SPRS$ by converting them into $\botPRG$s, which inherit many of the useful properties of $\PRG$s. However, $\botPRG$s still exhibit a form of non-determinism due to the possibility of outputting $\bot$, which may be problematic in certain applications. Specifically, for a $\botPRG$, there exists a set of ``good'' inputs, that produce deterministic outputs, and an inverse-polynomial fraction of inputs termed  ``bad'' inputs, that may yield $\bot$. To address this non-determinism, a natural solution is to test inputs during the sampling process to ensure that only good inputs are selected. We show that this technique can be used to construct a $\QKPRG$ from a $\botPRG$, thus resolving the pseudodeterminism issue. Note that this approach necessitates a quantum  input sampling procedure instead of traditional uniform input sampling, since the $\botPRG$ itself may be a quantum algorithm.  

We utilize this result, along with other key insights, to establish the following relationships among primitives with quantum input sampling. Specifically, we show fully black-box constructions for the following: 

\begin{enumerate}
    \item $\QKPRG$ from bounded-copy $\QKSPRS$ ($\BQSPRS$). 
    \item $\QKSPRS$ from $\QKPRG$. 
    \item $\BQPRU$ \footnote{This stands for bounded-query pseudorandom unitaries with quantum key sampling (see \cref{def:PRU-qs}).} from $\QKPRG$.
    \item $\QKPRU$ from $\QKPRF$.
\end{enumerate}

These findings mean that $\QKPRG$, $\BQSPRS$, and $\QKSPRS$ can all be built from one another under a certain parameter regime. This relationship is surprising, as it is not known to hold in the uniform input sampling setting. Specifically, it is not known how to construct a $\PRG$ from a $\SPRS$, nor how to build a $\SPRS$ from a bounded-copy $\SPRS$.

Furthermore, as a direct consequence of these results, we obtain both a method to reduce the output length of a $\QKSPRS$ and a way to transform a $\SPRS$ into a $\QKSPRS$ with a longer output length.


\subsubsection{Separations.} 

In the second part of our work, we extend our analysis with separation results that highlight the distinctions between quantum input sampling and uniform input sampling and between $\bot$-pseudodeterminism and determinism. 

Our separation results demonstrate the impossibilities of certain types of black-box constructions. There are different variants of black-box constructions in quantum cryptography (see \cite{CM24} for an exposition). We informally define the two variants considered in this work. 

\begin{definition}[Informal version of \cref{def:BB with  access to inverse}]
\label{inf def BB acccess to inverse}
A QPT algorithm $G^{(\cdot)}$ is a \emph{fully black-box construction of a primitive $Q$ from a primitive $P$ with inverse access} if there is a QPT algorithm $S^{(\cdot)}$ such that for every unitary implementation $U$ of $P$:
\begin{itemize}
    \item $G^{U,U^{-1}}$ is an implementation of $Q$.
    \item Every attack $\adv$ that breaks the security of $G^{U,U^{-1}}$, and every unitary implementation $\tilde{\adv}$ of $\adv$, it holds that $S^{\tilde{\adv},\tilde{\adv}^{-1}}$ breaks the security of $U$.
\end{itemize}
\end{definition}

We also consider a less general notion limited to \emph{completely-positive-trace-preserving (CPTP) maps}. These are quantum channels that are not necessarily unitary. 

\begin{definition}[Informal version of \cref{def:BB with CPTP access}]
\label{inf def BB with CPTP access}
A QPT algorithm $G^{(\cdot)}$ is a \emph{fully black-box construction of $Q$ from CPTP access to $P$} if there is a QPT algorithm $S^{(\cdot)}$ such that for every CPTP implementation $\mathcal{C}$ of $P$:
\begin{itemize}
    \item $G^{\mathcal{C}}$ is an implementation of $Q$.
    \item Every attack $\adv$ that breaks the security of $G^{\mathcal{C}}$, it holds that $S^{{\adv}}$ breaks the security of $U$.
\end{itemize}
\end{definition}

Since CPTP maps are not necessarily unitary and cannot be purified or inverted, \cref{inf def BB with CPTP access} only includes constructions that do not assume purified/inverse/unitary access. On the other hand, \cref{inf def BB acccess to inverse} includes constructions that use such access, thus covering a broader range of constructions. We discuss these distinctions more thoroughly in \cref{sec:discussion}, but for now, we state the results.

Our first and main separation is between $\PRG$s and $\QKPRF$s, with inverse access.

\begin{theorem}[Informal version of \cref{thm:separation 3}]
\label{informal thm 3}
There does not exist a fully black-box construction of a $\PRG$ from a (quantum-query-secure) $\QKPRF$s with inverse access. 
\end{theorem}

Given that $\QKPRF$s inherits many applications of $\PRF$s, we obtain other new separations as a corollary. 

\begin{corollary}
\label{informal cor 3}
   There is no fully black-box construction of a $\PRG$ from the following primitives with inverse access:
   \begin{enumerate}
    \item $\QKPRG$, $\QKSPRS$, linear-sized $\QKPRS$, and $\QKPRU$. 
 \item Statistically-binding, computationally hiding bit commitments with classical communication (\textsf{BC-CC}). 
\item Existentially unforgeable message authentication codes of classical messages with classical communication (\textsf{EUF-MAC}). 
        \item CCA2-secure symmetric encryption with classical keys and ciphertexts (\textsf{CCA2-SKE}). 
        \item \textsf{EV-OWPuzz}s. 
    \end{enumerate}
\end{corollary}


Our second separation is between $\bot$-pseudodeterministic notions and deterministic ones, but with CPTP access. 

\begin{theorem}[Informal version of \cref{lem:seperation}]
\label{informal thm 1}
    There does not exist a fully black-box construction of a $\OWSG$ from CPTP access to a $\botPRG$. 
\end{theorem}

Note that the $\OWSG$s considered in this paper are those with pure-state outputs and with uniform key generation. Our separation is further emphasized by the fact that a $\OWSG$ is considered weaker than a $\PRG$, since $\PRS$s imply $\OWSG$s and $\PRG$s are separated from $\PRS$s \cite{K21}. As $\botPRG$s have broad applicability \cite{BBO+24}, this result yields additional separations as corollaries.

\begin{corollary}[Informal version of \cref{cor:separation}]
\label{informal cor 1}
    There does not exist a fully black-box construction of a $\OWSG$s from CPTP access to:
    \begin{enumerate}
\item $\botPRF$s.
    \item (Many-time) existentially unforgeable digital signatures of classical messages with classical keys and signatures (\textsf{EUF-DS}). 
    \item CPA-secure quantum public-key encryption of classical messages with tamper-resilient keys and classical ciphertexts ($\textsf{CPA-QPKE}$). 
\end{enumerate}
\end{corollary}

Our third separation shows that CPTP access to $\QKPRF$s is insufficient for constructing $\botPRG$s. Recall that our second separation establishes a gap between $\PRG$s and $\botPRG$s. Taken together, this means the third separation strengthens the first by demonstrating a separation between $\QKPRF$ and $\botPRG$. However, this result is more limited in scope, as it only applies to CPTP access.

\begin{theorem}[Informal version of \cref{lem:no botowsg}]
\label{informal thm 2}
There does not exist a fully black-box construction of a $\botPRG$ from CPTP access to a (quantum-query-secure) $\QKPRF$s. 
\end{theorem}

We obtain other new separations as a corollary. 

\begin{corollary}[Informal version of \cref{cor:separation2}]
\label{informal cor 2}
    There does not exist fully black-box constructions of $\botPRG$s or $\SPRS$s from CPTP access to:
   \begin{enumerate}
    \item $\QKPRG$, $\QKSPRS$, linear-sized $\QKPRS$, and $\QKPRU$. 
 \item \textsf{BC-CC}. 
\item \textsf{EUF-MAC} and \textsf{CCA2-SKE}. 
        \item \textsf{EV-OWPuzz}s. 
    \end{enumerate}
\end{corollary}

Note that separations listed in \cref{informal cor 3,informal cor 1,informal cor 2} were not known prior to this work. This highlights how $\botPRG$s and $\QKPRG$ not only aid in building applications for $\SPRS$, but also in establishing separations among well-studied MicroCrypt assumptions that may be difficult to separate otherwise. For instance, $\textsf{EV-OWPuzz}$s have been studied and introduced as a potentially weaker replacement to (quantum-evaluable) $\OWF$s, but no separation existed prior to our work. 

Our results give a natural hierarchy in MicroCrypt as depicted in \cref{fig}.

\begin{figure}[ht]
\includegraphics[scale=0.33]{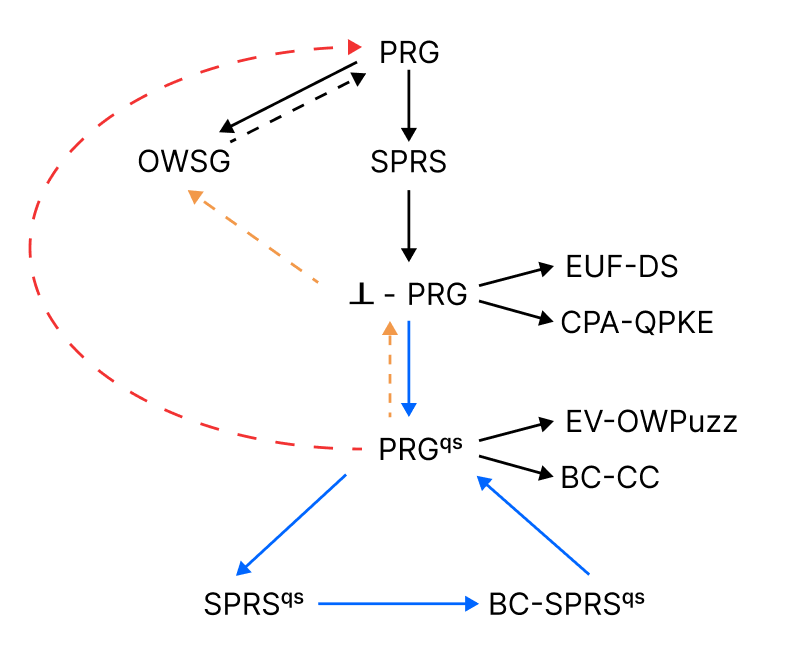}
\centering
\caption{The black straight arrows indicate implications that are trivial or from previous works \cite{MY22a,ALY23,BBO+24}. The black dotted arrow indicates a separation from previous work \cite{K21}. 
The \color{blue} blue \color{black} straight arrows are implications from this work. The \color{red} red \color{black} dotted arrow is a separation under inverse access from this work. The \color{orange} orange \color{black} dotted arrows are separations under CPTP access from this work. }
\label{fig}
\end{figure}

\subsubsection{Discussion of Separations.}
\label{sec:discussion}



Separation results make use of the fact that fully black-box constructions relativize \cite{IR89,CM24}. In particular, to establish the non-existence of black-box constructions granting unitary/inverse access in the plain model (\cref{def:BB with  access to inverse}), it is sufficient to show that there does not exist such constructions relative to a unitary oracle with inverse access. Similarly, to rule out black-box constructions with CPTP access (\cref{def:BB with CPTP access}), it is sufficient to show that no such constructions exist relative to a CPTP oracle. 

In particular, our first separation (\cref{informal thm 3}), separating $\PRG$ from inverse access to $\QKPRF$s, is achieved by demonstrating the unattainability of such constructions relative to a unitary oracle with inverse access. Advantageously, this separation precludes black-box constructions that use purified/inverse/unitary versions of the adversary. Certain reductions, such as in quantum proofs of knowledge \cite{U12,U16} and quantum traitor tracing \cite{Z20}, rely on inverse access to the adversary.   

On the other hand, the second and third separations (\cref{informal thm 1,informal thm 2}) are proven by demonstrating the impossibility relative to CPTP oracles, which are not unitary and cannot be purified or inverted. As a result, these separations are weaker as they only preclude black-box constructions that do not assume access to purified/inverse/unitary versions of the adversary. 

However, in many black-box constructions, such strong forms of access (e.g., purification, unitarization, or inversion of the adversary) are unnecessary. Indeed, many known black-box constructions in this domain, such as $\PRG\rightarrow \botPRG$, $\botPRG\rightarrow \QKPRG$, and $\PRG\rightarrow \OWSG$ do not require such strong access. Hence, we believe that separations based on CPTP oracles still yields meaningful results. Notably, prior works \cite{FK15,GMM+24} have also considered CPTP oracle separations. 

That said, we also investigated whether the second and third separations could be based on unitary oracles. For the separation between deterministic and psuedodeterministic primitives, specifically $\OWSG$s and $\botPRG$, doing so appears particularly challenging: our separation strategy relies on inherent randomness in the oracle to achieve a level of pseudodeterminism, which is critical to ensuring that the oracle cannot be leveraged to build deterministic primitives such as $\OWSG$s, while still being useful in building $\botPRG$s. Yet, another difficulty emerged when attempting to separate $\botPRG$ from $\QKPRF$ using a unitary oracle, due to the pseudodeterminism of $\botPRG$, as we shall discuss in \cref{sec:tech overview}. An interesting avenue for future work is to strengthen the second or third separation by basing them on unitary oracles. 


\subsection{Relation to Previous Work}

We discuss relation to previous work.
\begin{itemize}
    \item Prior research has typically defined $\PRS$s and $\OWSG$s with uniform input sampling. However, some works have defined them with quantum input sampling, such as \cite{MY22b,KT24}. Nevertheless, the relations among certain MicroCrypt primitives with quantum sampling and the advantage of quantum sampling in yielding potentially weaker assumptions have not been previously recognized. The latter insight is crucial for known MicroCrypt primitives as well, enabling us to identify multiple new separations, as outlined in the previous section (\cref{informal cor 3,informal cor 1,informal cor 2}). None of the separations mentioned were known prior to this work. 

    \item Previous research did not establish a connection between quantum input sampling and $\bot$-pseudodeterminism. This connection enabled us to address the pseudodeterminism inherent in $\botPRG$s (built in \cite{BBO+24} from $\SPRS$) by converting them to $\QKPRG$s. 

\item It may seem that \textsf{EV-OWPuzz}s can be viewed as a one-way function with a quantum input sampler. However, critically, these puzzles are non-deterministic. In fact, \cite{CGG24} use this property to show that uniform and quantum sampling versions of these puzzles are equivalent. Hence, quantum input sampling seems more interesting to study in deterministic notions such as \textsf{PRG}s. 

\item All definitions in this work consider quantum-evaluable algorithms, and we never require any algorithm to be classically-evaluable. This is noteworthy since, very recently, \cite{KQT24} used a classical oracle to show a separation between classically-evaluable \textsf{OWF}s and quantumly-evaluable \textsf{OWF}s. Our separations are not comparable with theirs. However, as a direct result, they find that classically-evaluable \textsf{OWF}s are black-box separated from many cryptographic applications of quantum-evaluable \textsf{OWF}s, such as \textsf{EUF-MAC} and \textsf{CCA2-SKE}. Given that many of these applications can be built from $\QKPRF$ in the same way, our separations imply that even $\OWSG$s are separated from these applications under CPTP access, and (quantum-evaluable) $\OWF$s are separated from these applications under inverse access.  
\end{itemize}

\subsection{Technical Overview}
\label{sec:tech overview}

We now describe our results in more detail. 

\subsubsection{Quantum Input Sampling.}

Our study reveals surprising equivalences among several variants of MicroCrypt primitives utilizing quantum input sampling.

\paragraph{$\BQSPRS$s imply $\QKPRG$s.}  

Our first result is that \emph{bounded-copy} $\QKSPRS$ imply $\QKPRG$s. Prior works \cite{ALY23,BBO+24} demonstrated that $\SPRS$s enable $\botPRG$s. We first extend this result by showing that bounded-copy $\SPRS$s (\textsf{BC}-\textsf{SPRS}) suffice for this construction.

At first glance, using \textsf{BC}-\textsf{SPRS} appears infeasible because each evaluation of the $\botPRG$ exhausts several copies of the $\SPRS$ and the adversary has access to arbitrarily many evaluations in the security experiment. However, the $\botPRG$ only uses these copies to perform tomography and extract a classical string. Crucially, the extracted strings remain largely consistent across evaluations. Thus, the information gained through multiple evaluations can be simulated with only a limited number of $\SPRS$ copies. By formalizing this observation, we construct a $\botPRG$ from a \textsf{BC}-\textsf{SPRS}.

Our next idea is to show that $\botPRG$s imply (deterministic) $\QKPRG$s, thereby showing that quantum input sampling resolves the pseudodeterminism problem. The idea is simple: to construct a $\QKPRG$ from a $\botPRG$, we search for a good input for the $\botPRG$ during the input sampling phase and use this input during evaluation. However, this approach sacrifices uniform input sampling, which compromises the security reduction given in \cite{ALY23,BBO+24}, thereby only giving a \emph{weak} $\QKPRG$. Standard amplification techniques are then applied to achieve strong security. Hence, we obtain $\QKPRG$ from \textsf{BC}-\textsf{SPRS}. 

Finally, since we are allowed to use a quantum input sampler for a $\QKPRG$, we can perform the same conversion starting instead with a $\BQSPRS$. Therefore, we obtain $\QKPRG$s from $\BQSPRS$.  


\paragraph{$\QKPRG$s imply $\QKSPRS$s.} 

We also establish the converse: $\QKSPRS$s can be built from $\QKPRG$s. This follows in the same way as the construction of $\PRS$s from $\PRF$s given in \cite{JLS18}, but instantiated with a polynomial-domain $\PRF$. Note that $\QKPRF$ with polynomial domain can be trivially derived from $\QKPRG$s \footnote{The output of the $\QKPRG$ can be interpreted as the complete description of a $\QKPRF$ with polynomial domain.}.

\paragraph{Modifying the Size of $\SPRS$.}

By leveraging the above equivalence, we obtain a way to decrease the output length of a $\QKSPRS$. Simply convert a $\QKSPRS$ into a $\QKPRG$, and then convert this back into a $\QKSPRS$. Due to the change in parameters during this conversion, the resulting $\QKSPRS$ is smaller in size. Interestingly, a similar result is not known for $\SPRS$. 

Furthermore, these equivalences can also be leveraged to increase the output length of a $\SPRS$. Note that it is trivial to extend the output length of a $\botPRG$ by composition. For instance, if $G_\lambda$ is a $\botPRG$ mapping $\{0,1\}^\lambda$ to $\{0,1\}^{2\lambda}$, i.e. with expansion factor of $2$, then the composition $G_{2\lambda}\circ G_{\lambda}$ is a $\botPRG$ with expansion factor $4$. Hence, starting with a $\SPRS$, building a $\botPRG$, extending the output length of the $\botPRG$ sufficiently through composition, building a $\QKPRG$, and finally converting this to a $\QKSPRS$, we obtain a method to convert a $\SPRS$ into a $\QKSPRS$ with larger output length.  

\paragraph{Other Constructions.}
 
On the downside, we face an obstacle in the quantum input sampling regime when attempting to build a $\QKPRF$s from $\QKPRG$s. While a $\QKPRG$ with sufficient expansion easily implies a $\QKPRF$ with polynomial domain, it is not clear if a $\QKPRG$ can be used to build full-fledged $\QKPRF$ with exponential domain. Note that the standard conversion of a $\PRG$ to a $\PRF$ \cite{GGM86}, and its quantum adaption \cite{Z12}, both implicitly use the uniform input sampling property of $\PRG$s. Hence, adapting this conversion to the quantum input sampling setting is an interesting open question.  

Fortunately, $\QKPRF$ with polynomial domain can still be useful. Using domain extension techniques \cite{DMP20}, we convert them into \emph{bounded-query} $\QKPRF$s with exponential domain. These, in turn, enable the construction of bounded-copy linear-length $\QKPRS$ and bounded-query $\QKPRU$, following a similar approach outlined in \cite{JLS18,MH24} for the uniform sampling setting.

\subsubsection{Separation Results}
We present three separation results in this work that complement the positive results discussed above. We only provide a simplified overview of the core ideas of the separations and many technical considerations are not discussed here to maintain clarity. 


\paragraph{Separating $\PRG$s from $\QKPRF$s.}

Our first and main separation is between $\PRG$s and $\QKPRF$s. This separation is established using a unitary oracle with access to the inverse. As discussed in \cref{sec:discussion}, this precludes a wide class of black-box constructions.

We define three oracles based on a pair of random functions $(O, P)$:

\begin{itemize}
    \item $\mathcal{C}$: An oracle for membership in a \textsf{PSPACE}-complete language.
    \item $\sigma$: A unitary ``flip'' oracle that swaps the state $\ket{0^{2n}}$ with the state $\ket{\psi_n}=\sum_{x\in \{0,1\}^n}\ket{x}\ket{O(x)}$ and acts as the identity on all other orthogonal states \footnote{Similar ``flip'' oracles were used in recent works: \cite{CCS24,BCN25,BMM+25}.}.  
    \item $\mathcal{O}$: Unitary of the classical function that maps $(x,y,z)$ to $P(x,z)$ if $O(x)=y$ and to $\bot$ otherwise.
\end{itemize}

It is straightforward to construct a $\QKPRF$ relative to these oracles. The quantum key generation algorithm of the $\QKPRF$ queries $\sigma(\ket{0^{2n}})$ and measures the response in the computational basis to obtain a pair $(x^*,O(x^*))$, which serves as the secret key. On input $z$, the evaluation algorithm computes $\mathcal{O}(x^*,O(x^*),z)=P(x^*,z)$ and outputs the result $P(x^*,z)$. Since the same key is reused across evaluations, this yields deterministic outputs. Furthermore, given that a quantum-accessible random oracle acts as a (quantum-query-secure) $\PRF$ \cite{SXY18}, it is not difficult to show that this construction constitutes a $\QKPRF$.  

The more difficult part is showing that $\PRG$s cannot exist relative to these oracles. Our strategy is to show that any candidate generator must be independent of the oracles $(\sigma,\mathcal{O})$ in order to establish that it can be broken using $\mathcal{C}$. In particular, measuring $\sigma(\ket{0^{2n}})$ returns a different pair $(x, O(x))$ each time, and $\mathcal{O}$ then evaluates a different function $P(x, \cdot)$ depending on the value of $x$. On the other hand, a $\PRG$ should produce almost-deterministic outputs, meaning that it cannot depend on this inconsistent randomness.  

To make this rigorous, we make small modifications to the oracle and argue that the generator’s output should remain stable under such perturbations. The main idea is that a $\PRG$ cannot distinguish between oracles that differ only negligibly in the states they produce. Given that a $\PRG$ produces deterministic classical values, this implies that its output must remain invariant under such small perturbations. By applying this reasoning inductively, we conclude that the generator’s output must remain stable even if the oracles are replaced entirely.

This, in turn, implies that the oracles can be simulated using internal randomness, effectively yielding a $\PRG$ that does not require oracle access. In other words, since the $\PRG$'s output is the same regardless of the oracle, we can simulate its computation by sampling random strings directly, rather than querying the oracles. But a $\PRG$ that does not use oracle access to $(\sigma,\mathcal{O})$ can easily be broken with a \textsf{PSPACE} oracle. Thus, we establish that relative to the unitary oracles $(\sigma,\mathcal{O},\mathcal{C})$, there is no fully black-box construction of a $\PRG$ from a $\QKPRF$ with inverse access. Given that fully black-box constructions relativize, we obtain the impossibility of such constructions in the plain model. 

\paragraph{Separating $\OWSG$s from $\botPRG$s.}

We show that there does not exist a fully black-box construction of a $\OWSG$ from CPTP access to a $\botPRG$. We prove this by demonstrating the separation relative to a CPTP oracle. This result is somewhat surprising, since $\OWSG$s are considered weaker than $\PRG$s, as evidenced by existing separations between them \cite{K21}. Instead, our separation emphasizes the distinction in determinism to separate $\OWSG$s from $\botPRG$s. Recall, a $\botPRG$ is the same as a $\PRG$ except on a ${1}/{\poly}$ fraction of inputs, the algorithm may return $\bot$ sometimes.

Our separation leverages two oracles. The first oracle is for membership in a \textsf{PSPACE}-complete language, used to break any $\OWSG$. The second oracle, denoted $\mathcal{O}$, is a modified quantum random oracle with an abort mechanism. This oracle exhibits inherent $\bot$-pseudodeterminism, making it well-suited for constructing $\botPRG$s but unsuitable for deterministic primitives like $\OWSG$s. One downside of this design is that it complicates lifting the oracle to a unitary map.

The oracle $\mathcal{O}$ operates as follows: on input $x\in \{0,1\}^n$, it computes $(a_x,b_x,c_x)\leftarrow O(x)$, where $O$ is a random function mapping $n$-bits to $3n$-bits. The first component, $a_x$, determines whether $x$ is a ``good'' (deterministic) or ``bad'' (may evaluate to $\bot$) input. If $x$ is deemed bad, which occurs with $1/\poly$ probability, then $\mathcal{O}$ outputs $c_x$ with probability $1-\frac{b_x}{2^n}$ and $\bot$ with $\frac{b_x}{2^n}$ probability, where $b_x$ is interpreted as an integer. Otherwise, if $x$ is deemed good, $\mathcal{O}$ outputs $c_x$ with probability 1. It is easy to verify that $\mathcal{O}$ behaves as a valid $\botPRG$. 

The challenge lies in showing that $\mathcal{O}$ cannot be used for constructing $\OWSG$s. Formalizing this requires some effort, but the main idea is that a generator, relying on $\mathcal{O}(x)$ for some $x$, cannot infer whether $x$ is good or bad with certainty, given that the error probability (i.e., the chance of receiving $\bot$) can be very small. When $x$ is bad, small changes in the error probability cannot be distinguished in polynomial time. As a result, the generator's output must remain stable under such small variations. Extending this reasoning, we show that the output must remain unchanged even when $\mathcal{O}(x)$ always returns $\bot$. Informally, this implies that the generator's output is independent of $\mathcal{O}$. This independence allows us to construct an $\OWSG$ that does not use $\mathcal{O}$. But since $\OWSG$s cannot exist relative to a \textsf{PSPACE} oracle \cite{CGG+23}, we reach a contradiction. Hence, no fully black-box construction of an $\OWSG$ from a $\botPRG$ is possible in this setting.

\paragraph{Separating $\botPRG$s from $\QKPRF$s.}

Our third result separates $\botPRG$s from CPTP access to a $\QKPRF$s. We prove this result by demonstrating a separation relative to a CPTP oracle. This separation, like our first, hinges on the distinction between quantum and classical input sampling procedures. We use a similar oracle setup—$(\sigma, \mathcal{O}, \mathcal{C})$—but in this case, $\sigma$ is defined as a CPTP map that samples $x$ uniformly at random and outputs the pair $(x, O(x))$.

Constructing a $\QKPRF$ relative to these oracles proceeds analogously to the first separation and is relatively straightforward. The more involved part is proving that $\botPRG$s cannot exist relative to this oracle setup. Intuitively, since $\sigma$ outputs fresh, random pairs $(x, O(x))$ with each call, every evaluation of the $\botPRG$ interacts with an essentially independent instantiation of $\mathcal{O}$. Because $\botPRG$s are expected to exhibit some degree of determinism, they cannot depend on these inconsistent oracle outputs. More precisely, we show that the behavior of the $\botPRG$ can be simulated by replacing oracle access with sampling random strings directly.

Consequently, there would exist a $\botPRG$ that operates without oracle access yet remains secure against adversaries with access to a \textsf{PSPACE} oracle, which is a contradiction. 

Importantly, this argument breaks down if $\sigma$ were defined as in the first separation. In that setting, $\sigma$ is a unitary map, so its outputs are no longer distinct. Furthermore, unlike in the first separation, we cannot argue that the $\botPRG$ output remains invariant under small perturbations to the oracle responses: slight variations in $\sigma$ could induce slight differences in the probability of returning $\bot$. Consequently, the set of ``bad'' inputs as well as the evaluations on this set can change across different oracles. To avoid this issue, we require that $\sigma$ yield classical outputs, ensuring that different evaluations receive completely independent query results.





\section{Preliminaries}
\label{sec:prelim}

\subsection{Notations}
\label{sec:notations}

We let $[n]= \{1,2,\ldots,n\}$, $[n:n+k]= \{n,n+1,\ldots, n+k\}$, and $y_{[n:n+k]}= y_ny_{n+1}\ldots y_{n+k}$ for every $n,k\in \NN$ and string $y$ of length at least $n+k$. Furthermore, we let $\negl[x]$ denote any function that is asymptotically smaller than the inverse of any polynomial. 

We let $x\leftarrow X$ denote that $x$ is chosen from the values in $X$, according to the distribution $X$. If $X$ is a set, then $x\leftarrow X$ simply means $x$ is chosen uniformly at random from the set. We say $(a,\cdot)\in X$ if there exists an element $b$ such that $(a,b)\in X$. We let $\Pi_{n,m}=(\{0,1\}^m)^{\{0,1\}^n}$ denote the set of functions mapping $n$-bits to $m$-bits, and $\Pi_n$ denote the set of permutations on $n$-bits.

We refer the reader to \cite{NC00} for a detailed exposition to preliminary quantum information. We let $\mathcal{S}(\hil{H})$ and $\mathcal{U}(\hil{H})$ denote the set of unit vectors and unitary operators, respectively, on the Hilbert space $\hil{H}$ and let $\textsf{Haar}(\mathbb{C}^d)$ denote the Haar measure over $\mathbb{C}^d$ which is the uniform measure over all $d$-dimensional unit vectors. We let $d_{\textsf{TD}}$ denote the total trace distance between two density matrices or two distributions.

We follow the standard notations to define quantum algorithms. We say that a quantum algorithm $A$ is \emph{QPT} if it consists of a family of quantum algorithms $\{A_\lambda\}_{\lambda}$ such that the run-time of each algorithm $A_\lambda$ is bounded by some polynomial $p(\lambda)$. Furthermore, we say that a quantum algorithm $A=\{A_\lambda\}_\lambda$ is \emph{almost-deterministic} if there exists a negligible function $\epsilon$, such that for every $\lambda\in \mathbb{N}$ and every input $x$ in the domain of $A_\lambda$, there exists an (possibly quantum) output $y$ satisfying $\Pr[A_\lambda(x)=y]\geq 1-\epsilon(\lambda)$. We also avoid using the $\lambda$ subscript in algorithms to avoid excessive notation.


We say $A$ has \emph{quantum-query access} to a classical function $P$ to mean that it is given polynomial quantum queries to the unitary map $    U_P:\ket{x}\rightarrow \ket{x}\ket{P(x)}.$


\subsection{Black-Box Separation}

The notion of oracle black-box separations was first considered in \cite{IR89} and later formalized in the quantum setting in \cite{CM24}. We just recall the definitions relevant for this work from \cite{CM24}. 

\begin{definition}
A \emph{primitive} $P$ is a pair $P = (\mathcal{F}_P , \mathcal{R}_P )$ \footnote{We can think of $\mathcal{F}_P$ to mean the ``correctness'' conditions of $P$ and $\mathcal{R}_P$ to mean the ``security'' conditions of $P$.} where $\mathcal{F}_P$ is a set of quantum channels, and $\mathcal{R}_P$ is a relation over pairs $(G, \adv)$ of quantum channels, where $G \in \mathcal{F}_P$. 

A quantum channel $G$ is an \emph{implementation} of $P$ if $G \in \mathcal{F}_P$. If $G$ is additionally a QPT channel, then
we say that $G$ is an \emph{efficient implementation} of $P$.
A quantum channel $\adv$ \emph{$P$-breaks} $G \in \mathcal{F}_P$ if $(G, \adv) \in \mathcal{R}_P$. We say
that $G$ is a \emph{secure implementation} of $P$ if $G$ is an implementation of $P$ such that no QPT channel $P$-breaks
it. The primitive $P$ \emph{exists} if there exists an efficient and secure implementation of $P$.
\end{definition}

We now formalize the notion of constructions relative to an oracle.

\begin{definition}
    We say that a primitive $P$ exists relative to an oracle $\mathcal{O}$ if:
    \begin{itemize}
        \item There exists QPT oracle-access algorithm $G^{(\cdot)}$ such that $G^\mathcal{O}\in \mathcal{F}_P$.
        \item The security of $G^\mathcal{O}$ holds against all QPT adversaries with access to $\mathcal{O}$ i.e. for all QPT $\adv$, $(G^\mathcal{O},\adv^\mathcal{O})\notin \mathcal{R}_P$.
    \end{itemize}
\end{definition}

We are now ready to define the notion of fully black-box construction. We define two versions: under unitary/inverse access and under CPTP access.   

\begin{definition}
\label{def:BB with  access to inverse}
A QPT algorithm $G^{(\cdot)}$ is a fully black-box construction of $Q$ from $P$ \textbf{with inverse access} if the following two conditions hold:
\begin{enumerate}
    \item For every unitary implementation $U$ of $P$, $G^{U,U^{-1}}\in \mathcal{F}_Q$.
\item There is a QPT algorithm $S^{(\cdot)}$ such that, for every
unitary implementation $U$ of $P$, every adversary $\adv$ that $Q$-breaks $G^{U,U^{-1}}$, and every unitary implementation
$\tilde{\adv}$ of $\adv$, it holds that $S^{\tilde{\adv},\tilde{\adv}^{-1}}$ $P$-breaks $U$.
\end{enumerate} 
\end{definition}

The following result from \cite{CM24} shows the relation between fully black-box constructions and oracle separations. 

\begin{theorem}[Theorem 4.2 in \cite{CM24}]
\label{thm:separation relative to unitary}
    Assume there exists a fully black-box construction of a primitive $Q$ from a primitive $P$ with inverse access. Then, for any unitary $\mathcal{O}$, if $P$ exists relative to $(\mathcal{O},\mathcal{O}^{-1})$, then $Q$ exists relative to $(\mathcal{O},\mathcal{O}^{-1})$.
\end{theorem}



We now consider the setting of CPTP oracles.

\begin{definition}
\label{def:BB with CPTP access}
A QPT algorithm $G^{(\cdot)}$ is a fully black-box construction of $Q$ from \textbf{CPTP access} to $P$ if the following two conditions hold:
\begin{enumerate}
    \item For every CPTP implementation $\mathcal{C}$ of $P$, $G^\mathcal{C}$ is an implementation of $Q$.
\item There is a QPT algorithm $S^{(\cdot)}$ such that, for every
CPTP implementation $\mathcal{C}$ of $P$, every adversary $\adv$ that $Q$-breaks $G^\mathcal{C}$, it holds that $S^{{\adv}}$ $P$-breaks $\mathcal{C}$.
\end{enumerate} 
\end{definition}

We also obtain a relationship between fully black-box constructions and oracle separations in the CPTP setting. 

\begin{theorem}
\label{thm:separation relative to CPTP}
    Assume there exists a fully black-box construction of a primitive $Q$ from CPTP access to $P$. Then, for any CPTP oracle $\mathcal{O}$, if $P$ exists relative to $\mathcal{O}$, then $Q$ exists relative to $\mathcal{O}$.
\end{theorem}

The proof of the above result follows in the same way as the proof of Theorem 4.2 in \cite{CM24}.


\subsection{MicroCrypt Primitives}

We recall several MicroCrypt assumptions relevant to this work. 

First, we recall the notion of one-way state generators ($\OWSG$s). In this work, we only consider \emph{pure} $\OWSG$s meaning that the output is always a pure state. 

Note that different works define correctness of a $\OWSG$ slightly differently. Our notion requires that on any input, the generator produces a fixed pure-state except with negligible probability. This encompasses some earlier definitions, such as in \cite{CGG+23}, but is less general than other variants, such as in \cite{MY22a}. We note, however, that our separation can be adapted to a slightly more general notion which only requires that the generator produce a fixed pure-state except with negligible probability on all but a negligible fraction of inputs. 

However, our separation does not apply to certain more general notions such as in \cite{MY22a}. Specifically, our separation result uses our correctness condition and the result does not hold for the more generalized notion of \cite{MY22a}. We do not believe this to be a major issue since, to our knowledge, all known constructions of $\OWSG$s satisfy our correctness condition. 


\begin{definition}[One-Way State Generator]
\label{defn:OWSG}
    Let $\lambda\in \mathbb{N}$ be the security parameter and let $n=n(\lambda)$ be polynomial in $\lambda$. An almost-deterministic \footnote{Recall, $G$ is \emph{almost-deterministic} if on any input, $G$ outputs the same value except with negligible probability. See \cref{sec:notations} for the formal definition.} QPT algorithm $G$ is called a $n$-\emph{one-way state generator (\textsf{OWSG})}, if it generates $n$-qubit pure-states and for any polynomial $t=t(\lambda)$ and QPT distinguisher $\adv$, there exists a function $\epsilon(\cdot)$ such that:
        \begin{align*}
            \textsf{Adtg}^{\textsf{OWSG}}_{\adv,G}(1^\lambda,1^t)\coloneqq \Pr\left[\textsf{Exp}^{\textsf{OWSG}}_{\adv,G}(1^\lambda,1^t)=1\right]\leq \epsilon(\lambda).
        \end{align*}
        We say that $G$ is a $\OWSG$ if for every QPT $\adv$, $\epsilon$ is a negligible function. We say that $G$ is a \emph{weak} $\OWSG$ if for every QPT $\adv$, $\epsilon\leq\frac{1}{p}$ for some polynomial $p$.
\end{definition}

\smallskip \noindent\fbox{%
    \parbox{\textwidth}{%
    $\textsf{Exp}^{\textsf{OWSG}}_{\adv,G}(1^\lambda,1^t)$:
\begin{enumerate}
    \item Sample $k\leftarrow \{0,1\}^\lambda$.
    \item For each $i\in [t+1]$ generate $\ket{\psi_i}\leftarrow G(k)$.
    \item $k' \leftarrow \adv(\otimes_{i\in [t]}\ket{\psi_i})$. Let $\ket{\phi_{k'}}\leftarrow G(k')$.
    \item Measure $\ket{\psi_{t+1}}$ with $\{\ket{\phi_{k'}}\bra{\phi_{k'}},I-\ket{\phi_{k'}}\bra{\phi_{k'}}\}$ and if the result is $\ket{\phi_{k'}}\bra{\phi_{k'}}$, then output $b=1$, and output $b=0$ otherwise. 
\end{enumerate}}}
    \smallskip  

\cite{CGG+23} show that if $\textsf{PSPACE}=\textsf{BQP}$, then $\OWSG$s do not exist. In fact, the attack presented succeeds with probability at least $\frac{1}{2}$. 

\begin{lemma}[\cite{CGG+23}]
\label{lem:OWSG 3}
For any $\OWSG$ $G$, there exists a \textsf{PSPACE} algorithm ${\adv}$ and a polynomial $t=t(\lambda)$ such that 
\begin{align*}
            \textsf{Adtg}^{\textsf{OWSG}}_{\adv,G}(1^\lambda,1^t)\geq \frac{1}{2}.
        \end{align*}
        \end{lemma}

We define pseudorandom states ($\PRS$s), first introduced in \cite{JLS18}. 

\begin{definition}[Pseudorandom State Generator]
\label{def:prs}
    Let $\lambda\in \mathbb{N}$ be the security parameter and let $n=n(\lambda)$ be polynomial in $\lambda$. An almost-deterministic QPT algorithm $\textsf{PRS}$ is called a $n$-\emph{pseudorandom state generator (\textsf{PRS})}, if it generates $n$-qubit pure-states and the following holds:
    
    For any polynomial $t(\cdot)$ and QPT distinguisher $\adv$:
        \begin{align*}
            \left|  \Pr_{{k}\leftarrow \{0,1\}^\lambda} \left[\adv \left(\textsf{PRS} ({k})^{\otimes t(\lambda)}\right)=1\right]-\Pr_{\ket{\phi}\leftarrow \textsf{Haar}(\mathbb{C}^n)} \left[\adv \left(\ket{\phi}^{\otimes t(\lambda)}\right)=1\right]\right| \leq \negl[\lambda].
        \end{align*} 
    We divide $\PRS$ into three regimes, based on the state size $n$:
    \begin{enumerate}
        \item $n=c\cdot \log(\lambda)$ with $c\ll 1$. 
        \item $n=c\cdot \log(\lambda)$ with $c\geq 1$, which we call \emph{short pseudorandom states ($\SPRS$s)}.
        \item $n=\Omega (\lambda)$, which we call \emph{long pseudorandom states ($\LPRS$s)}.
    \end{enumerate}  
\end{definition}

We will also recall the standard definition for $\PRG$s.

\begin{definition}[Pseudorandom Generator]
    Let $\lambda\in \mathbb{N}$ be the security parameter and let $n=n(\lambda)$ be polynomial in $\lambda$. An almost-deterministic QPT algorithm $G$ mapping $\{0,1\}^\lambda$ to $\{0,1\}^n$ is called a $n$-\emph{pseudorandom generator ($\PRG$)}, if $n>\lambda$ for all $\lambda\in \mathbb{N}$ and for any QPT distinguisher $\adv$:
        \begin{align*}
            \left|  \Pr_{{k}\leftarrow \{0,1\}^\lambda} \left[\adv (G(k))=1\right]-\Pr_{y\leftarrow \{0,1\}^n} \left[\adv (y)=1\right]\right| \leq \negl[\lambda].
        \end{align*}
\end{definition}

\subsection{Pseudodeterministic Pseudorandom Strings from Pseudorandom States}
\label{sec:pr-from-sprs}

We describe the procedure given in \cite{ALY23} to extract classical pseudorandom strings from pseudorandom states. The original procedure is, for some states, pseudodeterministic meaning that running this procedure on the same state may yield different outcomes each time. However, it was shown that there exists a good set of states such that the extraction procedure is deterministic. 

We first recall the notion of tomography, which is used to extract a classical approximate description of a quantum state. 

\begin{lemma}[Corollary 7.6 \cite{AGQ22}]
\label{lem:tomography}
For any error tolerance $\delta=\delta(\lambda)\in (0,1]$ and any dimension $d=d(\lambda)\in \mathbb{N}$ using at least $t=t(\lambda)\coloneqq 36\lambda d^3/\delta$ copies of a $d$-dimensional density matrix $\rho$, the process $\textsf{Tomography}(\rho^{\otimes t})$ runs in polynomial time with respect to $\lambda$, $d$, $1/\delta$ and outputs a matrix $M\in \mathbb{C}^{d\times d}$ such that 
\begin{align*}
    \Pr\left[\| \rho -M\|\leq \delta: M\leftarrow \textsf{Tomography}(\rho^{\otimes t})\right]\geq 1-\negl[\lambda].
\end{align*}    
\end{lemma}

We now recall how \textsf{Tomography} is used to extract pseudorandom strings in \cite{ALY23}. 

\begin{construct}[\textsf{Extract} \cite{ALY23}]
Let $\lambda\in \mathbb{N}$ be the security parameter. The algorithm $\textsf{Extract}$ is defined as follows:
    \begin{itemize}
        \item \textsf{Input:} $t\coloneqq 144\lambda d^8$ copies of a $d$-dimensional quantum state $\rho$.
        \item Perform $\textsf{Tomography}(\rho^{\otimes t})$ with error tolerance $\delta\coloneqq d^{-5/6}$ to obtain a classical matrix $M\in \mathbb{C}^{d\times d}$.
        \item Run $\textsf{Round}(M)$ to get $y\in \{0,1\}^\ell$. Output $y$.
    \end{itemize}
\end{construct}

\smallskip \noindent\fbox{%
    \parbox{\textwidth}{%
    $\textsf{Round}(M)$:
    \text{\textsf{Input:} Matrix $M\in \mathbb{C}^{d\times d}$.}
\begin{itemize}
  \item Define $k\coloneqq d^{5/6}$, $r\coloneqq d^{2/3}$, and $\ell\coloneqq d^{1/6}$.
  \item Let $p_1,\ldots, p_d$ be the diagonal entries of $M$. 
  \item For $i\in [\ell]$:
  \begin{enumerate}
      \item Let $q_i\coloneqq \sum_{j=1}^rp_{(i-1)r+j}$. 
      \item Define $b_i\coloneqq \begin{cases}
   1        & \text{if } q_i>r/d \\
        0        & \text{if } q_i\leq r/d.
    \end{cases}$
  \end{enumerate}
  \item Output $b_1\|\ldots \| b_\ell$.
\end{itemize}}}
    \smallskip  

This extraction procedure was shown to satisfy the following two lemmas.

\begin{lemma}[Lemma 3.6 in \cite{ALY23}]
\label{lem:uniform}
If $\textsf{Extract}$ is run on a Haar random state, then ${d}_{\textsf{TD}}((q_1,\ldots, q_\ell),Z/(2d))\leq O(k/d)+O(\ell/\sqrt{r})$ where $Z$ is a random variable in $\mathbb{R}^\ell$ with i.i.d. $\mathcal{N}(2r,4r)$ entries. In other words, $Z/(2d)$ has an i.i.d. $\mathcal{N}(r/d,r/d^2)$ entries.
\end{lemma}

\begin{lemma}[Claim 3.7 in \cite{ALY23}]
\label{claim 3.7}
     Define 
    \begin{align*}
        \mathcal{G}_d\coloneqq \left\{\ket{\psi}\in \mathcal{S}(\mathbb{C}^d):\forall i\in [\ell],\left| q_i-\frac{r}{d}\right| > 2/d\right\}.
    \end{align*}
    Then, 
    \begin{enumerate}
        \item $\Pr\left[\ket{\psi}\in \mathcal{G}_d:\ket{\psi}\leftarrow \textsf{Haar}(\mathbb{C}^d)\right]\geq 1-O(d^{-1/6}).$
    \item There exists a negligible function $\negl$ such that for any $\ket{\psi}\in \mathcal{G}_d$, there exists a string $y$ such that $\Pr\left[y\leftarrow  \textsf{Extract}\left(\ket{\psi}^{\otimes t}\right)\right]\geq 1-\negl[\lambda]$.
    \end{enumerate}
\end{lemma}

\subsection{Pseudodeterministic Primitives in MicroCrypt}

The $\textsf{Extract}$ procedure described in the previous section was used to construct a QPT algorithm termed pseudodeterministic $\PRG$s from a $\SPRS$. The non-determinism in these algorithms make their use in cryptography difficult. Hence, the follow-up work \cite{BBO+24} converted pseudodeterministic $\PRG$s into a notion known as $\botPRG$s.

We introduce $\botPRG$s now. Note that it is easy to distinguish $\bot$ evaluations from random. However, it is sufficient to only require indistinguishability for non-$\bot$ evaluations. This is incorporated in the security game by providing $\bot$ in the truly random case as well. See \cite{BBO+24} for a more in-depth discussion. 

\begin{definition}[$\isbot$]
We define the operator 
\begin{align*}
    \isbot(a,b):=\begin{cases}
    \bot        & \text{if } a = \bot \\
        b        & \text{otherwise}.
    \end{cases}
\end{align*}
\end{definition}

\begin{definition}[$\bot$-Pseudorandom Generator]
\label{def:botprg}
    Let $\lambda\in \mathbb{N}$ be the security parameter and let $m=m(\lambda)$ be polynomial in $\lambda$. A QPT algorithm $ G$ mapping $\{0,1\}^\lambda$ to $\{0,1\}^m \cup\{\bot\}$, is a \emph{$(\mu,m)$-$\bot$-pseudodeterministic pseudorandom generator ($\botPRG$)} if:
    \begin{enumerate}
        \item \emph{(Expansion)} $m(\lambda)>\lambda$ for all $\lambda \in \mathbb{N}$.
        \item \emph{(Pseudodeterminism)} There exist a constant $c>0$ such that $\mu(\lambda)= O(\lambda^{-c})$ and for sufficiently large $\lambda\in \mathbb{N}$ there exists a set $\mathcal{G}_\lambda \subseteq\{0,1\}^\lambda$ such that the following holds:
        \begin{enumerate}
            \item \[\Pr_{x\gets \{0,1\}^\lambda}\left[x\in\mathcal{G}_\lambda \right] \geq 1-\mu(\lambda).\] 
            \item For every $x\in \mathcal{G}_\lambda$ there exists a non-$\bot$ value $y\in\{0,1\}^m$ such that: 
            \begin{align}
                \Pr\left[G(x)=y \right] \geq 1 - \negl[\lambda].
            \end{align} 

            \item For every $x\in \{0,1\}^\lambda$, there exists a non-$\bot$ value $y\in\{0,1\}^m$ such that: 
            \begin{align}
                \Pr\left[G(x)\in \{y,\bot\} \right] \geq 1 - \negl[\lambda].
            \end{align}  
            \end{enumerate}
        \item \emph{(Security)} For every polynomial $q=q(\lambda)$ and QPT distinguisher $\adv$, there exists a negligible function $\epsilon$ such that: 
\begin{align*}
        \left| \Pr \left[ \begin{matrix}
            k\gets \{0,1\}^\lambda\\
            y_1\gets G(k)\\
            \vdots \\
            y_q \gets G(k)      
        \end{matrix} : \adv(y_1,...,y_q) = 1 \right] - \Pr \left[ \begin{matrix}
            k\gets \{0,1\}^\lambda \\
            y\gets \{0,1\}^{m} \\
            y_1\gets \isbot(G(k),y)\\
            \vdots \\
            y_q \gets \isbot(G(k),y)      
        \end{matrix}: \adv(y_1,\ldots,y_q) = 1    
        \right] \right| \leq \epsilon(\lambda)
    \end{align*}
    \end{enumerate}
\end{definition}

$\botPRG$s were constructed from $\SPRS$s in \cite{BBO+24}. 

\begin{lemma}[Corollary 1 \cite{BBO+24}]
\label{lem:bot-prg-from-prs}
     If there exists $(c\log\lambda)$-$\SPRS$ for some constant $c>12$, then there exists a
    $(O(\lambda^{-c/12 +1}),\lambda^{c/12})$ - $\botPRG$.
\end{lemma}

\section{Definitions: Cryptography with Quantum Input Sampling}
\label{sec:def}

We present definitions for various known cryptographic primitives but with quantum input sampling algorithms. First of all, we define $\PRS$ with quantum input sampling.

\begin{definition}[Pseudorandom State Generator with Quantum Input Sampling]
\label{def:qkprs}
    Let $\lambda\in \mathbb{N}$ be the security parameter and let $n=n(\lambda)$ and $m=m(\lambda)$ be polynomials in $\lambda$. A pair of QPT algorithms $(\textsf{QSamp},\textsf{StateGen})$ is called a $(m,n)$-\emph{$\PRS$ with quantum input sampling ($\QKPRS$)}, if the following conditions hold:
    \begin{itemize}
    \item $\textsf{QSamp}(1^\lambda):$ Outputs a string $k\in \{0,1\}^m$.
    \item $\textsf{StateGen}(k):$ Takes a $m$-bit string $k$ and outputs a $n$-qubit pure-state.
      \item \emph{(Determinism)} For every $k\in \{0,1\}^m$, there exists a $n$-qubit state $\ket{\psi_k}$, such that the following condition is satisfied over the distribution of inputs:  
     \begin{align*}
\Pr_{k\leftarrow \textsf{QSamp}(1^\lambda)}\left[ \textsf{StateGen}(k)=\ket{\psi_k}\right] \geq 1-\negl[\lambda].
    \end{align*}
        \item \emph{(Security)} For any polynomial $t(\cdot)$ and QPT distinguisher $\adv$:
        \begin{align*}
            \left|  \Pr_{{k}\leftarrow \textsf{QSamp}(1^\lambda)} \left[\adv \left(\textsf{StateGen} ({k})^{\otimes t(\lambda)}\right)=1\right]-\allowbreak
            \Pr_{\ket{\phi}\leftarrow \textsf{Haar}(\mathbb{C}^n)} \left[ \adv \left(\ket{\phi}^{\otimes t(\lambda)}\right)=1\right] \right| \allowbreak \leq \negl[\lambda].
        \end{align*}
        In the case where security only holds for $t\le q$ for some polynomial $q=q(\lambda)$, then we call this $(q,m,n)$-\emph{bounded-copy $\QKPRS$ ($\BQPRS$)}. If $n=c\cdot \log(\lambda)$ with $c>1$, then we call this $\QKSPRS$ and if $n=\Omega (\lambda)$, then we call this $\QKLPRS$. 
    \end{itemize}
\end{definition}

We now introduce $\PRG$ with quantum input sampling and $\PRF$ with quantum key generation. As far as we know, these notions has not been defined previously. 

\begin{definition}[Pseudorandom Generator with Quantum Input Sampling]\label{def:PRG_qs}
    Let $\lambda\in \mathbb{N}$ be the security parameter and let $n=n(\lambda)$ and $m=m(\lambda)$ be polynomials in $\lambda$. A pair of QPT algorithms $ (\textsf{QSamp},G) $ is a \emph{$(m,n)$-$\PRG$ with quantum input sampling ($\QKPRG$)} if: 
    \begin{enumerate}
      \item $\textsf{QSamp}(1^\lambda):$ Outputs a string $k\in \{0,1\}^m$.
      \item $G(k)$: Takes an input $k\in \{0,1\}^m$ and outputs $y\in \{0,1\}^n$
        \item \emph{(Expansion)} $m(\lambda)<n(\lambda)$ for all $\lambda \in \mathbb{N}$.
          \item \emph{(Determinism)} For every $k\in \{0,1\}^m$, there exists a string $y_k\in \{0,1\}^n$, such that the following condition is satisfied over the distribution of inputs:  
     \begin{align*}
\Pr_{k\leftarrow \textsf{QSamp}(1^\lambda)}\left[ G(k)=y_k\right] \geq 1-\negl[\lambda].
    \end{align*}
        \item \emph{(Security)} For any QPT distinguisher $\adv$, there exists a negligible function $\epsilon$ such that:
   \begin{align*}
            \left|  \Pr_{{k}\leftarrow \textsf{QSamp}(1^\lambda)} \left[\adv (G(k))=1\right]-\Pr_{y\leftarrow \{0,1\}^n} \left[\adv (y)=1\right]\right| \leq \epsilon(\lambda).
        \end{align*}
        We say that $G$ is a $\QKPRG$ if for every QPT $\adv$, $\epsilon$ is a negligible function. We say that $G$ is a \emph{weak} $\QKPRG$ if for every QPT $\adv$, $\epsilon\leq\frac{1}{p}$ for some polynomial $p$.
    \end{enumerate}
\end{definition}

\begin{definition}[Pseudorandom Function with Quantum Key Generation]
    Let $\lambda\in \mathbb{N}$ be the security parameter and let $n=n(\lambda)$ and $m=m(\lambda)$ be polynomials in $\lambda$. A pair of QPT algorithms $(\textsf{QSamp},F)$ is called a $(m,n)$-\emph{$\PRF$ with quantum key generation ($\QKPRF$)}, if:
    \begin{enumerate}
     \item $\textsf{QSamp}(1^\lambda):$ Outputs a key $k\in \{0,1\}^m$.
        \item $F_k(x)$: Takes a key $k\in \{0,1\}^m$ and an input $x\in \{0,1\}^n$ and outputs a string $y\in \{0,1\}^n$ \footnote{For simplicity, we only consider $\QKPRF$s with the same input and output lengths. All our results easily generalize to $\QKPRF$s with different input and output lengths.}. 

        \item \emph{(Determinism)} For every $k\in \{0,1\}^m$ and $x\in \{0,1\}^n$, there exists a string $y_{k,x}\in \{0,1\}^n$ such that for all $x\in \{0,1\}^n$, the following is satisfied over the distribution of keys:
        \begin{align*}
            \Pr_{k\leftarrow \textsf{QSamp}(1^\lambda)}[F_k(x)=y_{k,x}]\geq 1-\negl[\lambda]. 
        \end{align*}

    \item \emph{(Security)} For any QPT distinguisher $\adv$:
        \begin{align*}
            \left|  \Pr_{{k}\leftarrow \textsf{QSamp}(1^\lambda)} \left[\adv^{F_k}(1^\lambda)=1\right]-\Pr_{O\leftarrow \Pi_{n,n}} \left[\adv^O(1^\lambda)=1\right]\right| \leq \negl[\lambda].
        \end{align*}
        We say a $\QKPRF$ is \emph{quantum-query-secure} if the above holds even if $\adv$ is given quantum-query access to $F_k$ and $O$. Furthermore, in the case where security only holds for $t\le q$ queries for some polynomial $q=q(\lambda)$, then we call this a $q$-query $\QKPRF$.
            \end{enumerate}
\end{definition}

\section{Relations among Primitives with Quantum Input Sampling}
\label{sec:qkg}

In this section, we explore relations among MicroCrypt primitives with quantum input sampling algorithms. In summary, we find that there exists black-box constructions for the following: 

\begin{enumerate}
\item $\QKPRG$ from $\botPRG$.
    \item $\QKPRG$ from $\BQSPRS$. 
    \item $\QKSPRS$ from $\QKPRG$. 
    \item $\BQPRU$ from $\QKPRG$ \footnote{This part is shown in \cref{sec:pru-from-prg}. See \cref{def:PRU-qs} for the definition of $\BQPRU$s.}.
\end{enumerate}

We also discuss how these results can be used to modify the output size of a $\SPRS$ in \cref{sec:prs-from-prg}.


\subsection{\texorpdfstring{$\QKPRG$ from $\botPRG$}{PRGqs from bot-PRG}}

We show how to build $\QKPRG$s from $\botPRG$s. First, we build a weak $\QKPRG$, and then amplify security to achieve a standard $\QKPRG$.

\begin{construct}
    \label{con:qkprg-from-botprg}
    Let $m$ be a polynomial on the security parameter $\lambda\in \mathbb{N}$ such that $m>\lambda$. Let $\mu=O(\lambda^{-c})$ for some constant $c>0$. Let $G$ be a $(\mu,m)$-$\botPRG$. The construction for a weak $(\lambda,m)$-$\QKPRG$ is as follows:
    \begin{itemize}
        \item $\textsf{QSamp}(1^\lambda):$ For $i\in [\lambda]:$
        \begin{enumerate}
            \item Sample $k_i\leftarrow \{0,1\}^\lambda$.
            \item For each $j\in [\lambda]$, compute $y_{i,j}\leftarrow G(k_i)$.
            \item If $\textsf{vote}_\lambda(y_{i,1},\ldots, y_{i,\lambda})\neq \perp$ \footnote{$\textsf{vote}(a_1,\ldots,a_\lambda)$ outputs the first most common element in the tuple $(a_1,\ldots,a_\lambda)$.}, then output $k_i$. 
        \end{enumerate}
        Otherwise, output $\bot$.
        
        \item $\textsf{PRG}(k):$ 
        \begin{enumerate}
        \item If $k=\bot$, output $0^m$. 
            \item For each $j\in [\lambda]$, compute $y_j\leftarrow G(k)$. 
            \item If $y_j=\bot$ for all $j\in [\lambda]$, then output $\bot$.
            \item Otherwise, output the first most common non-$\bot$ value in $(y_1,\ldots, y_\lambda)$.
        \end{enumerate}
    \end{itemize}
\end{construct}

\begin{lemma}
\label{lem:prg-from-botprg}
    Construction \ref{con:qkprg-from-botprg} is a weak $(\lambda,m)$-$\QKPRG$ assuming the existence of a $(\mu,m)$-$\botPRG$. 
\end{lemma}

\begin{proof}
We first show that the algorithms $(\textsf{QSamp},\textsf{PRG})$ satisfy the determinism condition of $\QKPRG$s. By the pseudodeterminism condition of $G$, it is clear that there is negligible probability that $ \textsf{QSamp}(1^\lambda)$ outputs $\bot$. 

For any $k\in \{0,1\}^\lambda$, if $\Pr\left[G(k)=\bot\right]\geq \frac{2}{3}$, then there is negligible probability that $\textsf{vote}_\lambda(y_{1},\ldots, y_{\lambda})\neq \perp$, where $y_{j}\leftarrow G(k)$ for all $j\in [\lambda]$. By the union bound, there is negligible probability that $\textsf{QSamp}(1^\lambda)$ outputs a string $k$ such that $\Pr\left[G(k)=\bot\right]\geq \frac{2}{3}$. Therefore, except with negligible probability, the output of $\textsf{QSamp}(1^\lambda)$ is a non-$\bot$ string $k$ such that $\Pr\left[G(k)=\bot\right]< \frac{2}{3}$. 

By the pseudodeterminism of $G$, for any $k\in \{0,1\}^\lambda$, there exists an output $y_k\in \{0,1\}^m$ such that $\Pr[G(k)\in \{y_k,\bot\}]\geq 1-\negl[\lambda]$. Therefore, if $\Pr\left[G(k)=\bot\right]< \frac{2}{3}$, then $\Pr\left[G(k)=y_k\right]> \frac{1}{3}-\negl[\lambda]$. Hence, it is clear that for $k\leftarrow \textsf{QSamp}(1^\lambda)$, $\textsf{PRG}(k)$ outputs $y_k$, except with negligible probability. To sum up, 
 \begin{align*}
\Pr_{k\leftarrow \textsf{QSamp}(1^\lambda)}\left[ \textsf{PRG}(k)=y_k\right] \geq 1-\negl[\lambda].
    \end{align*}

Next, we need to show that the security condition is satisfied. In other words, we need to show that for any QPT distinguisher $\adv$
   \begin{align*}
            \left|  \Pr_{{k}\leftarrow \textsf{QSamp}(1^\lambda)} \left[\adv (\textsf{PRG}(k))=1\right]-\Pr_{y\leftarrow \{0,1\}^m} \left[\adv (y)=1\right]\right| \leq 1/\poly[\lambda].
        \end{align*}

We commence with a hybrid argument.

\begin{itemize}
    \item Hybrid $\Hy_0$: This is the output distribution of the generator.
\begin{itemize}
    \item $k\gets \textsf{QSamp}(1^\lambda)$.
    \item For each $j\in [\lambda]$, compute $y_j\leftarrow G(k)$. 
    \item If $y_j=\bot$ for all $j\in [\lambda]$, then output $\bot$.
    \item Otherwise, output the first most common non-$\bot$ value in $(y_1,\ldots, y_\lambda)$.        
    \end{itemize}
    
    \item Hybrid $\Hy_1$: The same as hybrid $\Hy_0$ except the input is sampled from the good set $\mathcal{G}_\lambda$ of $G$. 
\begin{itemize}
    \item $k\gets \mathcal{G}_\lambda$.
    \item For each $j\in [\lambda]$, compute $y_j\leftarrow G(k)$. 
    \item If $y_j=\bot$ for all $j\in [\lambda]$, then output $\bot$.
    \item Otherwise, output the first most common non-$\bot$ value in $(y_1,\ldots, y_\lambda)$.        
    \end{itemize}

    \item Hybrid $\Hy_2$: The same as hybrid $\Hy_1$ except the output is computed using a single evaluation of $G(k)$. 
\begin{itemize}
    \item $k\gets \mathcal{G}_\lambda$.
    \item Output $G(k)$.        
    \end{itemize}
    
   \item Hybrid $\Hy_3$: The same as hybrid $\Hy_2$ except output is a $\bot$ random string. 
\begin{itemize}
    \item $k\gets \mathcal{G}_\lambda$.
    \item $y\gets \{0,1\}^{m} .$
    \item Output $\isbot(G(k),y)$.        
    \end{itemize}

       \item Hybrid $\Hy_4$: The output is a random string. 
\begin{itemize}
    \item $y\gets \{0,1\}^{m} .$
    \item Output $y$.        
    \end{itemize}
   \end{itemize}

We now show that no QPT adversary can distinguish these hybrids, except with inverse polynomial advantage. 

\begin{Claim}
For any QPT adversary $\adv$:
\begin{align*}
    \left| \Pr_{y\leftarrow \Hy_0}\left[\adv(y)=1\right]-\Pr_{y\leftarrow \Hy_1}\left[\adv(y)=1\right]\right| \leq \mu+\negl[\lambda].
\end{align*}
\end{Claim}

\begin{proof}
Recall the algorithm of $\textsf{QSamp}(1^\lambda)$ is as follows:

\smallskip \noindent\fbox{%
    \parbox{\textwidth}{%
$\textsf{QSamp}(1^\lambda)$:
For $i\in [\lambda]:$
        \begin{enumerate}
            \item Sample $k_i\leftarrow \{0,1\}^\lambda$.
            \item For each $j\in [\lambda]$, compute $y_{i,j}\leftarrow G(k_i)$
            \item If $\textsf{vote}_\lambda(y_{i,1},\ldots, y_{i,\lambda})\neq \perp$, then output $k_i$.
        \end{enumerate}
        Otherwise, output $\bot.$
    \smallskip  
    }}

Note that if $k_1\in \mathcal{G}_\lambda$ in the algorithm of $\textsf{QSamp}$, then the output is $k_1$, except with negligible probability, by the correctness condition of $\botPRG$. In other words, if $k_1\in \mathcal{G}_\lambda$, then the output of $\textsf{QSamp}$ is statistically indistinguishable from sampling a random element from $\mathcal{G}_\lambda$. Note that $\Pr_{k\leftarrow \{0,1\}^\lambda}[k\in \mathcal{G}_{\lambda}]\geq 1-\mu$, hence the probability $k_1\notin \mathcal{G}_\lambda$ is at most $\mu$. Therefore, we have:
\begin{align*}
    \left| \Pr_{y\leftarrow \Hy_0}\left[\adv(y)=1\right]-\Pr_{y\leftarrow \Hy_1}\left[\adv(y)=1\right]\right| \leq \mu+\negl[\lambda].
\end{align*}    
\qed
        \end{proof}

\begin{Claim}
    Hybrids $\Hy_1$ and $\Hy_2$ are statistically indistinguishable. 
\end{Claim}

\begin{proof}
    In both hybrids, the first step is to sample a input from the good set. By the pseudodeterminism of $\botPRG$s, for any input $k\in \mathcal{G}_\lambda$, there is a string $y$ satisfying $\Pr\left[y\leftarrow G(k)\right]\geq 1-\negl[\lambda]$. In this case, the probability that $\textsf{vote}_\lambda(y_1,\ldots,y_\lambda)=y$, where $y_i\leftarrow G(k)$ for $i\in [\lambda]$, is at least $1-\negl[\lambda]$. Hence, both hybrids are statistically indistinguishable. 
    \qed
\end{proof}

\begin{Claim}
For any QPT adversary $\adv$:
\begin{align*}
    \left| \Pr_{y\leftarrow \Hy_2}\left[\adv(y)=1\right]-\Pr_{y\leftarrow \Hy_3}\left[\adv(y)=1\right]\right| \leq 2\mu+\negl[\lambda].
\end{align*}
\end{Claim}

\begin{proof}
By the pseudodeterminism property of $G$, $\Pr_{k\leftarrow \{0,1\}^\lambda}[k\in \mathcal{G}_{\lambda}]\geq 1-\mu$ so 
\begin{align*}
    \left| \Pr_{y\leftarrow \Hy_2}\left[\adv(y)=1\right]-\Pr_{k\leftarrow \{0,1\}^\lambda}\left[\adv(G(k))=1\right]\right| \leq \mu+\negl[\lambda].
\end{align*}

Furthermore, by the security of $G$, there is negligible function such that 
\begin{align*}
    \left| \Pr_{k\leftarrow \{0,1\}^\lambda}\left[\adv(G(k))=1\right]-\Pr_{\begin{subarray}{c} k\leftarrow \{0,1\}^\lambda\\ y\leftarrow \{0,1\}^m\end{subarray}}\left[\adv(\isbot(y,G(k)))=1\right]\right| \leq \negl[\lambda].
\end{align*}

Finally, by the pseudodeterminism property of $G$, 
\begin{align*}
        \left| \Pr_{\begin{subarray}{c} k\leftarrow \{0,1\}^\lambda\\ y\leftarrow \{0,1\}^m\end{subarray}}\left[\adv(\isbot(y,G(k)))=1\right] - \Pr_{\begin{subarray}{c}k\leftarrow \mathcal{G}_\lambda\\ y\leftarrow \{0,1\}^m\end{subarray}}\left[\adv(\isbot(y,G(k)))=1\right]\right|  \leq \mu+\negl[\lambda].
\end{align*}

The second term on the left hand side of the equation above is hybrid $\Hy_3$. All in all, the triangle inequality gives:
\begin{align*}
    \left| \Pr_{y\leftarrow \Hy_2}\left[\adv(y)=1\right]-\Pr_{y\leftarrow \Hy_3}\left[\adv(y)=1\right]\right| \leq 2\mu+\negl[\lambda].
\end{align*}
    \qed
\end{proof}

\begin{Claim}
    Hybrid $\Hy_3$ is statistically indistinguishable from $\Hy_4$.
\end{Claim}
\begin{proof}
    This follows directly from the pseudodeterminism property of $\botPRG$s. 
    \qed
\end{proof}

By the triangle inequality, we get that for any QPT adversary $\adv$ such that, 
\begin{align*}
    \left|  \Pr_{{k}\leftarrow \textsf{QSamp}(1^\lambda)} \left[\adv (\textsf{PRG}(k))=1\right]-\Pr_{y\leftarrow \{0,1\}^m} \left[\adv (y)=1\right]\right|&= \\  \left| \Pr_{y\leftarrow \Hy_0}\left[\adv(y)=1\right]-\Pr_{y\leftarrow \Hy_4}\left[\adv(y)=1\right]\right| &\leq 3\mu+\negl[\lambda].
\end{align*}
Hence, Construction \ref{con:qkprg-from-botprg} is a weak $\QKPRG$. 
\qed
\end{proof}

Next, it was shown in \cite{DIJ+09} that a weak $\PRG$ can be upgraded to a strong $\PRG$ through a standard amplification argument, increasing the input length from $\lambda$ to $\lambda^2$. This argument applies to $\QKPRG$s, giving the following result.

\begin{theorem}
    \label{thm:prg-from-botprg}
         If there exists $(\mu,m)$-$\botPRG$ where $m>\lambda^2$ and $\mu=O(\lambda^{-c})$ for some constant $c>0$, then there exists a
    $(\lambda^2,m)$-$\QKPRG$ satisfying strong security.
\end{theorem}

\subsection{\texorpdfstring{\textsf{PRG}\textsuperscript{qs} from \textsf{BC}-\textsf{SPRS}\textsuperscript{qs}}{PRGqs from BC-SPRSqs}}
\label{sec:prg-from-sprs}

In this part, we show that $\BQSPRS$ can be used to construct $\QKPRG$s using properties of pseudorandom states described in \cref{sec:pr-from-sprs}.

\begin{construct}
    \label{con:prg-from-prs}
    Let $\lambda\in \mathbb{N}$ be the security parameter and let $ n=c\cdot \log\lambda$, $d=\lceil \lambda^c\rceil$, and $m\coloneqq \lceil \lambda^{c/12}\rceil$, where $c>24$ is a constant. Let $(\textsf{QS},\textsf{PRS})$ be a $(500\lambda d^8, \lambda,n)$-$\BQSPRS$. The construction for a weak $(\lambda, m)$-$\QKPRG$ is as follows:
    \begin{itemize}
        \item $\textsf{QSamp}(1^\lambda):$
        For each $i\in [\lambda]:$
        \begin{enumerate}
            \item Sample $s_i\leftarrow \textsf{QS}(1^\lambda)$. 
            \item Extract $t \coloneqq 144\lambda d^8$ copies of pseudorandom state $\rho_i^{\otimes t}$ using $\textsf{PRS}(s_i)$.
            \item Run $M_i\leftarrow \textsf{Tomography}(\rho_i^{\otimes t})$.
            \item If $M_i\in \mathcal{G}_\lambda$ (as defined in \cref{claim 3.7}), then output $k\coloneqq s_i$. 
        \end{enumerate}
        Otherwise, output $\bot$.
        
        \item $G(k):$ 
        \begin{enumerate}
            \item If $k=\bot$, then output $\bot$. 
            \item Extract $t \coloneqq 144\lambda d^8$ copies of pseudorandom state $\rho^{\otimes t}$ using $\textsf{PRS}(k)$.
            \item Compute $y\leftarrow \textsf{Extract}(\rho^{\otimes t})$. 
            \item Output $y$. 
        \end{enumerate}
    \end{itemize}
\end{construct}

\begin{lemma}
\label{lem:prg-from-prs}
    Construction \ref{con:prg-from-prs} is a weak $(\lambda, m)$-$\QKPRG$ assuming the existence of a $(500\lambda d^8, \lambda,n)$-$\BQSPRS$. 
\end{lemma}

\begin{proof}
We first show that the determinism condition of a $\QKPRG$ is satisfied (see \cref{def:PRG_qs}). 

Assume for contradiction that the determinism condition is not satisfied. Then, there is a non-negligible probability such that sampling a key $k\leftarrow \textsf{QSamp}(1^\lambda)$ and evaluating $G(k)$ twice yields distinct values. We obtain a contradiction as follows. 

By \cref{claim 3.7}, there is negligible probability that $\textsf{QSamp}(1^\lambda)$ outputs $\bot$. The other possibility is that $\textsf{QSamp}(1^\lambda)$ outputs a key $k=s$ satisfying $M_s\in \mathcal{G}_\lambda$ where $M_s\leftarrow \textsf{Tomography}(\rho^{\otimes t})$ and $\rho\leftarrow \textsf{PRS}(s)$. Therefore, with non-negligible probability, $k=s$ satisfies $M_s\in \mathcal{G}_\lambda$ but two evaluations of $G(k)$ yield distinct values. 



Then, we construct a distinguisher $\mathcal{D}$ that breaks the security of $\BQSPRS$ as follows. $\mathcal{D}$ receives $500\lambda d^8>3t$ copies of a state $\rho$ and needs to distinguish whether $\rho\leftarrow \textsf{PRS}(k)$ for $k\leftarrow \textsf{QS}(1^\lambda)$ or $\rho\leftarrow \textsf{Haar}(\mathbb{C}^n)$ is a Haar-random state. $\mathcal{D}$ computes  $M\leftarrow \textsf{Tomography}(\rho^{\otimes t})$ with the first $t$ copies, $y_1\leftarrow \textsf{Extract}(\rho^{\otimes t})$ with the second $t$ copies, and $y_2\leftarrow \textsf{Extract}(\rho^{\otimes t})$ with the last $t$ copies. If $y_1\neq y_2$ and $M\in \mathcal{G}_\lambda$, then $\mathcal{D}$ guesses that $\rho$ is a $\BQSPRS$ i.e. outputs 1. Otherwise, $\mathcal{D}$ outputs 0. By our assumption, $y_1\neq y_2$ and $M\in \mathcal{G}_\lambda$ occurs with non-negligible probability if $\rho$ is generated using $(\textsf{QS},\textsf{PRS})$. On the other hand, \cref{claim 3.7} states that this occurs with negligible probability when $\rho$ is a Haar-random state. It is clear that $\mathcal{D}$ has non-negligible advantage in distinguishing $\BQSPRS$ from Haar-random states using only $3t$ copies of the state, which contradicts the security condition of $\BQSPRS$. Therefore, $(\textsf{QSamp},G)$ satisfies the correctness condition of a $\QKPRG$. 

Next, we need to show security. Specifically, we need to show that for any QPT distinguisher $\adv$:
        \begin{align*}
            \left|  \Pr_{{k}\leftarrow \textsf{QSamp}(1^\lambda)} \left[\adv (G(k))=1\right]-\Pr_{y\leftarrow \{0,1\}^{m} }\left[\adv (y)=1\right]\right| \leq O(1/d^{1/6}).
        \end{align*}

The proof is through a hybrid argument:
\begin{itemize}
    \item $\Hy_0$: 
    \begin{enumerate}
        \item Sample ${k}\leftarrow \textsf{QSamp}(1^\lambda)$.
        \item Compute $ y\leftarrow G(k)$.
        \item Run $\adv (y)$ and output the result. 
    \end{enumerate}  
    \item $\Hy_1$: The same as $\Hy_0$, except we modify the $\QKPRG$ algorithms by extracting a classical string from the pseudorandom states during input sampling instead of during evaluation. 
    \begin{itemize}
        \item $\textsf{QSamp}_{\Hy_1}(1^\lambda):$
        For each $i\in [\lambda]:$
        \begin{enumerate}
            \item Sample $s_i\leftarrow \textsf{QS}(1^\lambda)$. 
            \item Extract $t \coloneqq 144\lambda d^8$ copies of pseudorandom state $\rho_i^{\otimes t}$ using $\textsf{PRS}(s_i)$.
            \item Run $M_i\leftarrow \textsf{Tomography}(\rho_i^{\otimes t})$.
            \item If $M_i\in \mathcal{G}_\lambda$, then compute $y_i\leftarrow \textsf{Round}(M_i)$. 
            \item Output $k\coloneqq y_i$. 
        \end{enumerate}
        Otherwise, output $\bot$.
        
        \item $G_{\Hy_1}(k):$ 
        Output $k$. 
    \end{itemize}
    
    \item $\Hy_2$: Same as $\Hy_1$, except we modify the $\QKPRG$ algorithms by using random Haar states instead of pseudorandom states.
     \begin{itemize}
        \item $\textsf{QSamp}_{\Hy_2}(1^\lambda):$
        For each $i\in [\lambda]:$
        \begin{enumerate}
            \item Sample $t \coloneqq 144\lambda d^8$ copies of a random Haar state $\rho_i^{\otimes t}$ from $\textsf{Haar}(\mathbb{C}^d)$.
            \item Run $M_i\leftarrow \textsf{Tomography}(\rho_i^{\otimes t})$.
     \item If $M_i\in \mathcal{G}_\lambda$, then compute $y_i\leftarrow \textsf{Round}(M_i)$. 
            \item Output $k\coloneqq y_i$. 
            \end{enumerate}
        Otherwise, output $\bot$.
        
        \item $G_{\Hy_2}(k):$ Output $k$.
    \end{itemize}

    \item $\Hy_3$: Same as $\Hy_2$, except we modify the $\QKPRG$ algorithms by removing the condition on the random Haar states.
     \begin{itemize}
        \item $\textsf{QSamp}_{\Hy_3}(1^\lambda):$
        \begin{enumerate}
            \item Sample $t \coloneqq 144\lambda d^8$ copies of a random Haar state $\rho^{\otimes t}$ from $\textsf{Haar}(\mathbb{C}^d)$.
            \item Run $M\leftarrow \textsf{Tomography}(\rho^{\otimes t})$.
     \item Compute $y\leftarrow \textsf{Round}(M)$. 
            \item Output $k\coloneqq y$. 
            \end{enumerate}
        
        \item $G_{\Hy_3}(k):$ Output $k$.
    \end{itemize}
    
    \item $\Hy_4$: Same as $\Hy_3$, except we sample a random string instead of extracting randomness from the Haar random states.
     \begin{itemize}
        \item $\textsf{QSamp}_{\Hy_4}(1^\lambda):$
        Sample a random string $y\leftarrow \{0,1\}^{ m}$. Set $k=y$.

        \item $G_{\Hy_4}(k):$ 
        Output $k$. 
    \end{itemize}
\end{itemize}
   
First of all, hybrid $\Hy_0$ is statistically indistinguishable from hybrid $\Hy_1$, since we just move a step from the evaluation phase to the input sampling phase.

$\Hy_1$ is computationally indistinguishable from hybrid $\Hy_2$ because any QPT adversary that can distinguish between these hybrids can be used to break $\BQSPRS$ security. 

Next, if $M\in \mathcal{G}_\lambda$ in $\textsf{QSamp}_{\Hy_3}$ in hybrid $\Hy_3$ and $M_1\in \mathcal{G}_\lambda$ in $\textsf{QSamp}_{\Hy_2}(1^\lambda)$ in $\Hy_2$, then it is clear that $\textsf{H}_2$ and $\textsf{H}_3$ are statistically indistinguishable. This scenario occurs with $O(d^{-1/6})$ probability by \cref{claim 3.7}.

Finally, the variational distance between hybrids $\Hy_3$ and $\Hy_4$ is at most $O(d^{-1/6})$ by \cref{lem:uniform}. 

All in all, by the triangle inequality:
\begin{align*}
&\left|   \Pr_{{k}\leftarrow \textsf{QSamp}(1^\lambda)} \left[\adv (G(k))=1\right]-\Pr_{y\leftarrow \{0,1\}^{m} }\left[\adv (y)=1\right]\right| \\
&= \left|  \Pr_{{y}\leftarrow \Hy_0} \left[\adv (y)=1\right]-\Pr_{y\leftarrow \Hy_4 }\left[\adv (y)=1\right]\right|\\
&\leq O(d^{-1/6}).
        \end{align*}
\qed
\end{proof}

Recall that weak $\QKPRG$ can be upgraded to a strong $\QKPRG$ following the same argument in \cite{DIJ+09}. Hence, we obtain the following result.

\begin{theorem}
    \label{thm:prg-from-sprs}
         If there exists $(500\lambda d, \lambda,n)$-$\BQSPRS$, then there exists a
    $(\lambda^2,m)$-$\QKPRG$ satisfying strong security.
\end{theorem}


\subsection{\texorpdfstring{\textsf{SPRS}\textsuperscript{qs} from \textsf{PRG}\textsuperscript{qs}}{SPRSqs from PRGqs}}
\label{sec:prs-from-prg}

In this section, we show that $\QKPRG$s can be used to build $\QKSPRS$.

\begin{construct}
    \label{con:short-prs-from-bot-prf}
    Let $\lambda\in \mathbb{N}$ be the security parameter and let $c>12$ be a constant. Let $m=m(\lambda)$ be polynomial on $\lambda$ such that $m> \lambda^{2c+1}$. Let $(\textsf{QS},G)$ be a $( \lambda, m)$-$\QKPRG$. Let $N=\lceil \lambda^{c}\rceil$ and $\mathcal{X}=\{1,\ldots, N \}$. The construction for a $(\lambda,c\cdot \log\lambda)$-$\QKSPRS$ is as follows:

    \begin{itemize}
        \item $\textsf{QSamp}(1^\lambda):$
        Sample $s\leftarrow \textsf{QS}(1^\lambda)$. Output $k=s$.
        
        \item $\textsf{StateGen}(k):$ 
        \begin{enumerate}
            \item Compute $y\leftarrow G(k)$.
            \item Interpret $y$ as a function $f_{y}:\mathcal{X}\rightarrow \mathcal{X}$ \footnote{For $i\in \mathcal{X}$, $f_z(i)\coloneqq z[it:(i+1)t]$ where $t\coloneqq \lceil \log N\rceil$.}. 
            \item Output 
            \begin{align*}
                \ket{\psi_k}\coloneqq \frac{1}{\sqrt{N}}\sum_{x\in \mathcal{X}}\omega_N^{f_{y}(x)}\ket{x}.
            \end{align*}
        \end{enumerate}
    \end{itemize}
\end{construct}

\begin{theorem}
\label{thm:short-prs-from-prg}
    Construction \ref{con:short-prs-from-bot-prf} is a $(\lambda,c\cdot \log\lambda)$-$\QKSPRS$ assuming the existence of a $( \lambda, m)$-$\QKPRG$ where $m>\lambda^{2c+1}$. 
\end{theorem}

\cref{thm:short-prs-from-prg} follows directly from \cite{JLS18}. Specifically, \cite{JLS18} shows that a $n$-$\PRS$ can be built from a $\PRF$ with domain size $2^n$. This conversion applies to the quantum input sampling regime as well. Then, \cref{thm:short-prs-from-prg} follows by setting $n\coloneqq c\cdot \log \lambda$ and noting that $f_y$ in \cref{con:short-prs-from-bot-prf} is computationally indistinguishable from a random function by the security of $\QKPRG$s. 



\cref{thm:prg-from-sprs} states that we can construct a $\QKPRG$ from any $\QKSPRS$. On the other hand, \cref{thm:short-prs-from-prg} state that we can construct a $\QKSPRS$ from a $\QKPRG$. Applying these two results consecutively, we get a method to shrink the size of a $\QKSPRS$.  

\begin{corollary}
For any constant $c>36$, and any constant $0<m<c/36$, $(\lambda,c\log \lambda)$-$\QKSPRS$ implies $(\lambda,m\log \lambda)$-$\QKSPRS$.
\end{corollary}

Furthermore, by starting with a $\SPRS$, building a $\botPRG$ (\cref{lem:bot-prg-from-prs}), amplifying the output length sufficiently \footnote{It is easy to amplify the output length of a $\botPRG$ by re-applying the algorithm on the output, at the cost of increasing the pseudodeterminism error. }, building a $\QKPRG$ (\cref{thm:prg-from-botprg}), and finally a $\QKSPRS$ (\cref{thm:short-prs-from-prg}), we obtain a $\QKSPRS$ of larger size. 

\begin{corollary}
For any constant $c>36$, and any constant $m>c$, $(\lambda,c\log \lambda)$-$\SPRS$ implies $(\lambda,m\log \lambda)$-$\QKSPRS$.
\end{corollary}

We also show how to build bounded-query $\QKPRU$ from $\QKPRG$ in \cref{sec:pru-from-prg}. 

\section{Separations}
\subsection{\texorpdfstring{Separating $\PRG$ from $\QKPRF$}{Separating PRG from PRFqs}}
\label{sec:implications 3}

We show a distinction between uniform input sampling and quantum input sampling by demonstrating that there do not exist fully black-box constructions of a $\PRG$ from a $\QKPRF$ with inverse access.


We only state the result here and give the proof in \cref{sec:separation 3}. 

\begin{theorem}
\label{thm:separation 3}
There does not exist a fully black-box construction of a $\PRG$ from a (quantum-query-secure) $\QKPRF$ with inverse access. 
\end{theorem}

The following lemma generalizes the applications of $\PRF$s in cryptography \cite{JLS18,ALY23,BBO+24} to $\QKPRF$s. These follow in the same way as using quantum key generation rather than uniform key generation does not affect the proofs. 

\begin{lemma}
\label{thm:implications 2}
    There exists fully black-box constructions of the following primitives from (quantum-query-secure) $\QKPRF$s:
    \begin{enumerate}
    \item $\QKPRG$, $\QKSPRS$, $\QKLPRS$, and $\QKPRU$. 
 \item Statistically-binding computationally hiding bit commitments with classical communication. 
\item Message authentication codes of classical messages with classical communication. 
        \item CCA2-secure symmetric encryption for classical messages with classical keys and ciphertexts. 
        \item \textsf{EV-OWPuzz}s. 
    \end{enumerate}
\end{lemma}

As a result, of the applications of $\QKPRF$, we obtain the following corollary. 

\begin{corollary}
\label{cor:separation3}
There is no fully black-box construction for a $\PRG$ from the following primitives with inverse access:
   \begin{enumerate}
    \item $\QKPRG$, $\QKSPRS$, $\QKLPRS$, and $\QKPRU$. 
 \item Statistically-binding computationally hiding bit commitments with classical communication. 
\item Message authentication codes of classical messages with classical communication. 
        \item CCA2-secure symmetric encryption of classical messages with classical keys and ciphertexts. 
        \item \textsf{EV-OWPuzz}s. 
    \end{enumerate}
\end{corollary}

\subsubsection{Separation Proof.}
\label{sec:separation 3}

The idea of the separation proof is to consider three oracles. The first is a \textsf{PSPACE} oracle. The second is a restricted-access random oracle $\mathcal{O}$, that can only be accessed with a key as input. This oracle will acts as a $\QKPRF$. The third oracle $\sigma$ produces random keys such that each distinct key gives access to a different function in $\mathcal{O}$. Note that this is not an issue for the $\QKPRF$, as a key from $\sigma$ can be sampled during key generation and, then, reused during evaluation to obtain deterministic outputs from $\mathcal{O}$. 

On the other hand, a uniform key generation algorithm is incapable of accessing $\sigma$ during key generation. Consequently, a $\PRG$ is incapable of obtaining deterministic evaluations from $\mathcal{O}$. This informally means that any $\PRG$ cannot depend on $(\sigma,\mathcal{O})$ and, thus, cannot exist in the presence of a $\textsf{PSPACE}$ oracle. This implies that there does not exist a fully black-box construction of a $\PRG$ from a $\QKPRF$. 

We first introduce some preliminary information. We present $s$-$\PRG$ security relative to an oracle $\mathcal{T}$ in the form of an experiment to simplify notation later on. 

\smallskip \noindent\fbox{%
    \parbox{\textwidth}{%
$\textsf{Exp}^{\textsf{PRG}}_{\adv^\mathcal{T},G^\mathcal{T}}(1^\lambda)$:
\begin{enumerate}
    \item Sample $k\leftarrow \{0,1\}^\lambda$ and $b\leftarrow \{0,1\}$. 
    \item If $b=0$, generate $y\leftarrow G^\mathcal{T}(k)$. Else, sample $y\leftarrow \{0,1\}^s$.
    \item $b'\leftarrow \adv^\mathcal{T}(y)$.
\item If $b'=b$, output $1$. Otherwise, output $0$. 
\end{enumerate}}}
    \smallskip  

We define the advantage of the adversary in this experiment as follows. It is clear that $\PRG$ security implies that this advantage should be negligible.   
\begin{align*}            \textsf{Adtg}^{\textsf{PRG}}_{\adv^\mathcal{T},G^\mathcal{T}}(1^\lambda)\coloneqq \Pr\left[\textsf{Exp}^{\textsf{PRG}}_{\adv^\mathcal{T},G^\mathcal{T}}(1^\lambda)=1\right]-\frac{1}{2}.
        \end{align*}

Let $U_{\ket{\psi}}$ be a unitary that flips $\ket{0^n}$ with $\ket{\psi}$ and acts as the identity on all other orthogonal states. We will use the following result from \cite{CCS24} regarding the simulation of this oracle with polynomial copies of $\ket{\psi}$.

\begin{theorem}[Theorem 2.6 in \cite{CCS24}]
\label{thm flip oracle access}
    Let $T,n\in \mathbb{N}$, $\epsilon>0$, and $\ket{\psi}$ be a real-valued $n$-qubit state orthogonal to $\ket{0^n}$. Let $U_{\ket{\psi}}$ be the unitary defined above. For any oracle algorithm $\adv$ making $T$ queries to $U_{\ket{\psi}}$, there exists an algorithm $\tilde{\adv}$ with  access to $O(\frac{T^2}{\epsilon^2})$ copies of $\ket{\psi}$ that outputs a state that is $\epsilon$-close in terms of trace distance. 
\end{theorem}

Note that Theorem 2.6 in \cite{CCS24} is more involved than \cref{thm flip oracle access} since it deals with \emph{complex-valued} quantum states. Our version is limited to real-valued states. 

Also, note that \cite{Z12} shows that any computationally unbounded adversaries cannot distinguish quantum-query-access to a a randomly sampled polynomial of degree $(2r-1)$ and a random function given only $r$ quantum queries. Given that this result holds for computationally unbounded adversaries, it holds for adversaries with access to a \textsf{PSPACE}-oracle. We now describe the oracles used in the separation. 

\begin{construct}
\label{con:oracles 3}
    For $n\in \mathbb{N}$, let $O_n\leftarrow \Pi_{n,8n}$ and $P_n\leftarrow \Pi_{2n,n}$ be random functions. Let $\mathcal{T} \coloneqq (\sigma,\mathcal{O}, \mathcal{C})$ be a tuple of oracles, where $\sigma=\{\sigma_n\}_{n\in \mathbb{N}}$, $\mathcal{O}\coloneqq \{\mathcal{O}_n\}_{n\in \mathbb{N}}$, and $\mathcal{C}\coloneqq \{\mathcal{C}_n\}_{n\in \mathbb{N}}$ are defined as follows:  
\begin{enumerate}
    \item $\mathcal{C}$ is for membership in a \textsf{PSPACE}-complete language. 
    \item $\sigma_n$: Unitary that swaps $\ket{0^{9n+1}}$ with the state $\ket{\psi_n}\coloneqq\frac{1}{\sqrt{2^n}}\ket{1}\sum_{x\in \{0,1\}^n}\ket{x}\ket{O_n(x)}$ and acts as the identity on all other orthogonal states. 
    \item $\mathcal{O}_n$: Unitary of the classical function that maps $(x,y,a)$, where $x,a\in \{0,1\}^n$ and $y\in \{0,1\}^{8n}$, to $P_n(x,a)$ if $O_n(x)=y$ and to $\bot$ otherwise. 
\end{enumerate}
\end{construct}

Notice that the unitary oracles above are self-inverse, so our separation is relative to a unitary oracle with access to the inverse. 

We introduce some notation for the proof. Let $\mathbb{T}$ denote the set of all possible oracles and let $\mathcal{T}\leftarrow \mathbb{T}$ denote sampling an oracle in the way described in \cref{con:oracles 3}. For any oracle $\mathcal{T}$ and integer $m\in \mathbb{N}$, let $\mathcal{T}_{\leq m}$ denote the sequence of oracles $(\sigma_n,\mathcal{O}_n)_{n\leq m}$ and let $\mathbb{T}[{\mathcal{T}_{\leq m}}]$ denote the set
\begin{align*}
\mathbb{T}[{\mathcal{T}_{\leq m}}]\coloneqq \{\tilde{\mathcal{T}}\in \mathbb{T}:\tilde{\mathcal{T}}_{\leq m}=\mathcal{T}_{\leq m}\}.    
\end{align*}

\begin{theorem}
There does not exist a fully black-box construction of a $\PRG$ from a (quantum-query-secure) $\QKPRF$ with inverse access. 
\end{theorem}

\begin{proof}
For simplicity, we only prove the theorem for $(9n,n)$-$\QKPRF$ i.e. for $\QKPRF$s with $9n$-bit keys and $n$-bit inputs. However, the proof easily generalizes to other parameters by modifying the parameters of the oracles. 

Assume, for the purpose of obtaining a contradiction, that $\tilde{G}^{F,F^{-1}}$ is a black-box construction of a $\PRG$, from any $(9n,n)$-$\QKPRF$ $F$. First, the following result states that there exists a $\QKPRF$ relative to $\mathcal{T}$.

\begin{lemma}
\label{lem:qkprf 2}
    There exists a (quantum-query-secure) $(9n,n)$-$\QKPRF$ relative to $\mathcal{T}$ for any $O$ and with probability 1 over the distribution of $P$. Furthermore, correctness is satisfied for any oracle $\mathcal{T}$.
\end{lemma}

\begin{proof}
\begin{construct}
\label{con:prf}
The algorithms of a $\QKPRF$ with oracle access to $\mathcal{T}$ is as follows:
    \begin{itemize}
    \item $\textsf{QSamp}^\mathcal{T}(1^n):$ 
    \begin{enumerate}
        \item Query $\sigma_n(\ket{0^{9n+1}})$ and measure the result in the computational basis to obtain $(x,y)$, where $x\in \{0,1\}^{n}$ and $y\in \{0,1\}^{8n}$.
        \item Output $k\coloneqq(x,y)$.
    \end{enumerate} 
    \item $F_k^\mathcal{T}(a):$ 
    Interpret $k$ as $(x,y)$. Output $\mathcal{O}_n(x,y,a)$ \footnote{This is a slight abuse of notation, as $\mathcal{O}$ is a unitary but we interpret it as a classical function here. }. 
\end{itemize}
\end{construct}

It is clear that $(\textsf{QSamp}^\mathcal{T},F^\mathcal{T})$ satisfies the correctness condition of $\QKPRF$ for any oracle $\mathcal{T}$. 

For security, note that for any $(x,y)\leftarrow \textsf{QSamp}^\mathcal{T}(1^n)$, $F^\mathcal{T}_{(x,y)}(\cdot)= P_n(x,\cdot)$, where $P_n(x,\cdot)$ is a random function independent of $\sigma$. Lemma 2.2 from \cite{SXY18} states that a random oracle acts as a $\PRF$ i.e. for any QPT adversary $\adv$ that makes  $p(n)$ oracle queries:
\begin{align*}
   \underset{P_n\leftarrow \Pi_{2n,n}}{\mathbb{E}} \left[ \left| \Pr_{x\leftarrow \{0,1\}^n}\left[\adv^{{P_n},{P_n(x,\cdot)}}(1^n)=1\right]- \Pr_{\tilde{P}_n\leftarrow \Pi_{n,n}}\left[\adv^{{P_n}, {\tilde{P}_n}}(1^n)=1\right]\right| \right]\leq \frac{2p(n)}{2^n}< \frac{1}{2^{n/4}}.
\end{align*}

Note that this result even holds against unbounded-time adversaries as long as the number of queries to the oracle is polynomial. Hence, this result also holds against QPT adversaries with access to a \textsf{PSPACE}-oracle:
\begin{align*}
   \underset{
       P_n\leftarrow \Pi_{2n,n} 
 }{\mathbb{E}}\left[ \left| \Pr_{x\leftarrow \{0,1\}^n} \left[\adv^{{P_n},{P_n(x,\cdot)},\mathcal{C}}(1^n)=1\right]-    \Pr_{\tilde{P}_n\leftarrow \Pi_{n,n}}\left[\adv^{{P_n}, {\tilde{P}_n},\mathcal{C}}(1^n)=1\right]\right| \right]\leq \frac{1}{2^{n/4}}.
\end{align*}

 By Markov inequality, we get that
\begin{align*}
      \Pr_{
          P_n\leftarrow \Pi_{2n,n}}  \left[ \left| \Pr_{x\leftarrow \{0,1\}^n}\right.\right. & \left[\adv^{{P_n},{P_n(x,\cdot)},\mathcal{C}}(1^n)=1\right]- \\ &\left. \left.   \Pr_{\tilde{P}_n\leftarrow \Pi_{n,n}}\left[\adv^{{P_n}, {\tilde{P}_n},\mathcal{C}}(1^n)=1\right]\right| \geq {2^{-n/8}}\right] \leq 2^{-n/8}
    \end{align*}
  
By Borel-Cantelli Lemma, since $\sum_n 2^{-n/8}$ converges, with probability 1 over the distribution of $P$, it holds that
\begin{align*}
\left| \Pr_{x\leftarrow \{0,1\}^n}\left[\adv^{{P_n},{P_n(x,\cdot)},\mathcal{C}}(1^n)=1\right]-    \Pr_{\tilde{P}_n\leftarrow \Pi_{n,n}}\left[\adv^{{P_n}, {\tilde{P}_n},\mathcal{C}}(1^n)=1\right]\right| \leq {2^{-n/8}},
    \end{align*}

except for finitely many $n\in \mathbb{N}$. There are countable number of quantum algorithms $\adv$ making polynomial queries to $\mathcal{T}$, so this bound holds for every such adversary. 
\qed
\end{proof}


By \cref{lem:qkprf 2} and the assumed existence of a black-box construction $\tilde{G}$, there exists an algorithm $G^\mathcal{T}$ that is a $\PRG$ with probability 1 over the oracles $\mathcal{T}$ and satisfies correctness for any oracle $\mathcal{T}\in \mathbb{T}$. Let $s=s(\lambda)$ be a polynomial denoting the output length of $G$. 

\begin{Claim}
\label{claim:adv3}
    For any QPT adversary $\adv$:
\begin{align*}
    \Pr_{\mathcal{T}\leftarrow \mathbb{T}}\left[\textsf{Adtg}^{\PRG}_{\adv^\mathcal{T},G^\mathcal{T}}(1^\lambda)\leq O\left(\frac{1}{\lambda^4}\right)\right]\geq \frac{3}{4}.
\end{align*}
\end{Claim}

\begin{proof}
If this does not hold, then there exists a QPT adversary $\adv$ such that 
    \begin{align*}
\Pr_{\mathcal{T}\leftarrow \mathbb{T}}\left[\textsf{Adtg}^{\textsf{PRG}}_{\adv^{\mathcal{T}},G^\mathcal{T}}(1^\lambda)> \frac{1}{\lambda^5}\right]> \frac{1}{4}
\end{align*}
for infinitely many $\lambda\in \mathbb{N}$. 

By a variant of Borel-Cantelli Lemma (Lemma 2.9 in \cite{MMN+16}), this means $\adv$ is successful in breaking the security of $G^\mathcal{T}$ with probability $\frac{1}{4}$ over the oracle distribution. Therefore, $G^\mathcal{T}$ is not a $\PRG$ with probability $\frac{1}{4}$ over the oracle distribution, giving a contradiction.  
\qed
\end{proof}

Let $r=r(\lambda)$ denote the maximum run-time of $G$ and $m=m(\lambda)\coloneqq  10(r\lambda)^4+\lambda$. Hence, $G$ makes at most $r$ queries to the oracles. 

Fix an oracle $\mathcal{T}$. We will need to show the following lemma. 

\begin{lemma}
\label{claim 11} For large enough $\lambda$ and any $k\in \{0,1\}^\lambda$,
    \begin{align*}
           \Pr_{\mathcal{T}',\mathcal{T}''\leftarrow \mathbb{T}[\mathcal{T}_{\leq \log(2m)}]}\left[G^{\mathcal{T}'}(k)= G^{\mathcal{T}''}(k)\right]\geq 1-\frac{r}{m}.
    \end{align*} 
\end{lemma}

\begin{proof}
Let $\mathcal{T}'\coloneqq (\sigma',\mathcal{O}',\mathcal{C})$ and $     \mathcal{T}''\coloneqq  (\sigma'',\mathcal{O}'',\mathcal{C})$ be two oracles sampled from $\mathbb{T}[\mathcal{T}_{\leq \log(2m)}]$ and determined by the functions $(P',O')$ and $(P'',O'')$, respectively, as described in \cref{con:oracles 3}. Fix $k\in \{0,1\}^\lambda$. 
    
We will now describe how to construct a sequence of oracles $\mathcal{T}^i=({\sigma}^i,\mathcal{O}^i,\mathcal{C})$, starting with $\mathcal{T}'$ and ending with $\mathcal{T}''$, such that $G^{\mathcal{T}^i}(k)$ is invariant (to some degree) as we move along the sequence. Note that $\mathcal{T}'_n$ and $\mathcal{T}_n''$ only differ for $\log(2m)<n<r$. For such a value of $n$, let $x^*\in \{0,1\}^n$ be an arbitrary string. 

\begin{itemize}
    \item $\mathcal{T}^1$: This is the same as $\mathcal{T}'$. Define
      \begin{align*}
    \ket{\psi_n^1}\coloneqq \frac{1}{\sqrt{2^n}}\ket{1}\sum_{x\in \{0,1\}^n}\ket{x}\ket{O'_n(x)}.
    \end{align*}
    \item $\mathcal{T}^2:$ Same as $\mathcal{T}^1$, except we change $\sigma^{2}_n$ to be the unitary that swaps $\ket{0^{9n+1}}$ with 
    \begin{align*}
    \ket{\psi_n^2}\coloneqq \frac{1}{\sqrt{2^n-1}}\ket{1}\sum_{x\in \{0,1\}^n\setminus \{x^*\}}\ket{x}\ket{O'_n(x)}.
    \end{align*} 
    \item $\mathcal{T}^3:$ Same as $\mathcal{T}^2$, except for any $a\in \{0,1\}^n$, we set $\mathcal{O}^3_n(x^*,O'_n(x^*),a)$ to $\bot$ and, then, we set $\mathcal{O}^3_n(x^*,O''_n(x^*),a)$ to $P''_n(x^*,a)$. 
    \item $\mathcal{T}^4:$ Same as $\mathcal{T}^3$, except $\sigma^{4}_n$ is the unitary that swaps $\ket{0^{9n+1}}$ with \begin{align*}
        \ket{\psi_n^4}\coloneqq \frac{1}{\sqrt{2^n}}\left(\ket{1}\sum_{x\in \{0,1\}^n\setminus \{x^*\}}\ket{x}\ket{O_n(x)}+\ket{1}\ket{x^*}\ket{O_n''(x^*)}\right).
    \end{align*}
\end{itemize}

We perform these modifications for every input $x^*$ of length $n$ for every $\log(2m)< n<r$, one input at a time, until the oracle $\mathcal{T}'$ is completely replaced with $\mathcal{T}''$ on inputs of length less than $r$. What remains to show is that $G$ is invariant to some degree under these modifications. 

\begin{Claim} 
For large enough $\lambda,$ 
\label{claim 00}
    \begin{align*}
        d_{\textsf{TD}}\left(G^{\mathcal{T}^1}(k), G^{\mathcal{T}^2}(k)\right)\leq \frac{3}{\lambda}. 
    \end{align*}
\end{Claim}

\begin{proof}
In the oracles $\mathcal{T}^1$ and $\mathcal{T}^2$, only $\sigma^1_n$ and $\sigma^2_n$ differ. Both these oracles are unitaries that swap two states and, thus, have the structure required to apply \cref{thm flip oracle access}. Specifically, setting $T=r$ and $\epsilon=\lambda$ in \cref{thm flip oracle access}, we get that there exists an algorithm $\tilde{G}^{\mathcal{T}^1\setminus \sigma_n^1}(k,\ket{\psi^1_n}^{\otimes q})$ that simulates $G^{\mathcal{T}^1}(k)$ without oracle access to $\sigma^1_n$, given $q\coloneqq O(\frac{T^2}{\epsilon^2})=O(r^2\lambda^2)$ copies of $\ket{\psi^1_n}$, yielding an output that is $(\frac{1}{\lambda})$-close in trace distance. Similarly, $\tilde{G}^{\mathcal{T}^2\setminus \sigma_n^2}(k,\ket{\psi^2_n}^{\otimes q})$ simulates $G^{\mathcal{T}^2}(k)$ without oracle access to $\sigma^2_n$, given $q$ copies of $ \ket{\psi^2_n}$, yielding an output that is $(\frac{1}{\lambda})$-close in trace distance. 

Given that, for large enough $\lambda$, 
\begin{align*}
    d_{\textsf{TD}}\left(\ket{\psi_n^2}^{\otimes q},\ket{\psi^1_n}^{\otimes q}\right)\leq \frac{q}{m}<\frac{1}{\lambda}, 
\end{align*} 
and noting that $\mathcal{T}^1\setminus \sigma_n^1$ is equivalent to $\mathcal{T}^2\setminus \sigma_n^2$, we get 
\begin{align*}
    d_{\textsf{TD}}\left(\tilde{G}^{\mathcal{T}^1\setminus \sigma_n^1}\left(k,\ket{\psi^1_n}^{\otimes q}\right),\tilde{G}^{\mathcal{T}^2\setminus \sigma_n^2}\left(k,\ket{\psi^2_n}^{\otimes q}\right)\right)\leq  \frac{1}{\lambda}.
\end{align*}
Therefore, by the triangle inequality, we have that for large enough $\lambda$,
\begin{align*}
   d_{\textsf{TD}}\left(G^{\mathcal{T}^1}(k), G^{\mathcal{T}^2}(k)\right)\leq \frac{3}{\lambda}. 
\end{align*}
 \qed
\end{proof}

\begin{Claim} 
For large enough $\lambda,$
    \begin{align*}
        \Pr_{\mathcal{T}',\mathcal{T}''\leftarrow \mathbb{T}[\mathcal{T}_{\leq \log(2m)}]}\left[d_{\textsf{TD}}\left(G^{\mathcal{T}^2}(k), G^{\mathcal{T}^3}(k)\right)\geq {\frac{1}{2^{2n}}}\right]\leq {1/2^{2n}}. 
    \end{align*}
\end{Claim}

\begin{proof}
    Notice $\sigma^2=\sigma^3$ and $\mathcal{O}^2$ only differs from $\mathcal{O}^3$ on inputs starting with $(x^*,O'(x^*))$ or $(x^*,O''(x^*))$. Crucially, $O'(x^*)$ and $O''(x^*)$ are distributed uniformly at random and are of length $8n$. Therefore, to distinguish these two oracles, $G$ needs to do unstructured search for an input starting with $x^*$ that maps to a non-$\bot$ element. 
    
    The lower bound for unstructured search \cite{Z90,BBB+97} states that $G$ cannot distinguish these two oracles with better than $O(\frac{r^2}{2^{8n}}) \leq \frac{1}{2^{4n}}$ probability for large enough $\lambda$. This probability  is over the oracle distributions of $\mathcal{T}',\mathcal{T}''$ as well. The Markov inequality then gives for large enough $\lambda$:  
      \begin{align*}
        \Pr_{\mathcal{T}',\mathcal{T}''\leftarrow \mathbb{T}[\mathcal{T}_{\leq \log(2m)}]}\left[d_{\textsf{TD}}\left(G^{\mathcal{T}^2}(k), G^{\mathcal{T}^3}(k)\right)\geq {\frac{1}{2^{2n}}}\right]\leq {1/2^{2n}}. 
    \end{align*}
    \qed
\end{proof}

\begin{Claim}
For large enough $\lambda,$
    \begin{align*}
        d_{\textsf{TD}}\left(G^{\mathcal{T}^3}(k), G^{\mathcal{T}^4}(k)\right)\leq \frac{3}{\lambda}. 
    \end{align*}
\end{Claim}

\begin{proof}
 This follows in the same way as \cref{claim 00}. 
 \qed
\end{proof}

As a result of the past three claims and the triangle inequality, with probability $1-\frac{1}{2^{2n}}$ over the distributions $\mathcal{T}',\mathcal{T}''$, we have for large enough $\lambda$, \begin{align}
\label{eq 7}
    \Pr\left[G^{\mathcal{T}^1}(k)= G^{\mathcal{T}^4}(k)\right]\geq 1-\frac{7}{\lambda}.
\end{align}

We now argue that this inverse-polynomial difference must in fact be negligible. Notice that $\mathcal{T}^4\in \mathbb{T}$ so by \cref{lem:qkprf 2}, $G^{\mathcal{T}^4}$ satisfies the correctness condition of a $\PRG$. In other words, it is almost-deterministic, meaning it must output a fixed value except with negligible probability. Combining this with \cref{eq 7}, we get that, with probability $1-\frac{1}{2^{2n}}$ over the oracle distributions, $G^{\mathcal{T}^1}(k)$ and $ G^{\mathcal{T}^4}(k)$ agree except with negligible probability.  

Performing the same changes on every $\log(2m)<n<r$ and $x^*\in \{0,1\}^n$, we reach $\mathcal{T}''$. The union bound gives us that with probability 
\begin{align*}
    1-\sum_{\log(2m)<n<r}\left(\sum_{x^*\in \{0,1\}^n}\frac{1}{2^{2n}}\right)>1-\sum_{\log(2m)<n<r}\frac{1}{2^n}>1-\frac{r}{2m}. 
\end{align*} 
over the distribution of oracles $\mathcal{T}',\mathcal{T}''$, $G^{\mathcal{T}'}(k)$ and $ G^{\mathcal{T}''}(k)$ agree except with negligible probability. More formally, 
\begin{align*}
    \Pr_{\mathcal{T}',\mathcal{T}''\leftarrow \mathbb{T}[\mathcal{T}_{\leq \log(2m)}]}\left[G^{\mathcal{T}'}(k)= G^{\mathcal{T}''}(k)\right]\geq \left(1-\frac{r}{2m}\right)\left(1-\negl[\lambda]\right)\geq 1-\frac{r}{m}.
\end{align*}

\remark{It may seem problematic that our argument involves exponentially many hybrids, where each differs from its neighbor with some negligible error. Why don't the errors accumulate and become non-negligible? This issue is mitigated by the almost-determinism property of $\PRG$s, which guarantees that the most likely output must be produced with $1-\negl[\lambda]$ probability. As a result, the output cannot be gradually altered across the hybrids.}
\qed
\end{proof}

By \cref{claim 11}, for any $k\in \{0,1\}^\lambda$, there exists a string $y_k^{\mathcal{T}_{\leq \log(2m)}}$, such that  
\begin{align}
\label{eq 444}
\Pr_{\tilde{\mathcal{T}}\leftarrow \mathbb{T}[\mathcal{T}_{\leq \log(2m)}]}\left[ G^{\tilde{\mathcal{T}}}(k)= {y_k^{\mathcal{T}_{\leq \log(2m)}}}\right]\geq 1-\frac{r}{m}\geq 1-\frac{1}{\lambda^4}.
\end{align}

Furthermore, note that $G$ should not depend on $\mathcal{C}$, since the $\QKPRF$ given in \cref{lem:qkprf 2} does not rely on $\mathcal{C}$ (only on $\mathcal{O}$ and $\sigma$) and $G$ is based on a black-box construction from this $\QKPRF$. 

We now consider a generator $\overline{G}$ that does not depend on the oracles, and is defined as follows on inputs of length $\lambda+16m^3$:

\smallskip \noindent\fbox{%
    \parbox{\textwidth}{%
    $\overline{G}(k)$:
\begin{itemize}
\item Parse $k$ as $(k_1,k_2)$ where $k_1\in \{0,1\}^\lambda$ and $k_2\in \{0,1\}^{16m^3}$.
\item Construct functions $(\sigma_n^{k_2},{\mathcal{O}_n^{k_2}})_{n\leq \log(2m)}$ in the same way as \cref{con:oracles 3} but with the randomness determined by $k_2$.
\item For $j\in [\lambda]:$
\begin{enumerate}
    \item For each $\log(2m)\leq i\leq r$, sample uniformly at random two $2r$-degree polynomials $\tilde{O}_i:\mathbb{F}_{2^i}\rightarrow \mathbb{F}_{2^{8i}}$ and $\tilde{P}_i:\mathbb{F}_{2^{2i}}\rightarrow \mathbb{F}_{2^i}$. Let $(\tilde{\sigma}_i,\tilde{\mathcal{O}}_i)$ be the resulting oracles.  
   \item Run $G(k_1)$ and answer the queries as follows:
  \begin{enumerate}
  \item For a query $x$ of length $n\leq \log(2m)$, respond using $(\sigma_n^{k_2},{\mathcal{O}_n^{k_2}})$. 
  \item For a query $x$ of length $n> \log(2m)$, respond using $(\tilde{\sigma}_n,\tilde{\mathcal{O}}_n)$.
  \end{enumerate}
        \item Let $y_j$ be the result of $G(k_1)$.
        \end{enumerate}
\item Set $y=\textsf{vote}_\lambda(y_{1},\ldots, y_{\lambda})$.
        \item Output $(y,k_2)$. 
\end{itemize}}}
    \smallskip

Now consider the following variants of $\PRG$ security experiment.

\begin{itemize}
    \item $\textsf{Exp}^\adv_1(\lambda)$: 
    \begin{enumerate}
        \item Sample oracle $\mathcal{T}$ as in \cref{con:oracles 3}.
        \item $b\leftarrow \textsf{Exp}^{\textsf{PRG}}_{\adv^{\mathcal{T}},G^\mathcal{T}}(1^\lambda)$. 
        \item Output $b$. 
    \end{enumerate}
        \item $\textsf{Exp}^\adv_2(\lambda)$: 
    \begin{enumerate}
        \item Sample oracle $\mathcal{T}$ as in \cref{con:oracles 3}.
        \item $b\leftarrow \textsf{Exp}^{\textsf{PRG}}_{\adv^{\mathcal{T}_{\leq\log(2m)},\mathcal{C}},G^\mathcal{T}}(1^\lambda)$. \textit{Notice that $\adv$ only has access to $\mathcal{T}_{\leq \log(2m)}$ in this experiment.}
        \item Output $b$.
    \end{enumerate}
    \item $\textsf{Exp}^\adv_3(\lambda)$:
    \begin{enumerate}
       \item $b\leftarrow \textsf{Exp}^{\textsf{PRG}}_{\adv^\mathcal{C},\overline{G}}(1^\lambda)$. 
        \item Output $b$.
        \end{enumerate}      
\end{itemize}

By \cref{claim:adv3}, for any QPT adversary $\adv$,
\begin{align*}
    \Pr_{\mathcal{T}\leftarrow \mathbb{T}}\left[\textsf{Adtg}^{\PRG}_{\adv^\mathcal{T},G^\mathcal{T}}(1^\lambda)\leq O\left(\frac{1}{\lambda^4}\right)\right]\geq \frac{3}{4}.
\end{align*}
Therefore, for large enough $\lambda$,
\begin{align*}
    \Pr\left[\textsf{Exp}^\adv_1(\lambda)=1\right]-\frac{1}{2}\leq \frac{1}{4}\cdot \frac{1}{2}+\frac{3}{4}\cdot \frac{1}{\lambda^3}\leq \frac{1}{4}+\frac{1}{\lambda^2}.
\end{align*}

Next, the only difference between $\textsf{Exp}^\adv_1(\lambda)$ and $\textsf{Exp}^\adv_2(\lambda)$ is that $\adv$'s oracle access to $\mathcal{T}$ is restricted. Hence, for any QPT adversary $\adv$, there exists a QPT $\mathcal{B}$ such that $\Pr\left[\textsf{Exp}^\adv_2(\lambda)=1\right]\leq \Pr\left[\textsf{Exp}^{\mathcal{B}}_1(\lambda)=1\right].$

Finally, we need to relate the success probabilities of $\textsf{Exp}_2^\adv$ and $\textsf{Exp}_3^\adv$.  

\begin{Claim}
\label{claim:101}
For any QPT adversary $\adv$ and large enough $\lambda$, there exists a QPT algorithm $\mathcal{B}$ such that
    \begin{align*}
        \Pr\left[\textsf{Exp}^\mathcal{A}_3(\lambda)=1\right]&\leq \Pr\left[\textsf{Exp}^{\mathcal{B}}_2(\lambda)=1\right]+\frac{1}{\lambda^4}.
    \end{align*}
\end{Claim}

\begin{proof}
Fix an adversary $\adv$ in $\textsf{Exp}_3^\adv(\lambda)$. We construct an algorithm $\mathcal{B}$ in $\textsf{Exp}_2^{\mathcal{B}}(\lambda)$ as follows. 

In $\textsf{Exp}_2^{\mathcal{B}}(\lambda)$, a random input $k_1\leftarrow \{0,1\}^\lambda$ and a random bit $b\leftarrow \{0,1\}$ are sampled. Then, let $y_0\leftarrow G^\mathcal{T}(k_1)$ and $y_1\leftarrow \{0,1\}^s$. $\mathcal{B}^{\mathcal{T}_{\leq \log(2m)},\mathcal{C}}$ receives $y_b$ and must guess $b$. 

$\mathcal{B}^{\mathcal{T}_{\leq \log(2m)},\mathcal{C}}$ commences as follows. For every $n\leq \log(2m)$, it queries $\sigma_n(\ket{0^{9n+1}})$ $4m^2$ times and measures the result in the computational basis. This allows $\mathcal{B}$ to obtain all the evaluations of $O_n$ and learn the function entirely, except with negligible probability. Next, for each $n\leq \log(2m)$, $\mathcal{B}$ uses $O_n$ and the oracle $\mathcal{O}_n$ to learn $P_n$ entirely, which requires at most $2m$ queries. 

$\mathcal{B}$ encodes $(P_n,O_n)_{n\leq \log(2m)}$ into a string, say $k_2$, of length less than $16m^3$. $\mathcal{B}$ runs $\adv^\mathcal{C}$ on $(y_b,k_2)$ and receives a response $b'$. $\mathcal{B}$ outputs $b'$.  

As long as $\mathcal{B}$ encodes $(P_n,O_n)_{n\leq \log(2m)}$ correctly (which occurs with $1-\negl[\lambda]$ probability) and $\overline{G}(k_1,k_2)=y_0$ (which occurs with probability $1-2r/m$ by \cref{eq 444}), then it is clear that $\Pr\left[\textsf{Exp}^{\mathcal{B}}_2(\lambda)=1\right]$ is at least $\Pr\left[\textsf{Exp}^\mathcal{A}_3(\lambda)=1\right]$. Therefore, 
\begin{align*}
\Pr\left[\textsf{Exp}^{\mathcal{B}}_2(\lambda)=1\right]\geq \Pr\left[\textsf{Exp}^\mathcal{A}_3(\lambda)=1\right]\left(1-\negl[\lambda]\right)\left(1-\frac{2r}{m}\right),
\end{align*}
which implies, for large enough $\lambda$, 
\begin{align*}
    \Pr\left[\textsf{Exp}^\mathcal{A}_3(\lambda)=1\right]\leq \Pr\left[\textsf{Exp}^{\mathcal{B}}_2(\lambda)=1\right]+\frac{1}{\lambda^4}.
\end{align*}
    \qed
\end{proof}
To sum up, for any QPT adversary $\adv$, there exists a QPT algorithm $\mathcal{B}$ and large enough $\lambda$ such that
\begin{align}
    \Pr\left[\textsf{Exp}^\adv_3(\lambda)=1\right]&\leq \Pr\left[\textsf{Exp}^\mathcal{B}_2(\lambda)=1\right]+\frac{1}{\lambda^4}\leq \frac{3}{4}+\frac{1}{\lambda^2}+\frac{1}{\lambda^4}\leq \frac{3}{4}+\frac{1}{\lambda}\label{eq:322}
\end{align} 


Notice that $\textsf{Exp}^\adv_3(\lambda)$ is just the $\PRG$ security experiment for $\overline{G}$ against $\adv^\mathcal{C}$. On the other hand, there exists a trivial search attack, using a \textsf{PSPACE} oracle, against $\PRG$ security, given that any polynomial-space quantum computations with classical inputs can be simulated using a \textsf{PSPACE} oracle. In particular, there exists an adversary $\overline{\adv}$ such that 
\begin{align*}
    \Pr\left[\textsf{Exp}^{\overline{\adv}}_3(\lambda)=1\right]\geq 1-\negl[\lambda].
\end{align*} 
contradicting \cref{eq:322} above. 

Therefore, there does not exist a fully black-box construction of a $\PRG$ from a $\QKPRF$.
\qed
\end{proof}

\subsection{Separating \textsf{OWSG} from \texorpdfstring{$\bot$}{bot}-\textsf{PRG}}
\label{sec:implications 1}

We show that there does not exist a fully black-box construction of a $\OWSG$s from a $\botPRG$ with CPTP access. The proof is given in \cref{sec:separation 1}, but the result and implications are discussed below. 

\begin{theorem}
\label{lem:seperation}
Let $\lambda\in \mathbb{N}$ be the security parameter. For any polynomial $ m(\lambda)>\lambda$ and pseudodeterminism error $  \mu(\lambda)=O(\lambda^{-c})$ for $c>0$, there does not exist a fully black-box construction of a $\OWSG$ from CPTP access to a $(\mu,m)$-$\botPRG$. 
\end{theorem}

This separation is significant because there are multiple MicroCrypt primitives that can be built from $\botPRG$s, thus yielding new separations in MicroCrypt. 

First of all, we note that $\botPRG$s can be used to build the following primitives \cite{BBO+24}.

\begin{lemma}
\label{thm:implications of botprg}
    There exists black-box constructions of the following primitives from $\botPRG$s:
\begin{enumerate}
\item $\botPRF$s.
    \item (Many-time) digital signatures of classical messages with classical keys and signatures. 
    \item Quantum public-key encryption of classical messages with tamper-resilient keys and classical ciphertexts. 
\end{enumerate}
\end{lemma}

As a result of \cref{lem:seperation,thm:implications of botprg}, we obtain a separation between $\OWSG$s and some cryptographic applications. 

\begin{corollary}
\label{cor:separation}
    There does not exist a fully black-box construction of a $\OWSG$ from CPTP access to:
    \begin{enumerate}
\item $\botPRF$s.
    \item (Many-time) digital signatures of classical messages with classical keys and signatures. 
    \item Quantum public-key encryption of classical messages with tamper-resilient keys and classical ciphertexts. 
\end{enumerate}
\end{corollary}

\subsubsection{Separation Proof.}
\label{sec:separation 1}

To demonstrate the separation between $\OWSG$s and $\botPRG$s, we will use two independent oracles: an oracle for a \textsf{PSPACE}-complete language and a $\bot$-pseudodeterministic random oracle. We show that the random oracle already acts as a $\botPRG$, noting that the $\bot$-pseudodeterminism is not a problem given the nature of $\botPRG$s. On the other hand, the unpredictability of the $\bot$-pseudodeterminism in the random oracle prevents its use in the construction of almost-deterministic primitives such as a $\OWSG$. Specifically, through a careful argument, we show that any $\OWSG$ must exist independently of the random oracle and, thus, cannot exist in the presence of a \textsf{PSPACE} oracle. Given this distinction, there cannot exist a fully black-box construction of a $\OWSG$ from a $\botPRG$.


Let $\lambda,n\in \mathbb{N}$ be security parameters. Let $c>0$ be a constant. We define a sequence of CPTP oracles $\mathcal{O}\coloneqq \{\mathcal{O}_n\}_{n\in \mathbb{N}}$ as follows. 

\begin{construct}
\label{con:oracles 1}
   Fix a pseudodeterminism error $\mu(n)=n^{-c}$. Let $w=w(n)$ be any function such that $2^{-w}\in [ \mu/16 ,\mu/4]$ and let $m$ be polynomial such that $m(n)>n$. Sample a random permutation $P_n \leftarrow \Pi_{n}$ and random functions $Q_n\leftarrow \Pi_{n,n}$ and $O_n\leftarrow \Pi_{n,m}$. The quantum channel $\mathcal{O}_n\coloneqq \mathcal{O}[P_n,Q_n,O_n]$ on $n$-qubit input $\rho$ does as follows:
    \begin{itemize}
        \item Measure $\rho$ in the computational basis and let $x$ denote the result. 
        \item Compute $y=O_n(x)$. 
        \item Compute $q= Q_n(x)$ and let $p_x\coloneqq q/2^n$, where $q$ is interpreted as an integer in $[0:2^n]$.
        \item Compute $z= P_n(x)$. If the first $w$-bits of $z$ are $0^{w}$, then let $\ket{\phi_x}\coloneqq \sqrt{p_x}\ket{\bot}+\sqrt{1-p_x}\ket{y}$.
        \item Otherwise, let  $\ket{\phi_x}\coloneqq \ket{y}$.
        \item Measure $\ket{\phi_x}$ in the computational basis and output the result. 
    \end{itemize}
\end{construct}


We define the ``good'' set $\mathcal{G}_n^\mathcal{O}$ for $\mathcal{O}_n$ as follows: 
\begin{align*}
    \mathcal{G}_n^\mathcal{O}\coloneqq \{x\in \{0,1\}^n: P_n(x)_{[1:w]}\neq 0^w\},
\end{align*}
where $P_n(x)_{[1:w]}$ denotes the first $w$-bits of $P_n(x)$. 


The following lemma follows directly from the definition of $\mathcal{O}$.  
\begin{lemma}
\label{lem:O}
    $\mathcal{O}_n$ has the following properties:
\begin{itemize}
            \item $\Pr_{x\leftarrow \{0,1\}^n}\left[ x\in \mathcal{G}_n^\mathcal{O}\right] \geq 1-\frac{\mu}{4}.$

  \item For every $x\in \mathcal{G}_n^\mathcal{O}$, there exists a non-$\bot$ value $y\in\{0,1\}^{m}$ such that: 
  \begin{align*}
      \Pr\left[\mathcal{O}_n(x)=y \right] =1.
  \end{align*}
            \item For every $x\notin \mathcal{G}_n^\mathcal{O}$, there exists a probability $p_x\in [0,1]$ and non-$\bot$ value $y\in\{0,1\}^m$ such that: 
            \begin{enumerate}
                \item  $\Pr\left[y\gets \mathcal{O}_n(x) \right] =1-p_x.$ 
            \item $\Pr\left[\bot \gets \mathcal{O}_n(x) \right] =p_x$. 
            \end{enumerate}            
\end{itemize}
\end{lemma}


\begin{theorem}
Let $\lambda\in \mathbb{N}$ be the security parameter. For any polynomial $ m(\lambda)>\lambda$ and pseudodeterminism error $  \mu(\lambda)=O(\lambda^{-c})$ for $c>0$, there does not exist a fully black-box construction of a $\OWSG$ from CPTP access to a $(\mu,m)$-$\botPRG$. 
\end{theorem}

\begin{proof} 
Assume for contradiction that there exists a black-box construction of a $\OWSG$ $\tilde{G}^F$ from CPTP access to a $(\mu,m)$-$\botPRG$ $F$. First, we show that there exists a $(\mu,m)$-$\botPRG$ relative to the oracles $(\mathcal{O},\mathcal{C})$.

\begin{Claim}
\label{lem:determinism}
Under security parameter $n\in \mathbb{N}$, the sequence of functions $\{\mathcal{O}_n[P_n,Q_n,O_n]\}_{n\in \mathbb{N}}$ is a $(\mu(n),m(n))$-$\botPRG$ for all possible sequences $P$ and $Q$ and with probability 1 over the distribution of $O$. Furthermore, correctness is satisfied for all possible oracles.
\end{Claim}

\begin{proof}
By \cref{lem:O}, $\mathcal{O}$ satisfies the correctness/pseudodeterminism condition of a $(\mu,m)$-$\botPRG$. 

For security, we need to show that for any $P,Q$ and with probability 1 over the distribution of $O$: for every non-uniform QPT distinguisher $\adv$ and polynomial $q=q(n)$: 
\begin{align*}
        \left| \Pr \left[ \begin{matrix}
            k\gets \{0,1\}^n\\
            y_1\gets \mathcal{O}_n(k)\\
            \vdots \\
            y_q \gets \mathcal{O}_n(k)      
        \end{matrix} : \adv^{\mathcal{O},\mathcal{C}}(y_1,...,y_q) = 1 \right] - \Pr \left[ \begin{matrix}
            k\gets \{0,1\}^n \\
            y\gets \{0,1\}^{m} \\
            y_1\gets \isbot(\mathcal{O}_n(k),y)\\
            \vdots \\
            y_q \gets \isbot(\mathcal{O}_n(k),y)      
        \end{matrix}: \adv^{\mathcal{O},\mathcal{C}}(y_1,\ldots,y_q) = 1    
        \right] \right| \leq \negl[n]
    \end{align*}

Let $Z_n$ be the function that outputs $0^m$ on any input and let $\mathcal{Z}_n\coloneqq \mathcal{O}[P_n,Q_n,Z_n]$. Note that \begin{itemize}
    \item $\mathcal{Z}_n$ is independent of $O_n$,
    \item $\mathcal{O}_n(k)=\isbot(\mathcal{O}_n(k),O_n(k))=\isbot(\mathcal{Z}_n(k),O_n(k))$,
    \item $\isbot(\mathcal{O}_n(k),y) =\isbot(\mathcal{Z}_n(k),y) $.
\end{itemize}

Therefore, $\adv^{\mathcal{O},\mathcal{C}}$ needs to distinguish between evaluations of $\isbot(\mathcal{Z}_n(k),y)$ and $\isbot(\mathcal{Z}_n(k),O_n(k))$. 

Lemma 2.2 from \cite{SXY18} states that a random oracle acts as a $\PRG$ i.e.:
\begin{align*}
   \underset{{O\leftarrow \Pi_{n,m}}}{\mathbb{E}} \left[ \left| \Pr_{k\leftarrow \{0,1\}^n}\left[\adv^{O}(O(k))=1\right]-    \Pr_{y\leftarrow \{0,1\}^m}\left[\adv^{O}(y)=1\right]\right| \right]\leq \frac{1}{2^{n/4}}.
\end{align*}

Note that this result even holds against unbounded-time adversaries as long as the number of queries to the oracle is polynomial. Hence, this result also holds against adversaries with access to a \textsf{PSPACE}-oracle:
\begin{align*}
   \underset{{O\leftarrow \Pi_{n,m}}}{\mathbb{E}} \left[ \left| \Pr_{k\leftarrow \{0,1\}^n}\left[\adv^{O,\mathcal{C}}(O(k))=1\right]-    \Pr_{y\leftarrow \{0,1\}^m}\left[\adv^{O,\mathcal{C}}(y)=1\right]\right| \right]\leq \frac{1}{2^{n/4}}.
\end{align*}

Next, notice that for any functions $P,Q$, distinguishing between evaluations of $\isbot(\mathcal{Z}_n(k),y)$ and $\isbot(\mathcal{Z}_n(k),O_n(k))$ is just as hard as distinguishing the two scenarios in the equation above, given that $\mathcal{Z}_n$ is independent of $O_n$. Therefore, 
\begin{align*}
      \underset{{{O\leftarrow \Pi_{n,m}}}}{\mathbb{E}} \left[ \left| \Pr_{(y_1,...,y_q)\leftarrow D_{\mathcal{Z},O}^0} \left[ \adv^{\mathcal{O},\mathcal{C}}(y_1,...,y_q) = 1 \right] - \Pr_{(y_1,...,y_q)\leftarrow D_{\mathcal{Z}}^1} \left[  \adv^{\mathcal{O},\mathcal{C}}(y_1,\ldots,y_q) = 1    
        \right] \right|\right] \leq 2^{-n/4}
    \end{align*}
where, 
\begin{align*}
D_{\mathcal{Z},O}^0\coloneqq \left[
\begin{matrix}
            k\gets \{0,1\}^n\\
            y_1\gets \isbot(\mathcal{Z}_n(k),O_n(k))\\
            \vdots \\
            y_q \gets \isbot(\mathcal{Z}_n(k),O_n(k))     
        \end{matrix}    \right]  \  D_{\mathcal{Z}}^1\coloneqq \left[
\begin{matrix}
            k\gets \{0,1\}^n \\
            y\gets \{0,1\}^{m} \\
            y_1\gets \isbot(\mathcal{Z}_n(k),y) \\
            \vdots \\
            y_q \gets \isbot(\mathcal{Z}_n(k),y)       
        \end{matrix}\right]
\end{align*}

By Markov inequality, we get that
\begin{align*}
      \Pr_{{O\leftarrow \Pi_{n,m}}} \left[ \left|  \Pr_{(y_1,...,y_q \leftarrow D_{\mathcal{Z},O}^0}\right. \right. & \left[ \adv^{\mathcal{O},\mathcal{C}}(y_1,...,y_q) = 1 \right] -\\ 
      &\left. \left. \Pr_{(y_1,...,y_q)\leftarrow D_{\mathcal{Z}}^1} \left[  \adv^{\mathcal{O},\mathcal{C}}(y_1,\ldots,y_q) = 1   
        \right] \right| \geq {2^{-n/8}}\right] \leq 2^{-n/8}
    \end{align*}
    
By Borel-Cantelli Lemma, since $\sum_n 2^{-n/8}$ converges, with probability 1 over the distribution of $O$, it holds that
\begin{align*}
      \left|  \Pr_{(y_1,...,y_q)\leftarrow D_{\mathcal{Z},O}^0} \left[ \adv^{\mathcal{O},\mathcal{C}}(y_1,...,y_q) = 1 \right] - \Pr_{(y_1,...,y_q)\leftarrow D_{\mathcal{Z}}^1} \left[  \adv^{\mathcal{O},\mathcal{C}}(y_1,\ldots,y_q) = 1    
        \right] \right| \leq {2^{-n/8}},
    \end{align*}

except for finitely many $n\in \mathbb{N}$. There are countable number of quantum algorithms $\adv$ making polynomial queries to $(\mathcal{O},\mathcal{C})$, so this bound holds for every such adversary. Therefore, $\mathcal{O}$ is a $\botPRG$ for any $P,Q$ and with probability 1 over the distribution of $O$.  
    \qed
\end{proof}


By our assumption, the above claim implies the existence of a $\OWSG$ $G^{\mathcal{O}[P,Q,O]}$, by the assumed existence of the black-box construction $\tilde{G}$, for any $P,Q$ and with probability 1 over the distribution of $O$. 

\begin{Claim}
\label{claim:adv}
    For any QPT adversary $\adv$ and polynomial $t=t(\lambda)$:
\begin{align*}
\Pr_{\mathcal{O}}\left[\textsf{Adtg}^{\textsf{OWSG}}_{\adv^{\mathcal{O},\mathcal{C}},G^\mathcal{O}}(1^\lambda,1^t)\leq O\left(\frac{1}{\lambda^4}\right)\right]\geq \frac{3}{4}.
\end{align*}
where the probability is taken over the oracle distribution. 
\end{Claim}
\begin{proof}
Same as the proof of \cref{claim:adv3}. 
\qed
\end{proof}

Let $r=r(\lambda)$ be a polynomial denoting the maximum run-time of $G$ on any input and let $m=r^4+\lambda$.

Intuitively, we will argue that $G$ cannot depend on $\mathcal{O}_n$ for large $n$, due to the deterministic nature of $G$, and thus cannot exist in the presence of a \textsf{PSPACE} oracle. 

We use the notation $G\simeq G'$ to mean that there exists negligible function $\epsilon=\epsilon(\lambda)$ such that for every $k\in \{0,1\}^\lambda$, there exists a pure-state $\ket{\psi_k}$, such that $\Pr\left[G(k)=\ket{\psi_k}\right]$ and $\Pr\left[G'(k)=\ket{\psi_k}\right]$ are both at least $1-\epsilon$.



\begin{Claim}
\label{claim 1}
    Let $\mathcal{O}'\coloneqq \mathcal{O}[P',Q',O']$ and $\mathcal{O}''\coloneqq \mathcal{O}[P'',Q'',O'']$ be two oracles such that $(P'_n,Q'_n,O_n')=(P_n'',Q_n'',O_n'')$ for all $n\leq \log(2m)$ and $n\geq r$. Then, $G^{\mathcal{O}'}\simeq G^{\mathcal{O}''}$. 
\end{Claim}

\begin{proof}
We will first show that there exists a function $\ell=\ell(\lambda)$ and sequences $P^{1},P^{2},\ldots, P^{\ell}$, $Q^{1},Q^{2},\ldots, Q^{\ell}$, and $O^1,O^2,\ldots,O^\ell$, where $P^i\coloneqq \{P^i_n\}_{n\in \mathbb{N}}$, $Q^i\coloneqq \{Q^i_n\}_{n\in \mathbb{N}}$, and $O^i\coloneqq \{O^i_n\}_{n\in \mathbb{N}}$, such that: 
\begin{enumerate}
\item $P_n^{i}\in \Pi_n$, $Q_n^{i}\in \Pi_{n,n}$, and $O^i_n\in \Pi_{n,m}$ for any $i\in [\ell]$ and $n\in \mathbb{N}$.
    \item $P'=P^{1}$ and $P''=P^{\ell}$. 
    \item $Q'=Q^{1}$ and $Q''=Q^{\ell}$. 
    \item $O'=O^{1}$ and $O''=O^{\ell}$. 
    \item For any $i\in [\ell]$, 
    \begin{align*}
        d_{\textsf{TD}}(\mathcal{O}^i,\mathcal{O}^{i+1})\coloneqq\sum_{n\in \mathbb{N}}\sum_{x\in \{0,1\}^n}d_{\textsf{TD}}(\mathcal{O}^i_n(x),\mathcal{O}_n^{i+1}(x))\leq \frac{1}{m}, 
    \end{align*}
    where $\mathcal{O}^i\coloneqq \mathcal{O}[P^i,Q^i,O^i]$.
\end{enumerate} 

We will now describe how to construct such a sequence. Note that $\mathcal{O}'_n$ and $\mathcal{O}_n''$ only differ for $\log(2m)<n<r$. For such values of $n$, if we set $Q^{2}_n(x)$ to $ Q^1_n(x)+ 1$ or $Q^1_n(x)- 1$, while keeping the other functions fixed, the resulting oracles satisfy:
    \begin{align*}
      d_{\textsf{TD}}(\mathcal{O}^1,\mathcal{O}^{2})\leq \frac{1}{2^{\log(2m)}}\leq \frac{1}{m} 
    \end{align*}

It is not difficult to see that this allows constructing the sequence of functions described. Specifically, for any $\log(2m)<n<r$, we perform small changes to $Q_n'$ until we reach a function, say $Q_n^j$, that sends all values to $1^n$, while keeping all other functions fixed. Then, we set $O_n^{j+1}(x)= O_n''(x)$ for all $x$ such that $P_n(x)_{[1:w]}=0^w$ and keep $O_n^{j+1}(x)= O_n^j(x)$ otherwise. This step does not change the oracle, i.e. $\mathcal{O}^j= \mathcal{O}^{j+1}$, because $Q_n^{j}$ and $Q_n^{j+1}$ return $1^n$ on these inputs so both oracles return $\bot$. 

Next, we perform small changes to $Q_n^{j+1}$ until we reach a function, say $Q_n^t$ for some $t>j$, that sends all values to $0^n$. Then, we set $P_n^{t+1}$ to any function in $\Pi_{n}$. Again, this step does not change the oracle, i.e. $\mathcal{O}^t= \mathcal{O}^{t+1}$, because $Q_n^{t}$ and $Q_n^{t+1}$ return $0^n$ on any input. The new function $P_n^{t+1}$ allows us to perform the first step on a new set of inputs i.e. we can modify $O^{t+1}_n$ on all inputs such that $P_n^{t+1}(x)_{[1:w]}=0^w$. Iteratively applying these modifications allows us to reach the required functions $O''_n$ and $P''_n$. Finally, we perform small changes to $Q_n$ while keeping the other functions fixed to obtain $Q_n''$. These steps are performed for all $n\in [\log(2m):r]$ to build the sequence described.

Since $P_n^{i}\in \Pi_n$, $Q_n^{i}\in \Pi_{n,n}$, and $O^i_n\in \Pi_{n,m}$ for any $n\in \mathbb{N}$, by \cref{lem:determinism}, $G^{\mathcal{O}^i}$ is almost-deterministic for any $i\in [\ell]$. 

Note that for any $i\in [\ell]$, $\mathcal{O}_i$ is an oracle with classical output and $d_{\textsf{TD}}(\mathcal{O}^i,\mathcal{O}^{i+1})\leq \frac{1}{m}$. Therefore, with probability at least $(1-\frac{1}{m})^r\geq 1-\frac{r}{m}$, the responses that $G$ receives from $\mathcal{O}_i$ and $\mathcal{O}_{i+1}$ are indistinguishable. This means that for any $k\in \{0,1\}^\lambda$, with probability at least $1-\frac{r}{m}-\negl[\lambda]$, $G^{\mathcal{O}^{i+1}}(k)$ outputs the same state generated by $G^{\mathcal{O}^{i}}(k)$. In order to satisfy the determinism property of $\OWSG$s, this must mean that $G^{\mathcal{O}^{i}}\simeq G^{\mathcal{O}^{i+1}}$ for all $i\in [\ell]$. By induction, we obtain $G^{\mathcal{O}'}\simeq G^{\mathcal{O}''}$.
    \qed
\end{proof}

For any oracle $\mathcal{O}=\mathcal{O}[P,Q,O]$, $G$ is independent of $(P_n,Q_n,O_n)_{n\geq r}$. So, by the above claim, for any $(P_n,Q_n,O_n)_{n\leq \log(2m)}$ and $k\in \{0,1\}^\lambda$, there exists a state $\ket{\psi_k^{\mathcal{O}_{\leq \log(2m)}}}$, such that for any $(\tilde{P},\tilde{Q},\tilde{O})$ satisfying $(\tilde{P}_n,\tilde{Q}_n,\tilde{O}_n)_{n\leq \log(2m)}=(P_n,Q_n,O_n)_{n\leq \log(2m)}$,  
\begin{align}
\label{eq 2}
\Pr\left[ G^{\mathcal{O}[\tilde{P},\tilde{Q},\tilde{O}]}(k)= \ket{\psi_k^{\mathcal{O}_{\leq \log(2m)}}}\right]\geq 1-\negl[\lambda].
\end{align}

We now consider a generator $\overline{G}$ that does not depend on the oracle and is defined as follows on inputs of length $\lambda+16m^3$:

\smallskip \noindent\fbox{%
    \parbox{\textwidth}{%
    $\overline{G}(k)$:
\begin{itemize}
\item Parse $k$ as $(k_1,k_2)$ where $k_1\in \{0,1\}^\lambda$ and $k_2\in \{0,1\}^{16m^3}$.
\item Construct functions $({\mathcal{O}_n^{k_2}})_{n\leq \log(2m)}$ in the same way as \cref{con:oracles 1} but with the randomness determined by $k_2$.
\item Initiate an empty memory $\textbf{M}$.
   \item Run $G(k_1)$ and answer the queries as follows:
  \begin{enumerate}
  \item For a query $x$ of length $n\leq \log(2m)$, respond with ${\mathcal{O}_n^{k_2}}(x)$. 
  \item For a query $x$ of length $n> \log(2m)$, if $(x,y)\in \mathbf{M}$ for some $y$, respond with $y$. Otherwise, sample $y\leftarrow \{0,1\}^m$, store $(x,y)$ in $\textbf{M}$, and respond with $y$.  
  \end{enumerate}
        \item Let $\ket{\psi}$ be the result of $G(k_1)$.
        \item Output $\ket{\psi}\otimes \ket{k_2}$. 
\end{itemize}}}
    \smallskip

Consider the following experiment variants of $\OWSG$ security with some polynomial $t=t(\lambda)$.

\begin{itemize}
    \item $\textsf{Exp}^\adv_1(\lambda)$: 
    \begin{enumerate}
        \item Sample oracle $\mathcal{O}$ as in \cref{con:oracles 1}.
        \item $b\leftarrow \textsf{Exp}^{\textsf{OWSG}}_{\adv^{\mathcal{O},\mathcal{C}},G^\mathcal{O}}(1^\lambda,1^t)$. 
        \item Output $b$.
    \end{enumerate}
        \item $\textsf{Exp}^\adv_2(\lambda)$: 
    \begin{enumerate}
        \item Sample oracle $\mathcal{O}$ as in \cref{con:oracles 1}.
        \item $b\leftarrow \textsf{Exp}^{\textsf{OWSG}}_{\adv^{\mathcal{O}_{\leq\log(2m)},\mathcal{C}},G^\mathcal{O}}(1^\lambda,1^t)$. \textit{Notice that $\adv$ only has access to $\mathcal{O}_{\leq \log(2m)}$ in this experiment.}
        \item Output $b$.
    \end{enumerate}
    \item $\textsf{Exp}^\adv_3(\lambda)$:
    \begin{enumerate}
       \item $b\leftarrow \textsf{Exp}^{\textsf{OWSG}}_{\adv^\mathcal{C},\overline{G}}(1^\lambda,1^t)$. 
        \item Output $b$.
        \end{enumerate}      
\end{itemize}

By \cref{claim:adv}, for any QPT adversary $\adv$,
\begin{align*}
\Pr_{\mathcal{O}}\left[\textsf{Adtg}^{\textsf{OWSG}}_{\adv^{\mathcal{O},\mathcal{C}},G^\mathcal{O}}(1^\lambda,1^t)\leq O\left(\frac{1}{\lambda^4}\right)\right]\geq \frac{3}{4}.
\end{align*}
Therefore, for large enough $\lambda$
\begin{align*}
    \Pr\left[\textsf{Exp}^\adv_1(\lambda)=1\right]\leq \frac{1}{4}+\frac{3}{4}\cdot \frac{1}{\lambda^3}.
\end{align*}

Next, the only difference between $\textsf{Exp}^\adv_1(\lambda)$ and $\textsf{Exp}^\adv_2(\lambda)$ is that $\adv$'s oracle access to $\mathcal{O}$ is restricted. Hence, for any QPT adversary $\adv$, there exists a QPT $\mathcal{B}$ such that $\Pr\left[\textsf{Exp}^\adv_2(\lambda)=1\right]\leq \Pr\left[\textsf{Exp}^{\mathcal{B}}_1(\lambda)=1\right]$.

Finally, we need to relate $\textsf{Exp}_2^\adv$ and $\textsf{Exp}_3^\adv$.  

\begin{Claim}
\label{claim:remove oracle}
For any QPT adversary $\adv$, there exists a QPT algorithm $\mathcal{B}$ such that
    \begin{align*}
        \Pr\left[\textsf{Exp}^\mathcal{A}_3(\lambda)=1\right]&\leq \Pr\left[\textsf{Exp}^{\mathcal{B}}_2(\lambda)=1\right]+\negl[\lambda].
    \end{align*}
\end{Claim}

\begin{proof}
Fix an adversary $\adv$ in $\textsf{Exp}_3^\adv(\lambda)$. We construct an algorithm $\mathcal{B}$ in $\textsf{Exp}_2^{\mathcal{B}}(\lambda)$ as follows. 

In $\textsf{Exp}_2^{\mathcal{B}}(\lambda)$, a random input $k_1\leftarrow \{0,1\}^\lambda$ is sampled and $t+1$ evaluations are generated $\ket{\psi_i} \leftarrow G^{\mathcal{O}}(k_1)$ for $i\in [t+1]$. Then, $\mathcal{B}^{\mathcal{O}_{\leq \log(2m)},\mathcal{C}}$ receives $\bigotimes_{i\in [t]}\ket{\psi_{i}}$. 

$\mathcal{B}$ commences as follows. It queries the oracle $\mathcal{O}_{\leq \log(2m)}$ $4m^2$ times on every input of length less than $\log(2m)$. This allows $\mathcal{B}$ to describe $\mathcal{O}_{\leq \log(2m)}$ in a string, say $k_2$, of length less than $16m^3$. Specifically, $k_2$ is used to describe the randomness used in constructing $\mathcal{O}_{\leq \log(2m)}$ (see \cref{con:oracles 1}). It is not difficult to see that $\mathcal{B}$ can learn $\mathcal{O}_{\leq \log(2m)}$ exactly, except with negligible probability. $\mathcal{B}^{\mathcal{O}_{\leq \log(2m)},\mathcal{C}}$ runs $\adv^\mathcal{C}$ on $\bigotimes_{i\in [t]}(\ket{\psi_i}\otimes \ket{k_2})$ and receives a response $(k_1',k_2')$. $\mathcal{B}$ outputs $k_1'$.  

Next, the experiment in $\textsf{Exp}_2^{\mathcal{B}}(\lambda)$ computes $\ket{\psi_{k_1'}}\leftarrow G^\mathcal{O}(k'_1)$ and measures $\ket{\psi_{t+1}}$ with $\{\ket{\psi_{k_1'}}\bra{\psi_{k_1'}},I-\ket{\psi_{k_1'}}\bra{\psi_{k_1'}}\}$. If the result is $\ket{\psi_{k'}}\bra{\psi_{k'}}$, then the output is $b=1$, and the output is $b=0$ otherwise. 

Notice that the input $\adv$ receives and needs to invert from $\mathcal{B}$ has the same distribution as the input it receives and needs to invert in $\textsf{Exp}_3^\adv$. Moreover, as long as $\mathcal{B}$ encodes $\mathcal{O}_{\leq \log(2m)}$ correctly (which occurs with $1-\negl[\lambda]$ probability), it is clear that $\Pr\left[\textsf{Exp}^\mathcal{A}_3(\lambda)=1\right]$ is at most $\Pr\left[\textsf{Exp}^{\mathcal{B}}_2(\lambda)=1\right]$, since the inversion task in the former is at least as hard as in the latter. Therefore, 
\begin{align*}
    \Pr\left[\textsf{Exp}^\mathcal{A}_3(\lambda)=1\right]&\leq \Pr\left[\textsf{Exp}^{\mathcal{B}}_2(\lambda)=1\right]+\negl[\lambda].
\end{align*}
    \qed
\end{proof}

To sum up, for large enough $\lambda$, for any QPT adversary $\adv$, there exists a QPT algorithm $\mathcal{B}$ such that
\begin{align}
    \Pr\left[\textsf{Exp}^\adv_3(\lambda)=1\right]&\leq \Pr\left[\textsf{Exp}^\mathcal{B}_2(\lambda)=1\right]+\negl[\lambda]\label{eq:3}\\
    &\leq \frac{1}{4}+\frac{3}{4\lambda^3}+\negl[\lambda]\label{eq:3-2}
\end{align} 


Notice that $\textsf{Exp}^\adv_3(\lambda)$ is just the $\OWSG$ security experiment for $\overline{G}$ against $\adv^\mathcal{C}$. On the other hand, by \cref{lem:OWSG 3}, there exists an attack against any $\OWSG$ using a \textsf{PSPACE} oracle. In particular, there exists an adversary $\overline{\adv}$ such that 
\begin{align*}
    \Pr\left[\textsf{Exp}^{\overline{\adv}}_3(\lambda)=1\right]\geq \frac{1}{2}.
\end{align*} 
contradicting \cref{eq:3-2} above. 

Therefore, there does not exist a fully black-box construction of a $\OWSG$ from a $(\mu,m)$-$\botPRG$.
\qed
\end{proof}

\subsection{\texorpdfstring{Separating $\botPRG$ from $\QKPRF$}{Separating bot-PRG from PRFqs}}
\label{sec:implications 2}

We strengthen the first separation (\cref{thm:separation 3}) but in the more restricted setting of CPTP access. Specifically, we show that there does not exist a fully black-box construction of a $\botPRG$ from CPTP access to a $\QKPRF$. 

The proof is given in \cref{sec:separation 2}, but the result and its implications are discussed below. 

\begin{theorem}
\label{lem:no botowsg}
Let $\lambda,n\in \mathbb{N}$ be security parameters and let $s=s(\lambda)>3\lambda$ and $\mu=\mu(\lambda)\le \lambda^{-c}$ be functions where $c>0$ is a constant. There does not exist a fully black-box construction of a $(\mu,s)$-$\botPRG$ from CPTP access to a $\QKPRF$. 
\end{theorem}

As a result of \cref{lem:no botowsg}, we obtain a separation between $\botPRG$s and many other MicroCrypt primitives since $\QKPRF$ inherit many of the same applications as $\PRF$s. We note that this separation extends to $\SPRS$s since they imply $\botPRG$s \cite{BBO+24,MY22a}. 

\begin{corollary}
\label{cor:separation2}
    There does not exist a fully black-box construction of a $\botPRG$s or a $\SPRS$s from CPTP access to:
   \begin{enumerate}
    \item $\QKPRG$, $\QKSPRS$, $\QKLPRS$, and $\QKPRU$. 
 \item Statistically-binding computationally hiding bit commitments with classical communication. 
\item Message authentication codes of classical messages with classical communication. 
        \item CCA2-secure symmetric encryption of classical messages with classical keys and ciphertexts. 
        \item \textsf{EV-OWPuzz}s. 
    \end{enumerate}
\end{corollary}

\subsubsection{Separation Proof.}
\label{sec:separation 2}

We prove that there does not exist a fully black-box construction of a $\botPRG$ from a $\QKPRF$. 

We present $(\mu,s)$-$\botPRG$ security relative to an oracle $\mathcal{T}$ in the form of an experiment to simplify notation later on. 

\smallskip \noindent\fbox{%
    \parbox{\textwidth}{%
$\textsf{Exp}^{\botPRG}_{\adv^\mathcal{T},G^\mathcal{T}}(1^\lambda,1^q)$:
\begin{enumerate}
    \item Sample $k\leftarrow \{0,1\}^\lambda$, $b\leftarrow \{0,1\}$, and $y\leftarrow \{0,1\}^s$. 
    \item If $b=0$, for each $i\in [q]$, generate $y_i\leftarrow G^\mathcal{T}(k)$. 
    \item If $b=1$, for each $i\in [q]$, generate $y_i\leftarrow \isbot(G^\mathcal{T}(k),y)$.
    \item $b'\leftarrow \adv^\mathcal{T}(y_1,\ldots,y_q)$.
\item If $b'=b$, output $1$. Otherwise, output $0$. 
\end{enumerate}}}
    \smallskip  

We define the advantage of the adversary in this experiment as follows. 
\begin{align*}            \textsf{Adtg}^{\botPRG}_{\adv^\mathcal{T},G^\mathcal{T}}(1^\lambda)\coloneqq \Pr\left[\textsf{Exp}^{\botPRG}_{\adv^\mathcal{T},G^\mathcal{T}}(1^\lambda)=1\right]-\frac{1}{2}.
        \end{align*}
        
We now describe the oracles used in the separation. These are similar to the oracles given in \cref{con:oracles 3}, which were used to separate $\PRG$ and $\QKPRF$. However, in this case, $\sigma$ is no longer a unitary, which aids in preventing its use in a wider variety of primitives, including $\botPRG$s. 

\begin{construct}
\label{con:oracles 2}
    For $n\in \mathbb{N}$, let $O_n\leftarrow \Pi_{n,n}$ and $P_n\leftarrow \Pi_{2n,n}$ be random functions. Let $\mathcal{T} \coloneqq (\sigma,\mathcal{O}, \mathcal{C})$ be a tuple of oracles, where $\sigma=\{\sigma_n\}_{n\in \mathbb{N}}$, $\mathcal{O}\coloneqq \{\mathcal{O}_n\}_{n\in \mathbb{N}}$, and $\mathcal{C}\coloneqq \{\mathcal{C}_n\}_{n\in \mathbb{N}}$ are defined as follows:  
\begin{enumerate}
    \item $\mathcal{C}$ is for membership in a \textsf{PSPACE}-complete language.
    \item $\sigma_n(1^n)$:
    \begin{enumerate}
    \item Sample $x\leftarrow \{0,1\}^n$. 
    \item Output $\ket{x,O_n(x)}$. 
    \end{enumerate}
    \item $\mathcal{O}_n:$ Unitary of the classical function that maps $(x,y,a)$, where $x,y,a\in \{0,1\}^n$, to $P_n(x,a)$ if $O_n(x)=y$ and to $\bot$ otherwise.
\end{enumerate}
\end{construct}

We first introduce some notation for the proof, similar to \cref{sec:implications 3}. Let $\mathbb{T}$ denote the set of all possible oracles and let $\mathcal{T}\leftarrow \mathbb{T}$ denote sampling an oracle in the way given in \cref{con:oracles 2}. For any oracle $\mathcal{T}$ and integer $m\in \mathbb{N}$, let $\mathcal{T}_{\leq m}$ denote the sequence of oracles $(\sigma_n,\mathcal{O}_n)_{n\leq m}$ and let $\mathbb{T}[{\mathcal{T}_{\leq m}}]$ denote the set
\begin{align*}
\mathbb{T}[{\mathcal{T}_{\leq m}}]\coloneqq \{\tilde{\mathcal{T}}\in \mathbb{T}:\tilde{\mathcal{T}}_{\leq m}=\mathcal{T}_{\leq m}\}.    
\end{align*}

\begin{theorem}
Let $\lambda,n\in \mathbb{N}$ be security parameters and let $s=s(\lambda)>3\lambda$ and $\mu=\mu(\lambda)\le \lambda^{-c}$ be functions, where $c>0$ is a constant. There does not exist a fully black-box construction of a $(\mu,s)$-$\botPRG$ from CPTP access to a $\QKPRF$. 
\end{theorem}

\begin{proof}
For simplicity, we only prove the theorem for $(9n,n)$-$\QKPRF$, but the proof easily generalizes to other parameters by modifying the parameters of the oracles. 

Assume, for the purpose of obtaining a contradiction, that $\tilde{G}^F$ is a fully black-box construction of a $(\mu,s)$-$\botPRG$, from CPTP access to a $(9n,n)$-$\QKPRF$ $F$. We first show that there exists a $\QKPRF$ relative to $\mathcal{T}$.

\begin{Claim}
\label{lem:qkprf}
    There exists a (quantum-query-secure) $(9n,n)$-$\QKPRF$ relative to $\mathcal{T}$ for any $O$ and with probability 1 over the distribution of $P$. Furthermore, correctness is satisfied for any oracle $\mathcal{T}$.
\end{Claim}

\begin{proof}
    Similar to proof of \cref{lem:qkprf 2}.\qed
\end{proof}


Given that there exists a $\QKPRF$ relative to $\mathcal{T}$ with probability 1 over the oracle distribution and correctness is satisfied for all oracles, there exists a $\botPRG$ $G^\mathcal{T}$, from the assumed existence of the black-box construction $\tilde{G}$, with probability 1 over the distribution of $P$ and correctness is satisfied for any oracle $\mathcal{T}\in \mathbb{T}$.

\begin{Claim}
\label{claim:adv2}
    For any QPT adversary $\adv$:
\begin{align*}
    \Pr_{\mathcal{T}\leftarrow \mathbb{T}}\left[\textsf{Adtg}^{\botPRG}_{\adv^\mathcal{T},G^\mathcal{T}}(1^\lambda)\leq O\left(\frac{1}{\lambda^4}\right)\right]\geq \frac{3}{4}.
\end{align*}
\end{Claim}

\begin{proof}
    Same as the proof of \cref{claim:adv3}.\qed
\end{proof}

Let $r=r(\lambda)$ denote the maximum run-time of $G$ and $m\coloneqq 10(\frac{r\lambda}{\mu})^4+\lambda$. Hence, $G$ makes at most $r$ queries to the oracles. 

Fix an oracle $\mathcal{T}$. We will need to show the following lemma. 

\begin{lemma}
\label{lem:unique}
There exists a set $\mathbb{G}^{\mathcal{T}_{\leq \log(2m)}}_\lambda\subseteq \{0,1\}^\lambda$ such that: 
    \begin{enumerate}
        \item $\Pr_{k\leftarrow \{0,1\}^\lambda}\left[k\in \mathbb{G}^{\mathcal{T}_{\leq \log(2m)}}_{\lambda}\right]\geq 1-\sqrt{\mu}$. 
        \item If $k\in \mathbb{G}^{\mathcal{T}_{\leq \log(2m)}}_\lambda$, then there exists a string ${y_k^{\mathcal{T}_{\leq \log(2m)}}}$ such that:
        \begin{align*}
            \Pr\left[G_{\lambda}^{\tilde{\mathcal{T}}}(k)= {y_k^{\mathcal{T}_{\leq \log(2m)}}}:{\tilde{\mathcal{T}}}\leftarrow \mathbb{T}\left[\mathcal{T}_{\leq \log(2m)}\right]\right]\geq 1-3\sqrt{\mu},
        \end{align*}
         where the probability is taken over the distribution of oracles $\tilde{\mathcal{T}}$ satisfying $\tilde{\mathcal{T}}_{\le \log(2m)}=\mathcal{T}_{\leq \log(2m)}$ and the resulting distribution of $G_{\lambda}^{{\tilde{\mathcal{T}}}}(k)$.
    \end{enumerate}
\end{lemma}

\begin{proof}

Define $\mathbb{G}^{\mathcal{T}_{\leq \log(2m)}}_\lambda\subseteq \{0,1\}^\lambda$ as the set of inputs that are in the good set $\mathcal{G}_\lambda^{\tilde{\mathcal{T}}}$ with at least $1-\sqrt{\mu}$ probability over the distribution ${\tilde{\mathcal{T}}}\leftarrow \mathbb{T}\left[\mathcal{T}_{\leq \log(2m)}\right]$. 

We first show that at least $1-\sqrt{\mu}$ fraction of inputs are in this set i.e. $\Pr_{k\leftarrow \{0,1\}^\lambda}\left[k\in \mathbb{G}^{\mathcal{T}_{\leq \log(2m)}}_{\lambda}
\right]\geq 1-\sqrt{\mu}$. Otherwise, we have that at least $\sqrt{\mu}$ fraction of inputs are in the good set with probability less than $1-\sqrt{\mu}$ over the oracle distribution. In this case, even if the rest of the inputs are in the good set for any oracle, the average size of the good set is smaller than $1-\mu$. More explicitly, we get
\begin{align*}
    \underset{\mathcal{T}\leftarrow \mathbb{T}}{\mathbb{E}}\left[\lvert \mathcal{G}^\mathcal{T}_\lambda\rvert \right]< (1-\sqrt{\mu})\cdot 1+\sqrt{\mu}\cdot (1-\sqrt{\mu})=1-\mu.
\end{align*}
This gives a contradiction since $\underset{\mathcal{T}\leftarrow \mathbb{T}}{\mathbb{E}}\left[\lvert \mathcal{G}^\mathcal{T}_\lambda\rvert \right]\geq 1-\mu$ by the pseudodeterminism of $\botPRG$s. Hence, we must have $\Pr_{k\leftarrow \{0,1\}^\lambda}\left[k\in \mathbb{G}^{\mathcal{T}_{\leq \log(2m)}}_{\lambda}\right]\geq 1-\sqrt{\mu}$. 

Now let $k\in \mathbb{G}^{\mathcal{T}_{\leq \log(2m)}}_\lambda$. By the definition of this set, 
\begin{align}
\label{eq: 2}
    \Pr\left[G_{\lambda}^{{\tilde{\mathcal{T}}}}(k)\neq \bot:{\tilde{\mathcal{T}}}\leftarrow \mathbb{T}[\mathcal{T}_{\leq \log(2m)}]\right]\geq 1-2\sqrt{\mu}.
\end{align} 

It is sufficient to show that
    \begin{align}
    \label{eq:6}
\Pr\left[ \begin{tabular}{c|c}
 \multirow{3}{*}{$\bot\neq y_1\neq y_2\neq \bot$}  &  $\ {{\mathcal{T}}'},\mathcal{T}''\leftarrow \mathbb{T}\left[\mathcal{T}_{\leq \log(2m)}\right]$ \\ &  $\ y_1\leftarrow G^{{\mathcal{T}'}}(k)$ \\ 
 & $\ y_2\leftarrow G^{{\mathcal{T}''}}(k)$\\
 \end{tabular}\right] \leq \sqrt{\mu}.
    \end{align}
If this holds, then combining this with \cref{eq: 2} implies that there exists a unique output $y_k^{\mathcal{T}_{\le \log(2m)}}$ such that 
        \begin{align*}
            \Pr\left[G_{\lambda}^{\tilde{\mathcal{T}}}(k)= {y_k^{\mathcal{T}_{\leq \log(2m)}}}:{\tilde{\mathcal{T}}}\leftarrow \mathbb{T}\left[\mathcal{T}_{\leq \log(2m)}\right]\right]\geq 1-3\sqrt{\mu}
        \end{align*}
which is the result we need to show. Assume that \cref{eq:6} does not hold. To show a contradiction, we commence with a hybrid argument. 

\begin{itemize}
    \item Hybrid $\Hy_0$: 
    \begin{enumerate}
        \item Sample an oracle ${{\mathcal{T}}'}\leftarrow \mathbb{T}\left[\mathcal{T}_{\leq \log(2m)}\right]$ and let $O'$ and $P'$ denote the functions encoded in ${\mathcal{T}}'$. 
        \item Sample $k\leftarrow \{0,1\}^\lambda$.
        \item Compute $y_1\leftarrow G^{{\mathcal{T}}'}(k).$
        \item Compute $y_2\leftarrow G^{{\mathcal{T}}'}(k).$
        \item Output $(y_1,y_2)$.
    \end{enumerate}
    
    \item Hybrid $\Hy_1$:
        \begin{enumerate}
        \item Sample oracle ${{\mathcal{T}}'}\leftarrow \mathbb{T}\left[\mathcal{T}_{\leq \log(2m)}\right]$. 
        \item Sample $k\leftarrow \{0,1\}^\lambda$.
        \item Initiate empty memory $\textbf{M}$.
        \item Compute $y_1\leftarrow G^{{\mathcal{T}}'_{\textbf{M}}}(k).$
        \item Reset $\textbf{M}$ to empty. 
        \item Compute $y_2\leftarrow G^{{\mathcal{T}}'_{\textbf{M}}}(k).$
        \item Output $(y_1,y_2)$.
    \end{enumerate}
    Here, ${\mathcal{T}}'_\textbf{M}\coloneqq (\sigma'_\textbf{M},\mathcal{O}'_\textbf{M})$ is defined as follows.
        \begin{itemize}
        \item ${\sigma}'_\textbf{M}(1^n)$:
\begin{enumerate}
\item If $n\leq \log(2m)$, then output $\sigma(1^n)$.
    \item Otherwise, sample $x\leftarrow \{0,1\}^n$. 
    \item Store $x$ in memory $\textbf{M}$.
    \item Output $\ket{x,O_n'(x)}$. 
    \end{enumerate}
    
\item ${\mathcal{O}}'_\textbf{M}$: 
\begin{enumerate}
\item If the input is of size $3n$ and $n\leq \log(2m)$, then apply $\mathcal{O}_n$.
\item Otherwise, apply the unitary of the classical function which on input $(x,y,a)$, outputs $P_n'(x,a)$ if $x \in \textbf{M}$ and $y=O_n'(x)$, and outputs $\bot$ otherwise.     
\end{enumerate}
\end{itemize}

\item Hybrid $\Hy_2$: 
        \begin{enumerate}
        \item Sample oracles ${\mathcal{T}}',\mathcal{T}''\leftarrow \mathbb{T}\left[\mathcal{T}_{\leq \log(2m)}\right]$. 
        \item Sample $k\leftarrow \{0,1\}^\lambda$.
                \item Initiate empty memory $\textbf{M}$.
        \item Compute $y_1\leftarrow G^{{\mathcal{T}'}_\textbf{M}}(k).$
        \item Reset $\textbf{M}$ to empty. 
        \item Compute $y_2\leftarrow G^{{\mathcal{T}''}_\textbf{M}}(k).$
        \item Output $(y_1,y_2)$.
    \end{enumerate}
    where $\mathcal{T}'_{\textbf{M}}$ and $\mathcal{T}''_{\textbf{M}}$ are defined in the same way as in $\Hy_1$. 

    \item Hybrid $\Hy_3$:
    \begin{enumerate}
        \item Sample ${\mathcal{T}}',\mathcal{T}''\leftarrow \mathbb{T}\left[\mathcal{T}_{\leq \log(2m)}\right]$. 
        \item Sample $k\leftarrow \{0,1\}^\lambda$.
        \item Compute $y_1\leftarrow G^{{\mathcal{T}'}}(k).$
        \item Compute $y_2\leftarrow G^{{\mathcal{T}''}}(k).$
        \item Output $(y_1,y_2)$.
    \end{enumerate}
\end{itemize}

\begin{Claim}
\label{claim:1}
With probability at least $1-\mu/8$, hybrids $\Hy_0$ and $\Hy_1$ are indistinguishable. 
\end{Claim}
\begin{proof}
The oracles in these two hybrids only differ on inputs starting with $(x,O'(x))$ such that $(x,O'(x))$ was not the response of a query to $\sigma$ by $G$ in the same evaluation. Crucially, $O'$ is a random function. Therefore, to distinguish these two oracles, $G$ needs to do unstructured search for an input starting with $x\notin \textbf{M}$ that maps to a non-$\bot$ element. 
    
The lower bound for unstructured search \cite{Z90,BBB+97} states that $G$ cannot distinguish these two hybrids with better than $O(\frac{r^2}{2^{n}}) \leq \mu/8$ probability.
\qed
\end{proof}

\begin{Claim}
With probability at least $1-\mu/8$, hybrids $\Hy_1$ and $\Hy_2$ are indistinguishable. 
\end{Claim}

\begin{proof}
Consider the two evaluations in $\Hy_1$. Let $(x^1_i,O'(x^1_i))_{i\in \left[r\right]}$ and $(x_j^2,O'(x^2_j))_{j\in \left[r\right]}$ denote the responses of oracles $(\sigma'_n)_{n>\log(2m)}$ to the queries of $G$ in the first and second evaluation, respectively. 

If $\{x^1_i\}_{i\in \left[r\right]}\cap \{x^2_j\}_{j\in \left[r\right]}=\emptyset$, then these two hybrids are indistinguishable, since $O',O'',P',P''$ are random functions. This scenario occurs with at least $1-\frac{(2r)^2}{m}\geq 1-\mu/8$ probability by the birthday problem. 
    \qed
\end{proof}

\begin{Claim}
With probability at least $1-\mu/8$, hybrids $\Hy_2$ and $\Hy_3$ are indistinguishable. 
\end{Claim}

\begin{proof}
   This follows in the same way as \cref{claim:1}.
    \qed
\end{proof}

By the above three claims and the triangle inequality, we have that with probability at least $1-\mu/2$, hybrids $\Hy_0$ and $\Hy_3$ are indistinguishable. 

Notice that in $\Hy_3$, by our assumption, the probability that $(y_1,y_2)$ are non-$\bot$ distinct strings is at least $\sqrt{\mu}$. Therefore, the probability that the two strings generated in hybrid $\Hy_0$ also are non-$\bot$ distinct strings is $\sqrt{\mu}-\mu/2> \mu/2$. However, this contradicts the pseudodeterminism condition of $G$, since for a fixed oracle, two evaluations should yield the same string or $\bot$ except with negligible probability.  
    \qed
\end{proof}

 We are now ready to prove the main result (\cref{lem:no botowsg}) using a hybrid argument. But first, we consider a generator $\tilde{G}$ that only depends on $\mathcal{T}_{\leq \log(2m)}$ and is defined as follows. 

\smallskip \noindent\fbox{%
    \parbox{\textwidth}{%
    $\tilde{G}^{\mathcal{T}_{\leq \log(2m)}}(k)$:
\begin{itemize}
\item For $j\in [\lambda]:$
\begin{enumerate}
    \item For each $\log(2m)\leq i\leq r$, sample uniformly at random two $2r$-degree polynomials $\tilde{O}_i:\mathbb{F}_{2^i}\rightarrow \mathbb{F}_{2^{i}}$ and $\tilde{P}_i:\mathbb{F}_{2^{2i}}\rightarrow \mathbb{F}_{2^i}$. Let $(\tilde{\sigma}_i,\tilde{\mathcal{O}}_i)$ be the resulting oracles.  
   \item Run $G(k)$ and answer the queries as follows:
  \begin{enumerate}
  \item For a query $x$ of length $n\leq \log(2m)$, respond using $(\sigma_n,{\mathcal{O}_n})$. 
  \item For a query $x$ of length $n> \log(2m)$, respond using $(\tilde{\sigma}_n,\tilde{\mathcal{O}}_n)$.
  \end{enumerate}
        \item Let $y_j$ be the result of $G(k_1)$.
        \end{enumerate}
\item Set $y=\textsf{vote}_\lambda(y_{1},\ldots, y_{\lambda})$.
        \item Output $y$. 
\end{itemize}}}
    \smallskip

We also consider a generator $\overline{G}$ that does not depend on the oracles, and is defined as follows on inputs of length $\lambda+16m^3$:

\smallskip \noindent\fbox{%
    \parbox{\textwidth}{%
    $\overline{G}(k)$:
\begin{itemize}
\item Parse $k$ as $(k_1,k_2)$ where $k_1\in \{0,1\}^\lambda$ and $k_2\in \{0,1\}^{16m^3}$.
\item Construct functions $(\sigma_n^{k_2},{\mathcal{O}_n^{k_2}})_{n\leq \log(2m)}$ in the same way as \cref{con:oracles 2} but with the randomness determined by $k_2$.
\item For $j\in [\lambda]:$
\begin{enumerate}
    \item For each $\log(2m)\leq i\leq r$, sample uniformly at random two $2r$-degree polynomials $\tilde{O}_i:\mathbb{F}_{2^i}\rightarrow \mathbb{F}_{2^{i}}$ and $\tilde{P}_i:\mathbb{F}_{2^{2i}}\rightarrow \mathbb{F}_{2^i}$. Let $(\tilde{\sigma}_i,\tilde{\mathcal{O}}_i)$ be the resulting oracles.  
   \item Run $G(k_1)$ and answer the queries as follows:
  \begin{enumerate}
  \item For a query $x$ of length $n\leq \log(2m)$, respond using $(\sigma_n^{k_2},{\mathcal{O}_n^{k_2}})$. 
  \item For a query $x$ of length $n> \log(2m)$, respond using $(\tilde{\sigma}_n,\tilde{\mathcal{O}}_n)$.
  \end{enumerate}
        \item Let $y_j$ be the result of $G(k_1)$.
        \end{enumerate}
\item Set $y=\textsf{vote}_\lambda(y_{1},\ldots, y_{\lambda})$.
        \item Output $(y,k_2)$. 
\end{itemize}}}
    \smallskip

Now consider the following variants of $\botPRG$ security experiment, where we set $q(\lambda)=1$ (see \cref{def:botprg}).

\begin{itemize}
    \item $\textsf{Exp}^\adv_1(\lambda)$: 
    \begin{enumerate}
        \item Sample oracle $\mathcal{T}$ as in \cref{con:oracles 1}.
        \item $b\leftarrow \textsf{Exp}^{\botPRG}_{\adv^{\mathcal{T}},G^\mathcal{T}}(1^\lambda,1^q)$. 
        \item Output $b$. 
    \end{enumerate}
        \item $\textsf{Exp}^\adv_2(\lambda)$: 
    \begin{enumerate}
        \item Sample oracle $\mathcal{T}$ as in \cref{con:oracles 2}.
        \item $b\leftarrow \textsf{Exp}^{\botPRG}_{\adv^{\mathcal{T}_{\leq\log(2m)},\mathcal{C}},G^\mathcal{T}}(1^\lambda,1^q)$. \textit{Notice that $\adv$ only has access to $\mathcal{T}_{\leq \log(2m)}$ and $\mathcal{C}$ in this experiment.}
        \item Output $b$.
    \end{enumerate}

     \item $\textsf{Exp}^\adv_3(\lambda)$: 
    \begin{enumerate}
        \item Sample oracle $\mathcal{T}$ as in \cref{con:oracles 2}.
        \item $b\leftarrow \textsf{Exp}^{\botPRG}_{\adv^{\mathcal{T}_{\leq\log(2m)},\mathcal{C}},\tilde{G}^{\mathcal{T}_{\leq \log(2m)}}}(1^\lambda,1^q)$. 
        \item Output $b$.
    \end{enumerate}
    
    \item $\textsf{Exp}^\adv_4(\lambda)$:
    \begin{enumerate}
       \item $b\leftarrow \textsf{Exp}^{\botPRG}_{\adv^\mathcal{C},\overline{G}}(1^\lambda,1^q)$. 
        \item Output $b$.
        \end{enumerate}      
\end{itemize}

\remark{{Technically, $\overline{G}$ and $\tilde{G}$ do not satisfy the pseudodeterminism conditions of a $\botPRG$, but we can still run the $\botPRG$ security experiment on them.} }

\begin{Claim}
For any QPT adversary $\adv$, there exists a QPT adversary $\mathcal{B}$ such that
\begin{align*}
  \Pr\left[\textsf{Exp}^\adv_2(\lambda)=1\right]\leq \Pr\left[\textsf{Exp}^{\mathcal{B}}_1(\lambda)=1\right].
\end{align*}
\end{Claim}

\begin{proof}
    This is clear because the only difference between theses experiments is that the adversary's access to the oracle in $\textsf{Exp}_2$ is restricted. 
    \qed
\end{proof}

\begin{Claim}
For any QPT adversary $\adv$ and large enough $\lambda$,
\begin{align*}
  d_{\textsf{TD}}(\textsf{Exp}^\adv_2(\lambda),\textsf{Exp}^\adv_3(\lambda)) \leq 1/\lambda^4.
\end{align*}
\end{Claim}

\begin{proof}
By \cref{lem:unique}, for any oracle $\mathcal{T}$, there exists a set $\mathbb{G}^{\mathcal{T}_{\leq \log(2m)}}_\lambda$ such that: 
    \begin{enumerate}
   \item $\Pr_{k\leftarrow \{0,1\}^\lambda}\left[k\in \mathbb{G}^{\mathcal{T}_{\leq \log(2m)}}_{\lambda}\right]\geq 1-\sqrt{\mu}$. 
        \item If $k\in \mathbb{G}^{\mathcal{T}_{\leq \log(2m)}}_\lambda$, then there exists a string ${y_k^{\mathcal{T}_{\leq \log(2m)}}}$ such that:
        \begin{align}
        \label{eq:7}
            \Pr\left[G_{\lambda}^{\tilde{\mathcal{T}}}(k)= {y_k^{\mathcal{T}_{\leq \log(2m)}}}:{\tilde{\mathcal{T}}}\leftarrow \mathbb{T}\left[\mathcal{T}_{\leq \log(2m)}\right]\right]\geq 1-3\sqrt{\mu}.
        \end{align}
    \end{enumerate}

By definition of $\botPRG$s, for an input $k\in \{0,1\}^\lambda$, there exists a string $y^\mathcal{T}_k$ such that $\Pr\left[G_\lambda^\mathcal{T}(k)\in \{y^\mathcal{T}_k,\bot\}\right]\geq 1-\negl[\lambda]$. As usual, $\mathcal{G}_\lambda^\mathcal{T}$ denotes the good set of inputs for $G_\lambda^\mathcal{T}$ (see \cref{def:botprg}).  

Let $B$ denote the event that the key $k$ and oracle $\mathcal{T}$ sampled in $\textsf{Exp}_2^\adv(\lambda)$ or $\textsf{Exp}_3^\adv(\lambda)$ satisfy the following conditions: $k\in \mathbb{G}^{\mathcal{T}_{\leq \log(2m)}}_\lambda\cap \mathcal{G}_\lambda^\mathcal{T}$ and ${y_k^{\mathcal{T}_{\leq \log(2m)}}}={y_k^\mathcal{T}}$, where ${y_k^{\mathcal{T}_{\leq \log(2m)}}}$ is the string that satisfies \cref{eq:7}. 

We now show that event $B$ occurs with probability at least $1-6\sqrt{\mu}$. Note that 
\begin{align*}
    \Pr\left[k\in \mathcal{G}_\lambda^\mathcal{T}:\begin{matrix}k\leftarrow \{0,1\}^\lambda\\ \mathcal{T}\leftarrow \mathbb{T}\end{matrix}\right]&\geq 1-{\mu}\\
    \Pr\left[k\in \mathbb{G}^{\mathcal{T}_{\leq \log(2m)}}_\lambda:\begin{matrix}k\leftarrow \{0,1\}^\lambda\\ \mathcal{T}\leftarrow \mathbb{T}\end{matrix}\right]&\geq 1-\sqrt{\mu}.
\end{align*}
Therefore, $\Pr\left[k\in \mathcal{G}_\lambda^\mathcal{T}\cap \mathbb{G}^{\mathcal{T}_{\leq \log(2m)}}_\lambda:k\leftarrow \{0,1\}^\lambda, \mathcal{T}\leftarrow \mathbb{T}\right]\geq 1-2\sqrt{\mu}$. Given this, to show that event $B$ occurs with at least $1-6\sqrt{\mu}$ probability, it is sufficient to show that ${y_k^{\mathcal{T}_{\leq \log(2m)}}}={y_k^\mathcal{T}}$ occurs with at least $1-4\sqrt{\mu}$ probability when $k\leftarrow \mathcal{G}_\lambda^\mathcal{T}\cap \mathbb{G}^{\mathcal{T}_{\leq \log(2m)}}_\lambda$ and $\mathcal{T}\leftarrow \mathbb{T}$. If this does not hold, then by an averaging argument, we obtain 
    \begin{align*}
\Pr\left[G_{\lambda}^{\tilde{\mathcal{T}}}(k)= {y_k^{\mathcal{T}_{\leq \log(2m)}}}: \begin{matrix} {{\mathcal{T}}}\leftarrow \mathbb{T} \\ {\tilde{\mathcal{T}}}\leftarrow \mathbb{T}\left[\mathcal{T}_{\leq \log(2m)}\right]  \\
  \ k\leftarrow \mathbb{G}^{\mathcal{T}_{\leq \log(2m)}}_\lambda
 \end{matrix}\right] =\\
 \Pr\left[G_{\lambda}^{{\mathcal{T}}}(k)= {y_k^{\mathcal{T}_{\leq \log(2m)}}}: \begin{matrix} {{\mathcal{T}}}\leftarrow \mathbb{T} \\ 
  \ k\leftarrow \mathbb{G}^{\mathcal{T}_{\leq \log(2m)}}_\lambda
 \end{matrix}\right] \leq \\
 \left(1-2\sqrt{\mu}\right)\Pr\left[G_{\lambda}^{{\mathcal{T}}}(k)= {y_k^{\mathcal{T}_{\leq \log(2m)}}}: \begin{matrix} {{\mathcal{T}}}\leftarrow \mathbb{T} \\ 
  \ k\leftarrow \mathbb{G}^{\mathcal{T}_{\leq \log(2m)}}_\lambda \cap \mathcal{G}_\lambda^\mathcal{T}
 \end{matrix}\right] +2\sqrt{\mu}\leq \\
 \left(1-2\sqrt{\mu}\right)\left(1-4\sqrt{\mu}\right)+2\sqrt{\mu}<1-3\sqrt{\mu} 
\end{align*}
which contradicts \cref{eq:7}, as this equation states that the first probability should be at least $1-3\sqrt{\mu}$. Therefore, event $B$ occurs with at least $1-6\sqrt{\mu}$ probability.


If event $ B$ occurs, then in $\textsf{Exp}^\adv_2(\lambda)$, the generator $G_{\lambda}^{{\mathcal{T}}}(k)$ returns ${y_k^{\mathcal{T}_{\leq \log(2m)}}}$ with probability $1-\negl[\lambda]$ for every query. While, in $\textsf{Exp}^\adv_3(\lambda)$, $\overline{G}_\lambda(k)$ returns ${y_k^{\mathcal{T}_{\leq \log(2m)}}}$ with probability $1-\negl[\lambda]$. 

Overall, if event $ B$ occurs, then the two experiments can only be distinguished with at most negligible probability. Given that event $B$ occurs with probability at least $1-6\sqrt{\mu}$, for large enough $\lambda$ these two experiments can be distinguished with at most $6\sqrt{\mu}+\negl[\lambda]<\frac{1}{\lambda^4}$ probability.
    \qed
\end{proof}

\begin{Claim}
For any QPT adversary $\adv$, there exists a QPT algorithm $\mathcal{B}$ such that
    \begin{align*}
        \Pr\left[\textsf{Exp}^\mathcal{A}_4(\lambda)=1\right]&\leq \Pr\left[\textsf{Exp}^{\mathcal{B}}_3(\lambda)=1\right]+3\sqrt{\mu}.
    \end{align*}
\end{Claim}

\begin{proof}
Fix an adversary $\adv$ in $\textsf{Exp}_4^\adv(\lambda)$. We construct an algorithm $\mathcal{B}$ in $\textsf{Exp}_3^{\mathcal{B}}(\lambda)$ as follows. 

In $\textsf{Exp}_3^{\mathcal{B}}(\lambda)$, an input $k_1\leftarrow \{0,1\}^\lambda$, a bit $b\leftarrow \{0,1\}$, and a string $y\leftarrow \{0,1\}^s$ are sampled. Then, let $y_0\leftarrow \tilde{G}^{\mathcal{T}_{\leq \log(2m)}}(k_1)$ and $y_1\leftarrow \isbot(\tilde{G}^{\mathcal{T}_{\leq \log(2m)}}(k_1),y)$. $\mathcal{B}^{\mathcal{T}_{\leq \log(2m)},\mathcal{C}}$ receives $y_b$ and must guess $b$. 

$\mathcal{B}^{\mathcal{T}_{\leq \log(2m)},\mathcal{C}}$ commences as follows. For every $n\leq \log(2m)$, it queries $\sigma_n(1^n)$ $4m^2$ times. This allows $\mathcal{B}$ to obtain all the evaluations of $O_n$ and learn the function entirely, except with negligible probability. Next, for each $n\leq \log(2m)$, $\mathcal{B}$ uses $O_n$ and the oracle $\mathcal{O}_n$ to learn $P_n$ entirely, which requires at most $2m$ queries. 

$\mathcal{B}$ encodes $(P_n,O_n)_{n\leq \log(2m)}$ into a string, say $k_2$, of length less than $16m^3$. $\mathcal{B}$ runs $\adv^\mathcal{C}$ on $(y_b,k_2)$ and receives a response $b'$. $\mathcal{B}$ outputs $b'$.  

As long as $\mathcal{B}$ encodes $(P_n,O_n)_{n\leq \log(2m)}$ correctly (which occurs with $1-\negl[\lambda]$ probability) and $\Pr[\overline{G}(k_1,k_2)=(y_0,k_2)]\geq 1-\negl[\lambda]$, then it is clear that $\Pr\left[\textsf{Exp}^{\mathcal{B}}_3(\lambda)=1\right]$ is at least $\Pr\left[\textsf{Exp}^\mathcal{A}_4(\lambda)=1\right]-\negl[\lambda]$. If $k_1\in \mathbb{G}_\lambda^{\mathcal{T}_{\leq \log(2m)}}$, then these conditions occur with $1-\negl[\lambda]$ probability. Recall $k_1\in \mathbb{G}_\lambda^{\mathcal{T}_{\leq \log(2m)}}$ occurs with probability at least $1-\sqrt{\mu}$ by \cref{lem:unique}. All together, this means  
\begin{align*}
\Pr\left[\textsf{Exp}^{\mathcal{B}}_3(\lambda)=1\right]\geq \Pr\left[\textsf{Exp}^\mathcal{A}_4(\lambda)=1\right]\left(1-\negl[\lambda]\right)\left(1-2\sqrt{\mu}\right),
\end{align*}
which implies, for large enough $\lambda$, 
\begin{align*}
    \Pr\left[\textsf{Exp}^\mathcal{A}_3(\lambda)=1\right]\leq \Pr\left[\textsf{Exp}^{\mathcal{B}}_2(\lambda)=1\right]+3\sqrt{\mu}.
\end{align*}
    \qed
\end{proof}

By the triangle inequality, for any QPT adversary $\adv$ and large enough $\lambda$, there exists a QPT adversary $\mathcal{B}$ such that
\begin{align}
\label{eq:4}
   \left| \Pr\left[\textsf{Exp}^\adv_4(\lambda)=1\right]-\Pr\left[\textsf{Exp}^{\mathcal{B}}_0(\lambda)=1\right]\right|\leq \frac{1}{\lambda^4}+3\sqrt{\mu}+\negl[\lambda].
\end{align}

By \cref{claim:adv2}, for any QPT adversary $\adv$, there exists a negligible function $\delta$ such that:
\begin{align*}
    \Pr_{\mathcal{T}\leftarrow \mathbb{T}}\left[\textsf{Adtg}^{\botPRG}_{\adv^\mathcal{T},G^\mathcal{T}}(1^\lambda)\leq O\left(\frac{1}{\lambda^4}\right)\right]\geq \frac{3}{4}.
\end{align*}

Therefore, for any  QPT adversary $\adv$ and large enough $\lambda$,
\begin{align*}
    \Pr\left[\textsf{Exp}^\adv_0(\lambda)=1\right]-\frac{1}{2}\leq \frac{3}{4}\cdot \frac{1}{\lambda^3}+\frac{1}{4}\cdot \frac{1}{2}.
\end{align*}
By \cref{eq:4}, for any QPT adversary $\adv$ and large enough $\lambda$,
\begin{align}
\label{eq:5}
        \Pr\left[\textsf{Exp}^\adv_4(\lambda)=1\right]\leq  \frac{3}{4}+3\sqrt{\mu}+\frac{1}{\lambda^2}.
\end{align}
Notice that $\overline{G}$ does not use any oracle access.

On the other hand, we will show that there exists an adversary that contradicts \cref{eq:5}.  

\begin{claim}
There exists a QPT algorithm $\overline{\adv}$ such that
    \begin{align*}
\Pr\left[\textsf{Exp}^{\overline{\adv}}_4(\lambda)=1\right]\geq 1-2\sqrt{\mu}.
\end{align*}
\end{claim}

\begin{proof}
In the security experiment, $\textsf{Exp}^{\overline{\adv}}_4(\lambda)$, an input $k\leftarrow \{0,1\}^{\lambda+16m^3}$, a bit $b\leftarrow \{0,1\}$ and string $y\leftarrow \{0,1\}^{\ell}$ are sampled, where $\ell\coloneqq s+16m^3$. Then, sample $y_0\leftarrow \overline{G}(k)$ and $y_1\leftarrow \isbot(\overline{G}(k),y)$. $\overline{\adv}$ receives $y_b$ and needs to guess $b$. Interpret $y_b$ as $(y_b',k_2)$, where $y_b'\in \{0,1\}^s$ and $k_2\in \{0,1\}^{16m^3}$.

$\overline{\adv}$ uses oracle access to $\mathcal{C}$ to run the algorithm which computes $\overline{G}(k,k_2)$ on every input $k\in \{0,1\}^{\lambda}$ and if any computation yields $y_b$, then it outputs 0 and otherwise outputs 1. $\overline{\adv}$ returns the output of this algorithm. 

Let $I_\lambda\coloneqq \{y: \exists k\in \{0,1\}^{\lambda} \ \text{s.t}\ \Pr[\overline{G}(k,k_2)=y]\geq \frac{1}{2^{3\lambda/2}}\}$. Notice that $\lvert I_\lambda\rvert \leq 2^{5\lambda/2}$.

In the case $b=1$, note that $\Pr[y_1\in I_\lambda]\leq \frac{2^{5\lambda/2}}{2^s}\leq \frac{1}{2^{\lambda/2}}$ is negligible. This implies that there is a negligible probability that $\overline{\adv}^\mathcal{C}(y_1)$ outputs 0 by the union bound. Therefore, $\overline{\adv}$ guesses $b'=b$ correctly except with negligible probability when $b=1$. 

Meanwhile, if $b=0$, then $y_0\leftarrow \overline{G}(k)$. Since $k$ is sampled uniformly at random, by \cref{lem:unique}, 
\begin{align*}
\Pr[\overline{G}(k)=y_0]\geq (1-\negl[\lambda])(1-\sqrt{\mu})\geq 1-2\sqrt{\mu}    
\end{align*}
In other words, $\overline{\adv}^\mathcal{C}$ outputs $b'=b$ with probability $1-2\sqrt{\mu}$ when $b=0$. 

All in all,
\begin{align*}
\Pr\left[\textsf{Exp}^{\botPRG}_{\overline{\adv}^{\mathcal{C}},\overline{G}}(1^\lambda)=1\right]\geq (1-\negl[\lambda])\cdot \frac{1}{2} +(1-2\sqrt{\mu})\cdot \frac{1}{2}\geq 1-2\sqrt{\mu}.
\end{align*}
\qed
\end{proof}

The claim above contradicts \cref{eq:5}. So there does not exist a fully black-box construction of a $(\mu,s)$-$\botPRG$ from a $\QKPRF$. 
\qed
\end{proof}


\printbibliography

@article{KQT24,
  title={Quantum-Computable One-Way Functions without One-Way Functions},
  author={Kretschmer, William and Qian, Luowen and Tal, Avishay},
  journal={arXiv preprint arXiv:2411.02554},
  year={2024}
}

@article{GMM+24,
  title={CountCrypt: Quantum Cryptography between QCMA and PP},
  author={Goldin, Eli and Morimae, Tomoyuki and Mutreja, Saachi and Yamakawa, Takashi},
  journal={arXiv preprint arXiv:2410.14792},
  year={2024}
}

@inproceedings{Z20,
  title={Schr{\"o}dinger’s pirate: How to trace a quantum decoder},
  author={Zhandry, Mark},
  booktitle={Theory of Cryptography: 18th International Conference, TCC 2020, Durham, NC, USA, November 16--19, 2020, Proceedings, Part III 18},
  pages={61--91},
  year={2020},
  organization={Springer}
}

@inproceedings{U12,
  title={Quantum proofs of knowledge},
  author={Unruh, Dominique},
  booktitle={Annual international conference on the theory and applications of cryptographic techniques},
  pages={135--152},
  year={2012},
  organization={Springer}
}

@misc{B21,
      author = {James Bartusek},
      title = {Secure Quantum Computation with Classical Communication},
      howpublished = {Cryptology {ePrint} Archive, Paper 2021/964},
      year = {2021},
      url = {https://eprint.iacr.org/2021/964}
}

@inproceedings{BCN25,
  title={Oracle separation between quantum commitments and quantum one-wayness},
  author={Bostanci, John and Chen, Boyang and Nehoran, Barak},
  booktitle={Annual International Conference on the Theory and Applications of Cryptographic Techniques},
  pages={3--22},
  year={2025},
  organization={Springer}
}

@article{BBB+97,
  title={Strengths and weaknesses of quantum computing},
  author={Bennett, Charles H and Bernstein, Ethan and Brassard, Gilles and Vazirani, Umesh},
  journal={SIAM journal on Computing},
  volume={26},
  number={5},
  pages={1510--1523},
  year={1997},
  publisher={SIAM}
}

@article{CCS24,
  title={The power of a single Haar random state: constructing and separating quantum pseudorandomness},
  author={Chen, Boyang and Coladangelo, Andrea and Sattath, Or},
  journal={arXiv preprint arXiv:2404.03295},
  year={2024}
}

@article{Z90,
  title={Grover’s quantum searching algorithm is optimal},
  author={Zalka, Christof},
  journal={Physical Review A},
  volume={60},
  number={4},
  pages={2746},
  year={1999},
  publisher={APS}
}

@inproceedings{BKN+23,
  title={Obfuscation of pseudo-deterministic quantum circuits},
  author={Bartusek, James and Kitagawa, Fuyuki and Nishimaki, Ryo and Yamakawa, Takashi},
  booktitle={Proceedings of the 55th Annual ACM Symposium on Theory of Computing},
  pages={1567--1578},
  year={2023}
}

@inproceedings{BBV24,
  title={Quantum state obfuscation from classical oracles},
  author={Bartusek, James and Brakerski, Zvika and Vaikuntanathan, Vinod},
  booktitle={Proceedings of the 56th Annual ACM Symposium on Theory of Computing},
  pages={1009--1017},
  year={2024}
}

@inproceedings{U16,
  title={Computationally binding quantum commitments},
  author={Unruh, Dominique},
  booktitle={Advances in Cryptology--EUROCRYPT 2016: 35th Annual International Conference on the Theory and Applications of Cryptographic Techniques, Vienna, Austria, May 8-12, 2016, Proceedings, Part II 35},
  pages={497--527},
  year={2016},
  organization={Springer}
}

@article{FK15,
  title={Quantum vs classical proofs and subset verification},
  author={Fefferman, Bill and Kimmel, Shelby},
  journal={arXiv preprint arXiv:1510.06750},
  year={2015}
}

@article{CM24,
  title={On black-box separations of quantum digital signatures from pseudorandom states},
  author={Coladangelo, Andrea and Mutreja, Saachi},
  journal={arXiv preprint arXiv:2402.08194},
  year={2024}
}

@article{GGM86,
  author       = {Oded Goldreich and
                  Shafi Goldwasser and
                  Silvio Micali},
  title        = {How to construct random functions},
  journal      = {J. {ACM}},
  volume       = {33},
  number       = {4},
  pages        = {792--807},
  year         = {1986},
  url          = {https://doi.org/10.1145/6490.6503},
  doi          = {10.1145/6490.6503},
}

@inproceedings{BMM+25,
  title={A new world in the depths of microcrypt: Separating OWSGs and quantum money from QEFID},
  author={Behera, Amit and Malavolta, Giulio and Morimae, Tomoyuki and Mour, Tamer and Yamakawa, Takashi},
  booktitle={Annual International Conference on the Theory and Applications of Cryptographic Techniques},
  pages={23--52},
  year={2025},
  organization={Springer}
}

@inproceedings{AGQ22,
  title={Pseudorandom (Function-Like) Quantum State Generators: New Definitions and Applications},
  author={Ananth, Prabhanjan and Gulati, Aditya and Qian, Luowen and Yuen, Henry},
  booktitle={Theory of Cryptography: 20th International Conference, TCC 2022, Chicago, IL, USA, November 7--10, 2022, Proceedings, Part I},
  pages={237--265},
  year={2022},
  organization={Springer}
}

@inproceedings{K21,
  title={Quantum Pseudorandomness and Classical Complexity},
  author={Kretschmer, William},
  booktitle={16th Conference on the Theory of Quantum Computation, Communication and Cryptography},
  year={2021}
}

@inproceedings{AQY22,
  title={Cryptography from pseudorandom quantum states},
  author={Ananth, Prabhanjan and Qian, Luowen and Yuen, Henry},
  booktitle={Advances in Cryptology--CRYPTO 2022: 42nd Annual International Cryptology Conference, CRYPTO 2022, Santa Barbara, CA, USA, August 15--18, 2022, Proceedings, Part I},
  pages={208--236},
  year={2022},
  organization={Springer}
}

@article{BBS23,
  title={Pseudorandomness with proof of destruction and applications},
  author={Behera, Amit and Brakerski, Zvika and Sattath, Or and Shmueli, Omri},
  journal={Cryptology ePrint Archive},
  year={2023}
}

@article{Z16,
  title={A note on quantum-secure PRPs},
  author={Zhandry, Mark},
  journal={arXiv preprint arXiv:1611.05564},
  year={2016}
}

@article{MH24,
  title={How to construct random unitaries},
  author={Ma, Fermi and Huang, Hsin-Yuan},
  journal={arXiv preprint arXiv:2410.10116},
  year={2024}
}

@inproceedings{BS20,
  title={Scalable pseudorandom quantum states},
  author={Brakerski, Zvika and Shmueli, Omri},
  booktitle={Annual International Cryptology Conference},
  pages={417--440},
  year={2020},
  organization={Springer}
}

@article{CGG+23,
  title={On the computational hardness of quantum one-wayness},
  author={Cavalar, Bruno and Goldin, Eli and Gray, Matthew and Hall, Peter and Liu, Yanyi and Pelecanos, Angelos},
  journal={arXiv preprint arXiv:2312.08363},
  year={2023}
}

@inproceedings{SXY18,
  title={Tightly-secure key-encapsulation mechanism in the quantum random oracle model},
  author={Saito, Tsunekazu and Xagawa, Keita and Yamakawa, Takashi},
  booktitle={Advances in Cryptology--EUROCRYPT 2018: 37th Annual International Conference on the Theory and Applications of Cryptographic Techniques, Tel Aviv, Israel, April 29-May 3, 2018 Proceedings, Part III 37},
  pages={520--551},
  year={2018},
  organization={Springer}
}

@misc{MYY24,
      author = {Tomoyuki Morimae and Shogo Yamada and Takashi Yamakawa},
      title = {Quantum Unpredictability},
      howpublished = {Cryptology {ePrint} Archive, Paper 2024/701},
      year = {2024},
      url = {https://eprint.iacr.org/2024/701}
}

@article{MY22b,
  title={One-wayness in quantum cryptography},
  author={Morimae, Tomoyuki and Yamakawa, Takashi},
  journal={arXiv preprint arXiv:2210.03394},
  year={2022}
}

@inproceedings{KT24,
  title={Commitments from quantum one-wayness},
  author={Khurana, Dakshita and Tomer, Kabir},
  booktitle={Proceedings of the 56th Annual ACM Symposium on Theory of Computing},
  pages={968--978},
  year={2024}
}

@inproceedings{CGG24,
  title={On central primitives for quantum cryptography with classical communication},
  author={Chung, Kai-Min and Goldin, Eli and Gray, Matthew},
  booktitle={Annual International Cryptology Conference},
  pages={215--248},
  year={2024},
  organization={Springer}
}

@inproceedings{Z12,
  title={How to construct quantum random functions},
  author={Zhandry, Mark},
  booktitle={2012 IEEE 53rd Annual Symposium on Foundations of Computer Science},
  pages={679--687},
  year={2012},
  organization={IEEE}
}

@article{MMN+16,
  title={A note on black-box separations for indistinguishability obfuscation},
  author={Mahmoody, Mohammad and Mohammed, Ameer and Nematihaji, Soheil and Pass, Rafael and others},
  journal={Cryptology ePrint Archive},
  year={2016}
}

@inproceedings{MY22a,
  title={Quantum commitments and signatures without one-way functions},
  author={Morimae, Tomoyuki and Yamakawa, Takashi},
  booktitle={Advances in Cryptology--CRYPTO 2022: 42nd Annual International Cryptology Conference, CRYPTO 2022, Santa Barbara, CA, USA, August 15--18, 2022, Proceedings, Part I},
  pages={269--295},
  year={2022},
  organization={Springer}
}

@book{NC00,
	address = {New York, NY, USA},
	author = {Nielsen, Michael A. and Chuang, Isaac L.},
	doi = {10.1017/CBO9780511976667},
	isbn = {0-521-63503-9},
	publisher = {Cambridge University Press},
	title = {Quantum computation and quantum information},
	year = {2000}}

@inproceedings{IR89,
  title={Limits on the provable consequences of one-way permutations},
  author={Impagliazzo, Russell and Rudich, Steven},
  booktitle={Proceedings of the twenty-first annual ACM symposium on Theory of computing},
  pages={44--61},
  year={1989}
}

@inproceedings{JLS18,
  title={Pseudorandom quantum states},
  author={Ji, Zhengfeng and Liu, Yi-Kai and Song, Fang},
  booktitle={Advances in Cryptology--CRYPTO 2018: 38th Annual International Cryptology Conference, Santa Barbara, CA, USA, August 19--23, 2018, Proceedings, Part III 38},
  pages={126--152},
  year={2018},
  organization={Springer}
}

@article{ALY23,
  title={Pseudorandom Strings from Pseudorandom Quantum States},
  author={Ananth, Prabhanjan and Lin, Yao-Ting and Yuen, Henry},
  journal={arXiv preprint arXiv:2306.05613},
  year={2023}
}

@inproceedings{DIJ+09,
  title={Security amplification for interactive cryptographic primitives},
  author={Dodis, Yevgeniy and Impagliazzo, Russell and Jaiswal, Ragesh and Kabanets, Valentine},
  booktitle={Theory of Cryptography: 6th Theory of Cryptography Conference, TCC 2009, San Francisco, CA, USA, March 15-17, 2009. Proceedings 6},
  pages={128--145},
  year={2009},
  organization={Springer}
}

@inproceedings{DMP20,
  title={A combinatorial approach to quantum random functions},
  author={D{\"o}ttling, Nico and Malavolta, Giulio and Pu, Sihang},
  booktitle={Advances in Cryptology--ASIACRYPT 2020: 26th International Conference on the Theory and Application of Cryptology and Information Security, Daejeon, South Korea, December 7--11, 2020, Proceedings, Part II 26},
  pages={614--632},
  year={2020},
  organization={Springer}
}

@misc{BBO+24,
      title={Signatures From Pseudorandom States via $\bot$-PRFs}, 
      author={Mohammed Barhoush and Amit Behera and Lior Ozer and Louis Salvail and Or Sattath},
      year={2024},
      eprint={2311.00847},
      archivePrefix={arXiv},
      primaryClass={cs.CR},
      url={https://arxiv.org/abs/2311.00847}, 
}

\appendix

\section{\texorpdfstring{\textsf{BQ-PRU}\textsuperscript{qs} from \textsf{PRG}\textsuperscript{qs}}{BQ-PRUqs from PRGqs}}
\label{sec:pru-from-prg}
In this section, we show to build $\BQPRU$s from $\QKPRG$s. We first introduce some further definitions in the quantum input sampling regime. 

\subsection{\texorpdfstring{Definitions: \textsf{PRP}\textsuperscript{qs} and \textsf{BQ-PRU}\textsuperscript{qs}}{Definitions: PRPqs and BQ-PRUqs}}

We introduce pseudorandom permutations with quantum input sampling. 

\begin{definition}[Pseudorandom Permutation with Quantum Key Generation]
    Let $\lambda\in \mathbb{N}$ be the security parameter and let $n=n(\lambda)$ and $m=m(\lambda)$ be polynomials in $\lambda$. A tuple of QPT algorithms $(\textsf{QSamp},F,F^{-1})$ is called a $(m,n)$-\emph{pseudorandom permutation with quantum key generation ($\QKPRP$)}, if:
    \begin{enumerate}
        \item $\textsf{QSamp}(1^\lambda):$ Outputs a string $k\in \{0,1\}^m$.
        \item $F_k(x)$: Takes a key $k\in \{0,1\}^m$ and an input $x\in \{0,1\}^n$ and outputs a string $y\in \{0,1\}^n$.
                \item $F^{-1}_k(y)$: Takes a key $k\in \{0,1\}^m$ and an input $y\in \{0,1\}^n$ and outputs a string $x\in \{0,1\}^n$.
                \item (\emph{Inverse Relation}) For every $k\in \{0,1\}^m$, there exists a permutation $\pi_k$ over $\{0,1\}^n$ such that for all $x,y\in \{0,1\}^n$, the following conditions are satisfied:
$$
\Pr_{k\leftarrow \textsf{QSamp}(1^\lambda)}\left[F_k(x)=\pi_{k}(x)\right] \geq 1-\negl[\lambda].$$
and
$$
\Pr_{k\leftarrow \textsf{QSamp}(1^\lambda)}\left[F^{-1}_k(y)=\pi_{k}^{-1}(y)\right] \geq 1-\negl[\lambda].
$$
    \item (\emph{Security}) For any QPT distinguisher $\adv$:
        \begin{align*}
            \left|  \Pr_{{k}\leftarrow \textsf{QSamp}(1^\lambda)} \left[\adv^{F_k,F_k^{-1}}(1^\lambda)=1\right]-\Pr_{O\leftarrow {\Pi}_{n}} \left[\adv^{O,O^{-1}}(1^\lambda)=1\right]\right| \leq \negl[\lambda].
        \end{align*} 
        where $\Pi_n$ is the set of permutations on $\{0,1\}^n$. We say a $\QKPRP$ is \emph{quantum-query-secure} if the above holds even if $\adv$ is given quantum-query-access. Furthermore, in the case where security only holds for $t\le q$ queries for some polynomial $q=q(\lambda)$, then we call this $q$-query $\QKPRP$.
    \end{enumerate}
        \end{definition}

We also define pseudorandom unitaries with quantum input sampling. 

\begin{definition}[Pseudorandom Unitaries with Quantum Input Sampling]
\label{def:PRU-qs}
    Let $m=m(\lambda)$ and $n=n(\lambda)$ be polynomials in the security parameter $\lambda\in \mathbb{N}$. A pair of QPT algorithms $(\textsf{QSamp},U)$ is a \emph{$(m,n)$-pseudorandom unitary with quantum input sampling $(\QKPRU)$} if the following holds:
    \begin{enumerate}
    \item $\textsf{QSamp}(1^\lambda)$: Outputs a $m$-bit key $k$. 
    \item $U_k$: Quantum channel that takes an $m$-bit key $k$ and acts on $n$-qubit states.
        \item For any QPT adversary $\adv$,
        \begin{align*}
            \left| \Pr_{k\leftarrow \textsf{QSamp}(1^\lambda)}\left[\adv^{{U_k}}(1^\lambda)=1\right]-\Pr_{U\leftarrow \mu}\left[\adv^{U}(1^\lambda)=1\right]\right| \leq \negl[\lambda].
        \end{align*}
        where $\mu$ denotes the Haar measure on the unitary group $\mathcal{U}(\mathbb{C}^n)$. If $\adv$ is restricted to only $q=q(\lambda)$ queries to the unitary, then this is denoted as $(q,m,n)$-$\BQPRU$.  
    \end{enumerate}
\end{definition}

Note that our definition of is weaker than earlier definitions of $\PRU$ \cite{JLS18}, as we do not require that the pseudorandom unitary to be a unitary map. We only require that it is indistinguishable from a Haar random unitary. Unfortunately, due to the negligible error inherent in $\QKPRP$ and $\QKPRF$, our construction of a $\QKPRU$ from these primitives is not guaranteed to be a unitary map. 

\subsection{Result}

Note that a $\QKPRG$ with sufficient expansion easily implies a $\QKPRF$ with \emph{polynomial} domain through interpreting the output string as a function. However, it is not clear if a $\QKPRG$ can be used to build full-fledged $\QKPRF$ with \emph{exponential} domain since the standard construction converting a $\PRG$ to a $\PRF$ \cite{GGM86} and its quantum adaption \cite{Z12} both implicitly use the uniform input sampling property of $\PRG$s. Hence, adapting this conversion to the quantum input sampling setting is an interesting open question.  

Fortunately, $\QKPRF$ with polynomial domain can still be useful for applications by converting them to bound-query $\QKPRF$s with exponential domain using \cref{domain-extension}. Specifically, the paper \cite{DMP20} shows how to expand the domain size of a $\textsf{PRF}$, and the same construction and result apply to $\QKPRF$s as well. 

\begin{lemma}[Theorem 7 \cite{DMP20}]
\label{domain-extension}
    Let $\lambda\in \mathbb{N}$ be the security parameter and $q$ and $m$ be polynomials in $\lambda$. Let $(\textsf{QSamp},\textsf{F})$ be a (quantum-query-secure) $\QKPRF$ with key space $\mathcal{K}_q$, domain $\mathcal{X}:\{0,1\}^\ell$ where $\ell=O(\log(\lambda))$, and co-domain $\mathcal{Z}:\{0,1\}^m$. Then, there exists a $q$-query (quantum-query-secure) $\QKPRF$ $(\textsf{QSamp},F_q')$ with the same key sampling algorithm and key space $\mathcal{K}_q$, and with domain and co-domain $\mathcal{Z} $. 
\end{lemma}

We will use this result to build $\BQPRU$ and $\BQLPRS$ from $\QKPRG$s. First, \cite{Z16} shows how to build quantum-query-secure pseudorandom permutations from quantum-query-secure $\PRF$s. This conversion queries the $\PRF$ a polynomial number of times with respect to the security parameter $\lambda$ and input length $n$. Hence, the same proof can be used to show that for any $q'\in \poly[\lambda]$, there exists a $q\in \poly[\lambda,n]$ such that $q$-query pseudorandom functions imply $q'$-query pseudorandom permutations. 

\begin{corollary}
    Let $\lambda\in \mathbb{N}$ be the security parameter and $q=q(\lambda)$ and $n=n(\lambda)$ be polynomials in $\lambda$. There exists a polynomial $\ell=\ell(\lambda)$, such that $(\lambda,\ell)$-$\QKPRG$s imply $(q,\lambda,n)$-$\BQPRP$s.
\end{corollary}

Recently, \cite{MH24} showed how to build a $\PRU$ from $\PRP$s and $\PRF$s. Notably, each unitary evaluation uses a single quantum query to the $\PRP$ and to the $\PRF$. Therefore, we obtain bounded-copy $\QKPRU$s from bounded-query $\QKPRF$s and bounded-query $\QKPRP$s. Furthermore, $\BQPRU$ imply $\BQLPRS$ given that $\LPRS$ can be viewed as a special case of $\PRU$s, where the unitary can only be queried on the state $\ket{0^n}$.

\begin{theorem}
    Let $\lambda\in \mathbb{N}$ be the security parameter and $q$ and $n$ be polynomials in $\lambda$. There exists a polynomial $\ell$ in $\lambda$ such that $(\lambda,\ell)$-$\QKPRG$s imply $(q,\lambda,n)$-$\BQPRU$s and $(q,\lambda,n)$-$\BQLPRS$.
\end{theorem}

\end{document}